\documentclass[acmsmall, review=false]{acmart}

\input{datalabmacros}

\setcopyright{rightsretained}

\acmJournal{PACMMOD}
\acmYear{2024} \acmVolume{2} \acmNumber{1 (SIGMOD)} \acmArticle{61} \acmMonth{2} \acmPrice{}\acmDOI{10.1145/3639316}

\begin{document}

\newcommand{\ourTitle}{On The Reasonable Effectiveness of \diagrams}
\title[\ourTitle]{\ourTitle}
\titlenote{The title is a reference to Wigner's 1960 article~\cite{wigner:1960}
in which he states that ``the enormous usefulness of mathematics in the natural sciences is something bordering on the mysterious and that there is no rational explanation for it.''
While there have been similar observations of surprising effectiveness for both
data~\cite{DBLP:journals/expert/HalevyNP09}
and
logic~\cite{DBLP:journals/bsl/HalpernHIKVV01} in our community,
our strong experimental evidence of \diagrams\ helping users understand
relational patterns better is actually quite expected.
}

\subtitle{Explaining Relational Query Patterns and the Pattern Expressiveness of Relational Languages}

\author{Wolfgang Gatterbauer}
\orcid{0000-0002-9614-0504}
\affiliation{%
    \orcidicon{0000-0002-9614-0504}
	Northeastern University\country{USA}}
\email{w.gatterbauer@northeastern.edu}

\author{Cody Dunne}
\orcid{0000-0002-1609-9776}%
\affiliation{%
    \orcidicon{0000-0002-1609-9776}
	Northeastern University\country{USA}}
\email{c.dunne@northeastern.edu}

\begin{abstract}
Comparing relational languages by their logical expressiveness is well understood.
Less well understood is how to compare relational languages by their ability to represent \emph{relational query patterns}.
Indeed, what are query patterns other than ``a certain way of writing a query''?
And how can query patterns be defined across procedural and declarative languages, irrespective of their syntax?
To the best of our knowledge, we provide the first semantic definition of relational query patterns
by using a variant of structure-preserving mappings between the relational tables of queries.
This formalism allows us to analyze the \emph{relative pattern expressiveness} of relational language fragments and create a hierarchy of languages with equal logical expressiveness yet different pattern expressiveness.
Notably, for the non-disjunctive language fragment,
we show that relational calculus can express a larger class of patterns
than the basic operators of relational algebra.

Our language-independent definition of query patterns opens novel paths for assisting database users. 
For example, these patterns could be leveraged to create visual query representations that faithfully represent query patterns, speed up interpretation, and provide visual feedback during query editing. 
As a concrete example, we propose \diagrams, a complete and sound diagrammatic representation of safe relational calculus that is provably
($i$) unambiguous,
($ii$) relationally complete, 
and ($iii$) able to represent all query patterns for 
unions of non-disjunctive queries.
Among all diagrammatic representations for relational queries that we are aware of, 
ours is the only one with these three properties.
Furthermore, our anonymously preregistered user study shows that \diagrams\ allow users to recognize patterns meaningfully faster and more accurately than \SQL.

\end{abstract}

\maketitle

\setcounter{page}{1}

\section{Introduction}
\label{SEC:INTRODUCTION}

When designing and comparing query languages, we are usually concerned with \emph{logical expressiveness}: 
can a language express a particular query we want?
For relational languages, questions of expressiveness 
have been studied for decades, and 
formalisms for comparing expressiveness are well-developed and understood.

We do not yet have a similarly developed machinery to reason about 
\emph{relational query patterns} across languages.
Intuitively, a query pattern should capture ``a certain way of writing a query.''
To be universally applicable,
a formalization 
would have to be applicable across 
the four major languages---Datalog, Relational Algebra ($\RA$), Relational Calculus ($\RC$), and SQL---and 
thus be orthogonal to questions of syntax
and procedural or declarative language design.

\begin{wrapfigure}{r}{0.45\textwidth}
	\vspace*{-4mm}
	\centering
	\includegraphics[scale=0.37]{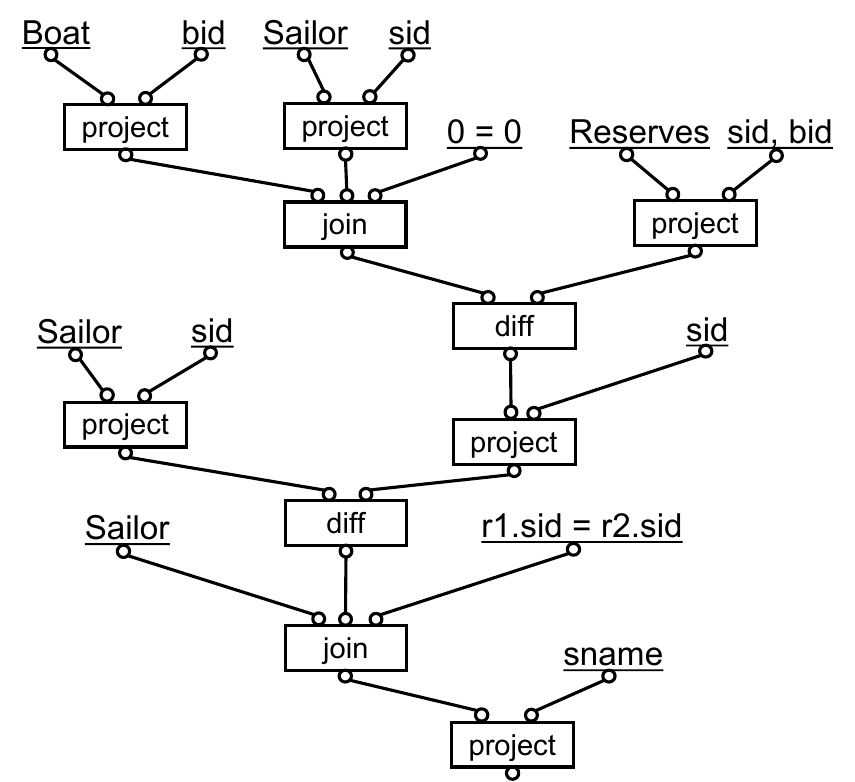}
	\caption{DFQL~\cite{DBLP:journals/vlc/CatarciCLB97} visualization of the $\TRC$
	query from \cref{ex:intro_allboats}. Notice the 3 instances of the Sailor relation
	and thus a different ``structure'' of the visualization from the original query.}
	\label{Fig_DFQL_all_boats}
\end{wrapfigure}

We posit that identifying patterns in queries could open novel paths for assisting users \cite{Gatterbauer2022PrinciplesQueryVisualization}, especially learners
who try to understand the structure behind relational queries written in different languages.
It could help learners spot similarities in queries across different schemas 
and thus more easily separate intent (the logic) from the particular syntactic expression.
On an even more fundamental level, establishing a separation on ``pattern expressiveness''
between relational languages could lead to new insights into the intrinsic properties of relational languages
and algebraic limits of visualizations.
An important insight that we establish in this paper is that visual languages which build upon 
\emph{the operators of $\RA$ cannot faithfully express all query patterns}, 
and instead necessitate reformulating queries and thus changing their patterns.

\begin{example}[Understanding the structure of a $\TRC$ query]
\label{ex:intro_allboats}
Imagine Kiyana, a theory-leaning  undergraduate student, trying to understand relational query languages better.
Kiyana has been reading the chapters on relational calculus across several books. 
In the textbook by Ramakrishnan and Gehrke~\cite{cowbook:2002} (page~121 of Sect.~4.3.1)
she finds the query
\emph{``(Q9) Find the names of sailors who have reserved all boats''}
written as follows:
\begin{align}
\begin{aligned}
	&\{P \mid \exists S \in \mathit{Sailor}\, 
		\forall B \in \mathit{Boat} (\exists R \in \mathit{Reserves}  \\
	&(S.\mathit{sid}=R.\mathit{sid} \wedge
		  R.\mathit{bid} = B.\mathit{bid}  \wedge
		  P.\mathit{sname} = S.\mathit{sname} ) ) \}
\end{aligned}
\label{trc:textbookquery1}
\end{align}
She tries to understand ``the structure'' of the query and 
translates it first into $\RA$, and then from there into 
DFQL (Dataflow Query Language)~\cite{DBLP:journals/iam/ClarkW94,DBLP:journals/vlc/CatarciCLB97,Girsang:DFQL}.\footnotemark
DFQL is a visual representation that is relationally complete 
by mapping its visual symbols to the operators of $\RA$.
Kiyana quickly notices that she cannot translate the query into $\RA$
without using \emph{additional} Sailor relations.
	\begin{align*}
	\begin{aligned}	
		\mathit{Q} =
		\pi_{\sql{sname}} \big(
			\sql{Sailor}\Join \big(
				\pi_{\sql{sid}} \sql{Sailor} - \pi_{\sql{sid}}\big(
				&(\pi_{\sql{sid}} \sql{Sailor} \times \pi_{bid}\mathit{Boat}) \\
			& - \pi_{\sql{sid,bid}} \sql{Reserves}\big) \big) \big)
	\end{aligned}
	\end{align*}
As a result, she does not find the resulting DFQL visualization (\cref{Fig_DFQL_all_boats}) very helpful 
because there is an obvious mismatch in ``its structure'' with 3 instances of Sailor relations.
She wonders whether she is missing an obvious simpler translation into $\RA$
or whether there is none.
As is, she does not find the query visualization helpful.
\end{example}
\footnotetext{
DFQL is one of several visual query languages mentioned as relationally complete in an influential survey~\cite{DBLP:journals/iam/ClarkW94}. 
Kiyana found a detailed online documentation~\cite{Girsang:DFQL}.
}

As a consequence, \emph{no query visualization that relies on the operators of \RA}
could help Kiyana with what she would like to see:
a simple visual representation that captures the structure of the query as written
in its original logical form.

\begin{example}[Comparing $\RC$ queries from textbooks]
\label{ex:intro_comparing_textbooks}
Kiyana continues looking through different textbooks and finds in Date's textbook~\cite{date2004introduction} (page~224 of Sect.~8.3) the query
\emph{``8.3.6 Get supplier names for suppliers who supply all parts''} written as follows:
\begin{align}
\begin{aligned}
	\texttt{SX.NAME } 	& \texttt{WHERE NOT EXISTS PX (NOT EXISTS SPX} \\[-1mm]
						& \texttt{SPX.SNO = SX.SNO AND SPX.PNO = PX.PNO))} 		
\end{aligned}
\label{trc:textbookquery2B}
\end{align}
From the natural language description, the query seems to follow 
a similar pattern as the earlier one
(``Return X which have a relationship with all Y'').
But that apparent similarity is difficult to see from the two expressions.
She wonders whether there is a simple way to see that those two queries somehow follow a ``similar structure.''
\end{example}

In this paper, we show that there is indeed a  simple and arguably-natural visualization that allows Kiyana
to 
($i$) represent her queries in a way that preserves their logical structure (or pattern), 
($ii$) decide whether two logically-equivalent queries have the same pattern,
and 
($iii$) see whether any two queries, even across different schemas, use a ``similar pattern.''
We call this visualization \diagrams~\cite{relationaldiagrams}.
See \cref{Fig_intro_textbook} and notice how 
every relation from the two queries maps to exactly one relation in the \diagrams.
Also, notice how the similar structure of both queries becomes natural to see.

\introparagraph{Our 1$^{\mathrm{st}}$ contribution: \emph{query patterns}} 
We develop a precise language-independent notion of relational query patterns
that allows us to compare 
the patterns of two queries.
Our definition is semantic (in 
the sense
that the definition involves relations over sets of attributes)
instead of syntactic (which would involve structural properties which are inherently language-dependent).
The intuition behind our formalism 
is to reason about mappings between the 
(existentially or universally) quantified
relations referenced in two queries.
Yet it is not trivial to turn this intuition into a working definition that 
can be applied to any relational
query and language 
(we include examples to show that seemingly easier mapping definitions would fail on queries). 
We 
believe
that our notion is the ``right'' definition
and show how to use it
to compare
relational query languages by their abilities to express query patterns 
present in other languages and thus compare their relative pattern expressiveness.
In particular, we contribute a novel hierarchy of pattern-expressiveness among the 
non-disjunctive fragments of four relational query languages.

\begin{figure}[t]
    \centering
    \begin{subfigure}[b]{.47\linewidth}
		\centering
        \includegraphics[scale=0.4]{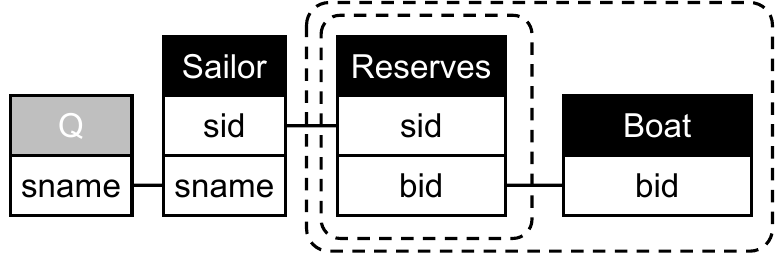}
    	\vspace{-1mm}
        \caption{$q_1$: ``Find names of sailors who reserved all boats.''}
    	\label{Fig_intro_textbook1}
    \end{subfigure}
	\hspace{3mm}
    \begin{subfigure}[b]{.49\linewidth}
		\centering		
    	\vspace{1mm}		
        \includegraphics[scale=0.4]{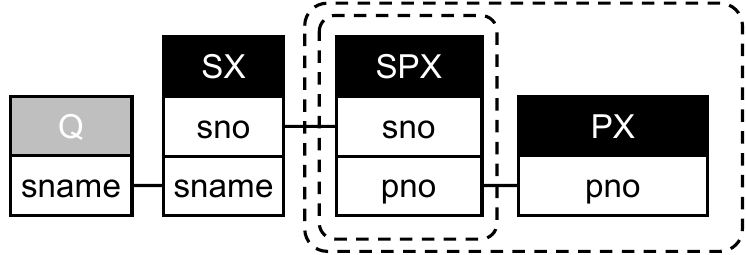}
    	\vspace{-1mm}
        \caption{$q_2$: ``Find names of suppliers who supply all parts.''}
    	\label{Fig_intro_textbook2}
    \end{subfigure}
    \caption{\diagrams\ representations of the two queries 
	from
	\cref{ex:intro_allboats} (\cite{cowbook:2002})
	and 
	\cref{ex:intro_comparing_textbooks} (\cite{date2004introduction}).
	Notice the similar ``relational query patterns.''}
    \label{Fig_intro_textbook}
\end{figure}

\introparagraph{Our 2$^{\mathrm{nd}}$ contribution: \diagrams} 
We formalize an arguably simple and intuitive diagrammatic representation of relational queries called \diagrams~\cite{relationaldiagrams}
and prove that ($i$) it is unambiguous (every diagram has a unique logical interpretation),
($ii$) it is relationally complete (every relational query can be expressed in a logically-equivalent \diagram),
and ($iii$) that it can express all query patterns in the non-disjunctive fragment of relational query languages
and those with union at the root.
In particular, we prove that no prior or future diagrammatic representation based on $\RA$ could represent 
all relational query patterns from $\RC$.
Our user study (\cref{SEC:USERSTUDY}) shows that our formalisms helps users recognize patterns faster than with $\SQL$.

\textbf{Outline of the paper.}
\Cref{sec:nondisjunctivefragment} defines the \emph{non-disjunctive fragment} of relational query languages for 
Datalog, Relational Algebra ($\RA$), Tuple Relational Calculus (\TRC), and $\SQL$,
and proves that they have equivalent logical expressiveness.
\Cref{sec:QV} 
shows that the non-disjunctive fragment allows for an arguably natural 
diagrammatic representation system that we term \diagramsNDnomath.
We give the formal translation from the non-disjunctive fragment of $\TRC$ to \diagramsNDnomath\ and back, 
and define their formal validity.
We use \diagramsNDnomath\ for the remainder of the paper to illustrate ``query patterns.''
\Cref{sec:structureisomorphism}
formalizes the notion of a relational query pattern
and contributes a novel hierarchy of pattern expressiveness among the above four languages
for the non-disjunctive fragment.
We prove that \diagrams\ 
have strong structure-preserving properties in that they 
can express all query patterns in this fragment.
\Cref{sec:pattern_schemas}
formalizes ``similar patterns'' across different schemas.
This extended notion allows us to see similarities
across queries that use different relations and are thus not logically equivalent.
\Cref{sec:completeness}
adds a single visual element (called a \emph{union cell}) to \diagramsNDnomath\ 
to make the resulting \diagrams\ relationally complete.\footnote{Although disjunctions can be composed of conjunction and negation using De Morgan's law ($A \vee B = \neg(\neg A \wedge \neg B)$),
this additional visual symbol is necessary: for safe relational queries, DeMorgan is not enough, as
there is no way to write a safe Tuple Relational Calculus ($\TRC$) expression ``\emph{Return all entries that appear in either R or S}'' 
that avoids a union operator.
This is part of the textbook argument for the union operator  being an essential, non-redundant operator for relational algebra.
}
We also show that even without that element, \diagramsNDnomath
can express all logical statements of first-order logic.
This allows us to compare our diagrammatic formalism against a long history of diagrams for 
representing logical sentences.
\Cref{sec:studies} includes two studies.
The first shows that more logical queries across five popular textbooks have
pattern-isomorphic representations in \diagrams\ than either \RA, \Datalog, \QBE, or \queryviz.
The second controlled user experiment demonstrates that using \diagrams\ instead of $\SQL$ helps users recognize patterns across different schemas faster and more often correctly.
\Cref{sec:relatedWork}
contrasts our formalism with selected related work. 
In particular, we discuss the connection to Peirce's existential graphs \cite{peirce:1933,Roberts:1992,Shin:2002} and 
show that our formalism is more general and solves interpretational problems of Peirce's graphs, which have been the focus of intense research for over a century.

Due to space constraints, we had to move proofs, several intuitive illustrating examples, all study details, and more detailed comparison against related work
\iflabelexists{appendix:beginning}
{to \hyperref[appendix:beginning]{the appendix}.}
{to the online appendix~\cite{relationaldiagrams:links}.}

\section{The non-disjunctive fragment of relational query languages}
\label{sec:nondisjunctivefragment}
\label{SEC:NONDISJUNCTIVEFRAGMENT}%

This section defines 
the \emph{non-disjunctive fragment} of 
relational query languages.
Throughout, we assume a linear order over the active domain 
and thus explicitly allow built-in predicates using ordered operators such as $<$,
in addition to equality $=$ and disequality $\neq$.

We assume the reader to be familiar with 
$\DatalogN$ (non-recursive Datalog with negation),
$\RA$ (Relational Algebra),
$\TRC$ (safe Tuple Relational Calculus),
$\SQL$ (Structured Query Language),
and the necessary safety conditions for $\TRC$ and $\DatalogN$ to be equivalent in logical expressiveness to $\RA$.
We also assume familiarity with concepts such as relations, predicates, atoms, and the named and unnamed perspective of relational algebra.
The most comprehensive exposition of these topics we know of is Ullman's 1988 textbook \cite{Ullman1988PrinceplesOfDatabase},
together with resources for translating between $\SQL$ and relational calculus 
\cite{DBLP:journals/tse/CeriG85, VanDenBusscheVansummeren:2009:SQL}.
These connections are also discussed in most database textbooks 
\cite{Elmasri:dq,cowbook:2002, Garcia-MolinaUW2009:DBSystems, DBLP:books/mg/SKS20},
though in less detail.
We only cover TRC and not Domain Relational Calculus (DRC) as 
the 1-to-1 correspondence between DRC and TRC is straight-forward \cite{Elmasri:dq, DBLP:books/mg/SKS20},
and---as we will discuss in \cref{sec:Peirce}---$\TRC$
has a more natural translation into diagrams than $\DRC$.

\subsection{Non-recursive Datalog with negation}
\label{sec:datalog}

We start with $\Datalog$ since the definition is most straightforward.
$\Datalog$ expresses disjunction (or union) by repeating an Intensional Database (IDB) predicate in the head of multiple rules.
For example, %
consider the following query in $\Datalog$:
\begin{align}
\begin{aligned}
	Q(x) &\datarule R(x,y), S(x), T(\_), y>5.	\\
	Q(x) &\datarule R(x,y), S(\_), T(x), y>5.
\end{aligned}
\label{datalog:disjunctive_example}	
\end{align}	
The underscore stands for a variable that appears only once~\cite{Garcia-MolinaUW2009:DBSystems}.
This query cannot be expressed without defining at least one IDB at least twice,
in our case the result table $Q(x)$.
This leads to a natural definition of the non-disjunctive fragment of $\DatalogN$:

\begin{definition}[$\DatalogND$]
\label{def:DatalogND}
	\emph{Non-disjunctive non-recursive Datalog with negation} ($\DatalogND$)
	is the non-recursive fragment of $\DatalogN$ with built-in predicates 
	where every IDB appears in the head of exactly one rule
	and can be used maximally once in any body.
\end{definition}

Notice that $\DatalogND$ inherits all restrictions from non-recursive $\DatalogN$ with built-in predicates~\cite{Ullman1988PrinceplesOfDatabase}, 
and thus
rules out the existence of an IDB in both the head and the body of the same rule.
The restriction of IDB's being used maximally once 
rules out views to be used multiple times (including simple copies of input tables).

\subsection{Relational Algebra (RA)}
\label{sec:relationalAlgebra}

We focus on the subfragment of basic $\RA$ 
($\times, \sigma, \Join_c, \pi, -$)
that contains no union operator $\cup$ 
and in which all selection conditions are simple (i.e.\ they do not use the disjunction operator $\vee$).
A simple condition is $C=(X \theta Y)$ where $X$ is an attribute, 
$Y$ is either an attribute or a constant, 
and $\theta$ is a comparison operator from $\{=, \neq, <, \leq, >, \geq, \}$.
Notice that conjunctions of selections can be modeled as concatenation of selections,
e.g., $\sigma_{C_1 \wedge C_2}(R)$ is the same as $\sigma_{C_1}(\sigma_{C_2}(R))$.
The $\DatalogN$ query from
\cref{datalog:disjunctive_example}	
cannot be expressed in that fragment and requires 
either 
the union operator $\cup$ as in:
\begin{align*}
	\pi_{A} \big(\sigma_{B>5}(R) \Join S \times \rho_{A \rightarrow C}(T)\big) 
	\cup
	\pi_{A} \big(\sigma_{B>5}(R) \Join T \times \rho_{A \rightarrow D}(S)\big) 
\end{align*}
or the disjunction operator $\vee$ as in:
\begin{align*}
	\pi_{A} \big(\sigma_{A=D \vee A=C} \big(\sigma_{B>5}(R) \times \rho_{A\rightarrow D} (S) \times \rho_{A \rightarrow C}(T)\big) \big)
\end{align*}

\begin{definition}
	[$\NDRA$]
	The non-disjunctive fragment of basic
	Relational Algebra ($\NDRA$) results from disallowing the union operators $\cup$ 
	and by restricting selections to conjunctions of simple predicates.
\end{definition}

\subsection{Tuple Relational Calculus (TRC)}
\label{sec:TRC}
Recall that safe $\TRC$ only allows existential quantification (and not universal quantification)
\cite{Ullman1988PrinceplesOfDatabase}.
Predicates are either join predicates 
``$r.A \,\theta\, s.B$'' or selection predicates 
``$r.A \,\theta\, v$'', with $r, s$ being table variables and $v$ a domain value.
WLOG, every existential quantifier can be pulled out as early as to either be at the start of the query, or directly following a negation operator. 
For example, instead of
$\h{\neg(\exists r \in R} [r.A = 0 \wedge \h{\exists s \in S}[s.B = r.B]])$
we rather write this sentence canonically as
$\h{\neg(\exists r \in R, s \in S} [r.A = 0 \wedge s.B = r.B])$.
This \emph{canonical representation} implies that a set of existential quantifiers is always preceded 
by the negation operator, 
except for the table variables outside any scope of negation operators.
Also, WLOG, we only allow equality conditions with the result table.
For example, instead of
$\{ q(A) \mid \exists r \in R, s \in S [q.A = r.A \wedge s.A \,\h{> q.A}])] \}$
we rather write
$\{ q(A) \mid \exists r \in R, s \in S [q.A = r.A \wedge s.A \,\h{> r.A}])] \}$.
Recall that at least one equality predicate for each output attribute is required due to standard safety conditions 
\cite{Ullman1988PrinceplesOfDatabase}.

We will define an additional requirement that each predicate contains 
a local (or what we refer to as \emph{guarded}) 
attribute
whose table is quantified within the scope of the last negation.
For example, we do not allow 
$\neg(\exists r \in R[\h{\neg(r.A=0)}])$ 
because the table variable $r$ is defined outside the scope of the most inner negation around the predicate $r.A=0$.
However, we allow the logically-equivalent 
$\neg(\exists r \in R[\h{r.A \neq 0}])$ where the table variable $r$ is existentially quantified within the same scope as the attribute $r.A \neq 0$.

\begin{definition}[Guarded predicate]\label{def:anchor}
	A predicate is \emph{guarded} if it contains at least one attribute of a table that is existentially quantified 
	inside the same negation scope as that predicate.
\end{definition}

Intuitively, guarding a predicate guarantees that the predicates can be applied in the same logical scope
where a table is defined.
This requirement also avoids a hidden disjunction.
To illustrate, consider the following $\TRC$ query:
\begin{align*}
	\{ q(A) \mid \h{\exists r \in R} [q.A = r.A \wedge 
	\h{\neg (}\exists s \in S[\h{r.A = 0} \wedge s.B = r.B])] \}	
\end{align*}
This query contains no apparent disjunction, however the predicate ``$r.A = 0$''
could be pulled outside the negation,
and after applying De Morgan's law on the expression we get a disjunction:
\begin{align*}
	\{ q(A) \mid \h{\exists r \in R} [q.A = r.A \wedge 
	\big(
	\h{r.A \neq 0 
	\,\vee}\,
	\neg (\exists s \in S[s.B = r.B])
	\big) ]\}		
\end{align*}

To avoid both disjunctions and ``hidden disjunctions'', the non-disjunctive fragment 
\emph{only allows conjunctions of guarded predicates}:

\begin{definition}
	[$\NDTRC$]	
	\label{def:NDTRC}
	The non-disjunctive fragment of safe $\TRC$ ($\NDTRC$)
	restricts predicates to conjunctions of \emph{guarded predicates}.
\end{definition}

In order to express 
the $\DatalogN$ query from \cref{datalog:disjunctive_example}	
we need the disjunction operator.
A possible translation is:
\begin{align*}
	\{ q(A) \mid \exists r \in R,  s \in S,  t \in T 
		[q.A \!=\! r.A \wedge r.B \!>\! 5 \wedge 
		 (r.A \!=\! s.A \vee r.A \!=\! t.A)]\}		
\end{align*}

\subsection{SQL under set semantics}

Structured Query Language (SQL) uses bag instead of set semantics and uses a ternary logic with NULL values.
In order to treat $\SQL$ as a logical query language,
we assume binary logic and no NULL values in the input database. 
It has been pointed out that 
``SQL's logic of nulls confuses people'' and even programmers tend to think in terms of the familiar two-valued logic~\cite{DBLP:journals/pvldb/ToussaintGLS22}.
Our focus here is devising a general formalism to capture logical query patterns across relational languages,
not on devising a visual representation of SQL's idiosyncrasies.
To emphasize the set semantic interpretation, we write the DISTINCT operator in all our SQL statements.

We define the non-disjunctive fragment of $\SQL$ as the Extended Backus–Naur form (EBNF)~\cite{pattis:EBNF} grammar
shown in \cref{table:supported_grammar},
interpreted under set semantics (no duplicates by using \sql{DISTINCT}) 
and under binary logic (no null values allowed in the input tables).
We also require the same syntactic restriction as for $\NDTRC$: 
\emph{every predicate needs to be guarded} (\cref{def:anchor}), 
i.e., every predicate must reference at least one table within the scope of the last NOT.
This restriction excludes hidden disjunctions,
such as
``NOT(NOT(P1) and NOT(P2))'' which is equivalent to 
``P1 or P2''.

\begin{figure}[t]
\centering
\renewcommand{\arraystretch}{0.85}
	\textup{
	\textsf{
	\footnotesize
	\setlength{\tabcolsep}{0.6mm}
	\begin{tabular}{@{} r @{\hspace{2mm}} l @{\hspace{4mm}} l @{\hspace{6mm}}l l}
		Q::=    	& SELECT [DISTINCT] (C \{, C\} $\mid \ ^{*}$)   & Non-Boolean query\\
		  	  		& FROM 		R \{, R\} 								\\
		  	  		& [WHERE  	P]  									\\
		$\mid\,$  	& SELECT NOT (P)      					& Boolean query\\
		$\mid\,$  	& SELECT [NOT] EXISTS (Q)				& Boolean query\\[0.5mm]
		C::=    	& [T.]A     							& column or attribute	\\[0.5mm]
		R::=    	& T [[AS] T]  							& table (table alias)	\\[0.5mm]
		P::=    	& P \{AND P\}     						& conjunction of predicates\\
		$\mid\,$  	& C O C      							& \hspace{3mm}join predicate\\
		$\mid\,$  	& C O V     							& \hspace{3mm}selection predicate\\
		$\mid\,$  	& NOT `('P`)'			 				& \hspace{3mm}negation\\
		$\mid\,$  	& [NOT] EXISTS `('Q`)' 					& \hspace{3mm}existential subquery\\
		$\mid\,$  	& C [NOT] IN `('Q`)'  					& \hspace{3mm}membership subquery\\
		$\mid\,$  	& C O (ALL `('Q`)' $\mid$ ANY `('Q`)')  & \hspace{3mm}quantified subquery\\[0.5mm]		
		O::=    	& $= \ \mid \ <> \ \mid \ < \ \mid \ \leq \ \mid \ \geq \ \mid \ > $ 
															& comparison operator\\[0.5mm]
		T::=    	&           					& table identifier\\[0.5mm]
		A::=    	&           					& attribute identifier\\[0.5mm]
		V::=    	&           					& string or number
	\end{tabular}
	}
	}
\vspace{1mm}
\caption{EBNF Grammar of $\NDSQL$:
Statements enclosed in [\hspace{1mm}] are optional; 
statements separated by \textbf{$\mid$} indicate 
a choice between alternatives;
parentheses without quotation marks (\hspace{1mm}) group alternative choices;
parentheses with quotation marks `('\hspace{1mm}`)' form part of the test.
Additionally, the main query requires the DISTINCT keyword (if non-Boolean), 
and all join and selection predicates need to be \emph{guarded} (\cref{def:anchor}), 
i.e., reference at least one table within the scope of the last NOT.
}
\label{table:supported_grammar}  
\end{figure}

\begin{definition}
	\emph{$\NDSQL$}:
	Non-disjunctive SQL under set semantics ($\NDSQL$) is the syntactic restriction of $\SQL$ 
	under binary logic (no NULL values in the input tables)
	to the grammar defined in \cref{table:supported_grammar}, and additionally
	requiring every predicate to be guarded.
\end{definition}

Every $\NDSQL$ query can be brought into a canonical form that maintains a straightforward one-to-one correspondence with $\NDTRC$.
The idea is to replace membership and quantified subqueries with existential subqueries 
(see grammar in \cref{table:supported_grammar})
and then unnest any existential quantifiers, i.e., to only use ``\sql{not exists}''.
This pulling up quantification as early as possible is identical to the way we defined the canonical form of $\NDTRC$.

The $\DatalogN$ query from
\cref{datalog:disjunctive_example}	
cannot be expressed in $\NDSQL$  and requires either a UNION operator or disjunction as in
\cref{fig:SQL_example_disjunction}.

\begin{figure}[t]
\centering
\begin{minipage}{0.255\linewidth}
\begin{lstlisting}
SELECT DISTINCT R.A
FROM R, S, T
WHERE R.B > 5 
AND (R.A = S.A OR R.A = T.A)
\end{lstlisting}
\end{minipage}
\vspace{-5mm}
\caption{Example $\SQL$ with disjunction.}
\label{fig:SQL_example_disjunction}
\vspace{-1mm}
\end{figure}

\subsection{Logical expressiveness of the fragment}

We show that the four languages restricted to the non-disjunctive fragment are equivalent in their logical expressiveness.
The proof is available 
\iflabelexists{appendix:proofoflogicalexpressivness}
{in \cref{appendix:proofoflogicalexpressivness}}
{in the optional online appendix~\cite{relationaldiagrams:links}}
and
is an adaptation of the standard proofs of equal expressiveness as found, for example, in
\cite{Ullman1988PrinceplesOfDatabase}.
However, the translations also 
need to pay attention to the restricted fragment (e.g.\ we cannot use union to define an active domain)
and attempt to keep the numbers of extensional database atoms the same, if possible.
This detail will be important later in \cref{sec:structureisomorphism}, where we show that those four fragments differ in the types of query patterns they can express.

\begin{theorem}\label{th:equivalence}\label{TH:EQUIVALENCE}[Logical expressiveness]
	$\DatalogND$, 
	$\NDRA$, 
	$\NDTRC$, and $\NDSQL$ 
	have the same logical expressiveness.
\end{theorem}

\section{\diagramsNDnomath}
\label{SEC:QV}
\label{sec:QV}

This section introduces our diagrammatic representation of relational queries.
It details the basic visual elements of \diagramsNDnomath
(\cref{sec:visualelements}),
gives the formal translation 
from $\NDTRC$ 
(\cref{sec:fromTRCtoRD})
and back (\cref{sec:fromRDtoTRC}),
and shows that there is a one-to-one correspondence between $\NDTRC$ expressions and \diagramsNDnomath,
thereby proving their validity (\cref{sec:RD_validity}).

\subsection[Visual elements]{Visual elements}
\label{sec:visualelements}

In designing our diagrammatic representation, 
we started from existing widely-used visual metaphors and then added the minimum necessary visual elements to 
obtain expressiveness for full $\NDTRC$.
In the following five points, we discuss both 
($i$) necessary specifications for \diagramsNDnomath\ and
($ii$) concrete design choices that are not formally required but justified based on 
best practices from HCI and visualization guidelines.
We use the term \emph{canvas} to refer to the plane in which a \diagramNDnomath\ is displayed.
We illustrate with \cref{Fig_TRC_vs_RD}.

\emph{(1) Tables and attributes}:
We use the set-of-mappings definition of relations~\cite{Ullman1988PrinceplesOfDatabase}
in which a tuple is a mapping from attributes' names to values, in contrast to the set-of-lists representation in which order of presentation matters and which more closely matches the typical vector representation.
Thus a table is represented by any visual grouping of its attributes.
We use the typical UML convention of representing tables as rectangular boxes
with table names on top and attribute names below in separate rows. 
Table names are shown with white text on a black background and, to differentiate them, attributes use black text on a white background.
For example, table \tableBox{R} with attribute \attributeBox{A}.
Similar to $\Datalog$ and $\RA$ (and different from $\SQL$ and $\TRC$),
we do not use table aliases. Such table aliases create extra cognitive burden
and are only needed in languages where references to repeated table instances cannot be otherwise disambiguated.
We also reduce visual complexity by only showing attributes that are used in the query, similar to $\SQL$ and $\TRC$ (and different from $\Datalog$).
Database users are commonly familiar with relational schema diagrams.
Thus, we believe that a simple conjunctive query should be visualized similarly to
a typical database schema representation,
as used prominently in standard introductory database textbooks \cite{Elmasri:dq,DBLP:books/mg/SKS20}.

\emph{(2) Selection predicates}:
Selection predicates are filters and are shown ``in place.'' 
For example, an attribute ``$r_2.C > 1$'' is shown as \selPredicateBox{C$>$1}
in the corresponding instance of table $R$.
An attribute participating in multiple selection predicates is repeated at least as many times as there are selections
(e.g.\ to display ``$r_2.C > 1 \wedge r_2.C < 3$'', we would repeat $\sql{R.C}$ twice as \selPredicateBox{C$>$1} and \selPredicateBox{C$<$3}).
An attribute participating in $k$ selection predicates is repeated $k$ times.

\emph{(3) Join predicates}:
Equi-join predicates (e.g.\ ``$s_2.A = t_2.A$''), which arguably are the most common type of join in practice, are represented
by lines connecting joined attributes.
For other less-frequent theta join operators 
$\{\neq, <, \leq, \geq, >\}$, 
we add the operator as a label on the line
and use an arrowhead to indicate the reading order and correct application of the operator \emph{in the direction of the arrow}.
For example, for a predicate
``$r_1.A > r_2.B$'', the label is $>$ and the arrow points from 
attribute $\sql{A}$ of the first $\sql{R}$ occurrence to $\sql{B}$ of the second:
$\sql{A} {\scriptscriptstyle \xrightarrow{>}} \sql{B}$.
Notice that the direction of arrows can be flipped, along with flipping the operator, while maintaining the identical meaning:
$\sql{A} {\scriptscriptstyle \xleftarrow{<}} \sql{B}$. 
To avoid ambiguity with the standard left-to-right reading convention for operators, we normalize arrows to never point from right to left.
An attribute participating in multiple join predicates needs to be shown only once and has several lines connecting it to other attributes.
An attribute participating in one or more join predicates and $k$ selection predicates is shown $k+1$ times.\footnote{In practice, one can reduce the size of a \diagramNDnomath\ by reusing an existing selection predicate  for joins. This comes at the conceptual complication that the exact graph topology of the \diagramNDnomath\ (which attributes are connected) is not uniquely determined (though it still allows only one correct interpretation).
In our example \cref{Fig_TRC_vs_RD_b}, one could remove the attribute \sql{R.C} of $r_2$ and connect \sql{Q.D} to either \selPredicateBox{C$>$1} or \selPredicateBox{C$<$3} instead.}

\emph{(4) Negation boxes}:
In $\NDTRC$, negations are either avoided (e.g.\ $\neg (R.A = S.B)$ is identical to $R.A \neq S.B$) or 
placed before the existential quantifiers.
We represent a negation with a closed line that partitions the canvas into a subcanvas that is negated (inside the bounding box)
and everything else that is not (outside of the bounding box).
As a convention, we use dashed rounded rectangles.\footnote{Rectangles allow better use of space than ellipses, and rounded corners together with dashed lines distinguish those negation boxes clearly from the rectangles with solid edges and right angles used for tables and attributes.}
Recursive partitioning of the canvas allows us to represent a tree-based nesting order that corresponds to
the nested scopes of quantified tuple variables in $\TRC$
(and also the nesting order of subqueries in $\SQL$).
We call \emph{the main canvas} the root of that nesting hierarchy and each node a \emph{partition of the canvas}.

\emph{(5) Output table}:
We display an output table to emphasize the compositional nature of relational queries: a relational query uses several tables as input, and returns one new table as output. 
We use the same symbol for that output table as the $\TRC$ expression, for which we most commonly use $Q$.
We use a gray background \selectBox{} to make this table visually distinct from input tables.

\subsection{\texorpdfstring{From $\NDTRC$ to \diagramsNDnomath}{From \NDTRC to \diagramsNDnomath}}
\label{sec:fromTRCtoRD}

We next describe the 5-step translation from any valid $\NDTRC$ expression to a \diagramsNDnomath.
We illustrate by translating a $\NDTRC$ expression (\cref{Fig_TRC_vs_RD_a}) into a \diagramsNDnomath\ (\cref{Fig_TRC_vs_RD_b}).
Notice that the translation critically leverages three conditions fulfilled by the input:
(1) Safe $\TRC$ (and thus also $\NDTRC$) only allows existential and not universal quantification \cite{Ullman1988PrinceplesOfDatabase},
(2) $\NDTRC$ only allows conjunction between predicates, and
(3) all predicates in $\NDTRC$ are guarded (recall \cref{def:anchor}).

\emph{(1) Creating canvas partitions}:
The scopes of the negations in a $\TRC$ are nested by definition.
We translate this hierarchy of the scopes for each negation (the \textit{negation hierarchy}) into a nested partition of the canvas.
\cref{Fig_TRC_vs_RD_d} illustrates the nested partitions as derived from the 
negation hierarchy \cref{Fig_TRC_vs_RD_c} of the original $\NDTRC$ expression.
Notice that the double negation ``$\neg(\neg(\ldots))$'' results in the scope $q_1$ of the negation hierarchy to be empty.

\emph{(2) Placing tables}:
Each table variable defines a table that gets placed into the canvas partition 
that corresponds to the respective negation scope.
For example, the tables corresponding to the table variables $r_1$, $r_2$, and $s_1$ are outside any negation scope and thus placed in the root partition $q_0$.
Similar to \Datalog\ and \RA\ (and in contrast to $\TRC$ and $\SQL$) \diagramsNDnomath\ do not need table aliases.

\begin{figure}[t]
\centering
\begin{subfigure}[b]{0.52\linewidth}
	\centering	
    \includegraphics[scale=0.42]{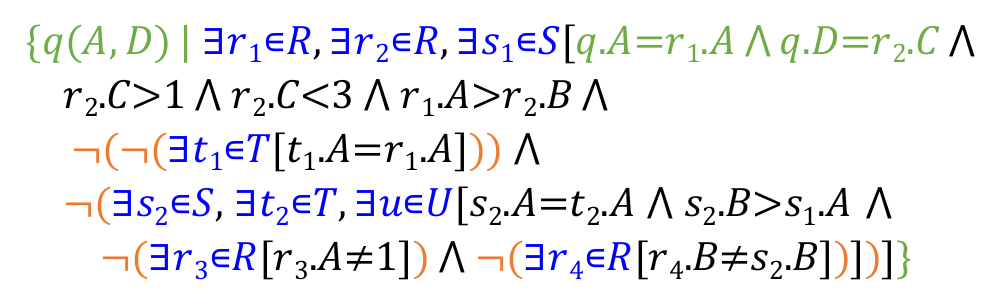}
    \caption{$\NDTRC$}
    \label{Fig_TRC_vs_RD_a}
\end{subfigure}
\hspace{-1mm}
\begin{subfigure}[b]{0.22\linewidth}
	\centering
    \includegraphics[scale=0.42]{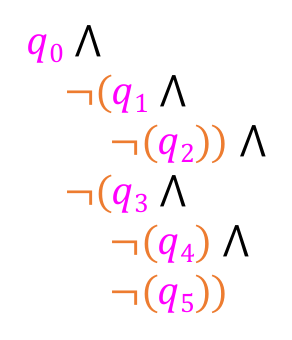}
    \caption{Negation hierarchy}
    \label{Fig_TRC_vs_RD_c}
\end{subfigure}
\hspace{2mm}
\begin{subfigure}[b]{0.23\linewidth}
	\centering
    \includegraphics[scale=0.42]{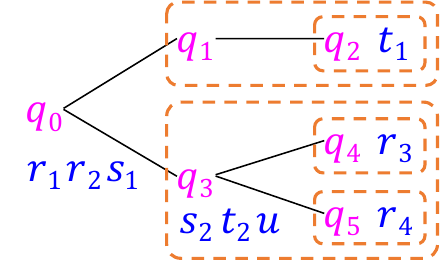}
    \caption{Canvas partitions}
    \label{Fig_TRC_vs_RD_d}
\end{subfigure}
\begin{subfigure}[b]{0.52\linewidth}
	\centering	
	\vspace{1mm}	
    \includegraphics[scale=0.4]{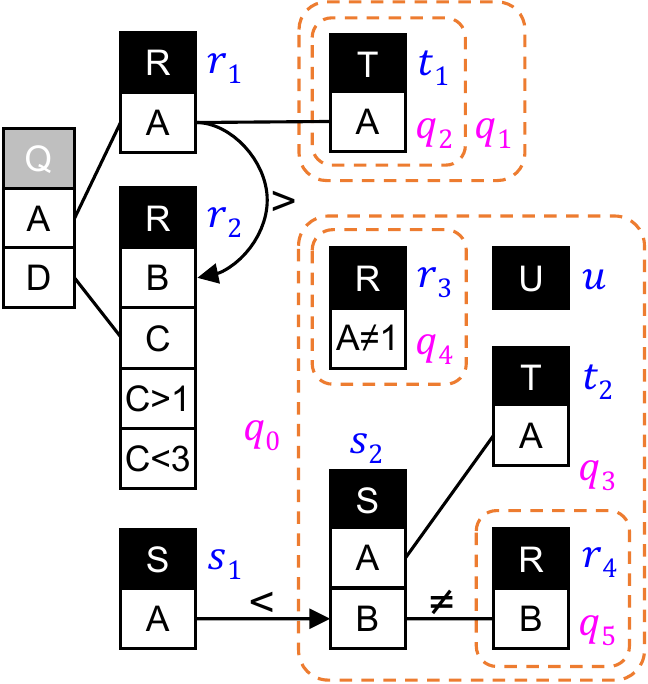}
    \caption{\diagramNDnomath}
    \label{Fig_TRC_vs_RD_b}
\end{subfigure}
\begin{subfigure}[b]{0.43\linewidth}
	\centering	
	\includegraphics[scale=0.42]{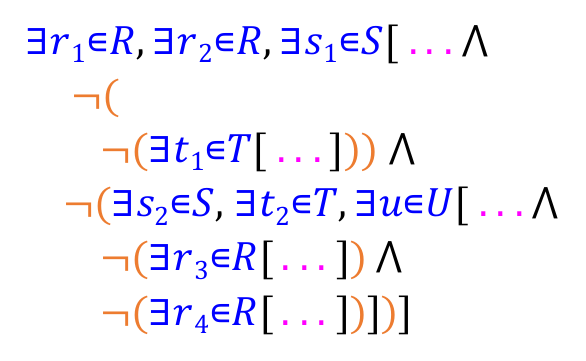}
	\vspace{2mm}	
	\caption{$\NDTRC$ stub}
	\label{Fig_TRC_vs_RD_intermediate}
\end{subfigure}
\caption{
\Cref{sec:fromTRCtoRD}: Example $\NDTRC$ expression (a), 
derivation of the negation hierarchy (b, c),
and corresponding \diagramNDnomath\ (d).
Colored partitions $q_i$ (purple) and table variables $r_i$ (blue) are not part of the diagrams and shown only to discuss the correspondence.
\Cref{sec:fromRDtoTRC}: $\NDTRC$ stub after step 2 of the translation (e).
}
\label{Fig_TRC_vs_RD}
\end{figure}

\emph{(3) Placing selection predicates}:
The predicates within each scope are combined via conjunction and are thus added one after the other.
Since all selection predicates are guarded, the selection predicates can be placed in the same partition as their respective table, which allows correct interpretation (see \cref{sec:fromRDtoTRC}).
For example, for 
$\neg(\exists r_3 \in R[r_3.A \neq 1])$, 
the predicate ``$A \neq 1$'' is placed directly below $R$ in $q_4$.
An example of a predicate that is not guarded would be $\exists r_3 \in R[\neg(r_3.A=1)]$: the scope of the negation contains a predicate of a table that is not existentially quantified in that scope.

\emph{(4) Placing join predicates}:
For each join predicate, we add the two attributes (if not already present) and connect them via an edge with any comparison operator drawn in the middle.
An attribute participating in multiple join predicates needs to be shown only once. 
Equi-joins are the standard and no operator is shown. 
Asymmetric joins include an arrowhead at one end of the edge (see \cref{sec:visualelements}).
As for guarded join predicates one or both attributes are in the partition of a local table, 
the negation can be correctly interpreted.
An example of an unguarded predicate would be
$\neg(r_4.B=s_2.B)$. 
What is possible is the logically-equivalent 
$r_4.B \neq s_2.B$ (as long as one of the two attributes is in the local scope of the last negation.
In our example, this is the case in 
$\neg(\exists r_4 \in R[r_4.B \neq s_2.B])$.

\emph{(5) Place and connect output table}:
The safety conditions 
for $\TRC$ \cite{Ullman1988PrinceplesOfDatabase} 
imply that output predicates can only be chosen from tables outside of all negations,
thus in the root scope or partition $q_0$.
If the query is non-Boolean, we add a new table named $Q$ for query
with a unique gray background \selectBox{} to imply the difference from table references.
If the query is Boolean, there is no output table and the query represents a logical sentence that is true or false.

\introparagraph{Completeness}
This five-step translation guarantees uniqueness of the following aspects:
(1) nesting hierarchy (corresponding to the negation hierarchy),
(2) where tables are placed (canvas partitions corresponding to the negation scope),
(3) which attributes have selection predicates, and
(4) which attributes participate in joins and how.
The following aspects are not uniquely defined
(without impact on the later interpretation):
(1) the order of attributes below each table;
(2) the direction of arrows can be flipped with a simultaneous label flip
e.g., 
$s_1.A {\scriptscriptstyle \xleftarrow{>}} s_2.B$
and
$s_1.A {\scriptscriptstyle \xrightarrow{<}} s_2.B$
are identical (by convention we avoid arrows from right-to-left, but allow them up-to-down and down-to-up);
(3) the size of visual elements and their relative arrangement; and
(4) any optional changes in style (e.g.\ other than dashed negation boxes, distinct visual appearance between tables and attributes).

\subsection{From \diagramNDnomath to TRC}
\label{sec:fromRDtoTRC}

We next describe the reverse five-step translation from any valid \diagramNDnomath\ to a valid and unique $\NDTRC$ expression. 
At the end, we summarize the conditions of a \diagramNDnomath\ to be valid,
which are the set of requirements listed for each of the five steps. 
We again illustrate with the examples from \cref{Fig_TRC_vs_RD}.

\emph{(1) Determine the nested scopes of negation}:
From the nested canvas partitions (\cref{Fig_TRC_vs_RD_d}), 
create the nested scopes of the negation operators of the later $\NDTRC$ expression (\cref{Fig_TRC_vs_RD_c}).

\emph{(2) Quantification of table variables}:
Each table in a partition corresponds to an existentially-quantified table variable. 
WLOG, we use a small letter indexed by the number of occurrence for repeated tables.
We add those quantified table variables in the respective scope of the negation hierarchy (\cref{Fig_TRC_vs_RD_d}).
For example, table $T$ in $q_2$ becomes
$\exists t_1 \in T[\ldots]$ and replaces $q_2$ in \cref{Fig_TRC_vs_RD_intermediate}.
Notice that partition $q_1$ is empty and the resulting negation scope does not contain any expression other than another negation scope.
We require that the leaves of the partition are not empty and contain at least one table.
Otherwise, expressions $\wedge \neg()$ and $\wedge \neg(\neg())$
would both have to be true, leaving the meaning of an empty leaf partition ambiguous.
This also implies that an empty canvas (there is only one partition, in which root and leaf are empty) is not a valid \diagramNDnomath.

\emph{(3) Selection predicates}:
Selection attributes are placed into the scope in which its table is defined.
For example, the predicate $R.A \neq 1$ in partition $q_4$ leads to
$\neg(\exists r_3 \in R[r_3.A \neq 1])$.

\emph{(4) Join predicates}:
For join predicates (lines connecting attributes in \diagramsNDnomath\ with optional direction and operator), we have a validity condition 
that they can only connect attributes of tables that are in the same partition or different partitions that are in a direct-descendant relationship.
In our example, $T.A$ in $q_2$ connects to $R.A$ in $q_0$ (here $q_0$ is the root and grandparent of $q_2$.)
However, we could not connect any attribute in $q_5$ with any attribute in $q_4$ (which are siblings in the nesting hierarchy).
This requirement is the topological equivalent of scopes for quantified variables in $\TRC$ and
guarantees that only already-defined table variables are referenced.
Each such predicate is placed in the scope of the lower of the two partitions in the hierarchy, which guarantees the predicate to be guarded. 
For example, the inequality join connecting $S.B$ in $q_3$ and $R.B$ in $q_5$ is placed in the scope of $q_5$.

\emph{(5) Output table}:
The validity condition for the output table is that each of its one or more attributes is connected to exactly one attribute from a table in the root partition $q_0$. This corresponds to the standard safety condition of safe $\TRC$.
This step adds the set parentheses, the output tables, and its attribute and output predicates shown in green in \cref{Fig_TRC_vs_RD_a}
for non-Boolean queries.

\introparagraph{Soundness}
Notice that this five-step translation guarantees that the resulting $\NDTRC$ is uniquely determined up to 
(1) renaming of the tuple variables;
(2) reordering the predicates in conjunctions, and
(3) flipping the left/right positions of attributes in each predicate.
It follows that \diagramsNDnomath\ are sound, 
and their logical interpretation is \emph{unambiguous}.

\subsection{Valid \diagramsNDnomath}
\label{sec:RD_validity}

In order for a \diagramsNDnomath
to be valid we require that
each of the conditions for the five-step translation process is fulfilled.

\begin{definition}[Validity]
\label{def:validRelationalDiagram}
A \diagramNDnomath\ is valid iff:
\begin{enumerate}
\item
The nested hierarchy of optional negation boxes partitions the canvas (any two dashed boxes are either disjoint or one is completely contained within the other).

\item
Each table, its attributes, and its selection predicates are discernible and reside in exactly one canvas partition.

\item
Each leaf in the canvas partition contains at least one table.

\item
Joins only happen between attributes of tables in partitions that are descendants (not siblings or their descendants).
Join predicates with asymmetric operators such as $<$ and $>$ require a line with directionality (e.g.\ an arrowhead).

\item
If there is an output table, then it has at least one attribute, 
and each attribute connects to exactly one attribute in the root partition $q_0$ 
(safety condition of $\TRC$).
\end{enumerate}
\end{definition}

\begin{theorem}[Unambiguous \diagramsNDnomath]
\label{unambiguous_RDs}
Every valid \diagramNDnomath\ has an unambiguous interpretation in $\NDTRC$.
\end{theorem}

The constructive translations from 
\cref{sec:fromTRCtoRD,sec:fromRDtoTRC} form the proof.
Also notice that there is an additional validity condition we will add later in \cref{def:validRelationalDiagramCont}
that will extend the logical expressiveness to include disjunction and go beyond $\NDTRC$.

\subsection{Logical statements}
\label{sec:sentences}
\label{SEC:SENTENCES}

Boolean queries (or logical sentences) are formulas without free variables.
Being able to express 
\emph{relational sentences} (or constraints)
allows us to compare our formalism 
against a long history of formalisms for logical statements~\cite{ICDE:2024:diagrammatic:tutorial}.
An additional freedom with sentences is that the otherwise important safety conditions of
relational calculus vanish. 
Thus, we need to be able to express statements that do not have any existentially-quantified relations in the main canvas.
We next give an intuitive example, with more example given in
\iflabelexists{sec:BooleanExamples}
{\cref{sec:BooleanExamples}.}
{the online appendix~\cite{relationaldiagrams:links}.}

\begin{example}
\label{ex:negatedBooleanQuery}
Consider 
the statement: \emph{``All sailors reserve a red boat.''}\label{ex:allsailors}
\begin{align}
\begin{aligned}
\neg(\exists s \in \mathit{Sailor} 
	&[\neg (\exists b \in \mathit{Boat}, r \in \mathit{Reserves} [b.\mathit{color} = \mathit{'red'} \wedge 	\\
	&r.\mathit{bid} = b.\mathit{bid} \wedge r.\mathit{sid} = s.\mathit{sid}])])
\end{aligned}
\label{trc:allsailoredboats}
\end{align}

The first 4 steps of the translation in \cref{sec:fromTRCtoRD} still work:
the root canvas $q_0$
does not contain any relation
(\cref{Fig_Sailor_all_Sailors_some_red_boat}).
Similarly, the equivalent canonical $\NDSQL$ statement contains no \sql{FROM} clause before the first \sql{NOT}.
Notice that 
\cref{def:isomorphism} of query pattern isomorphism still works as
it is defined based on the relational tables.
\end{example}

\subsection{A note on implementation}

Creating valid \diagramsNDnomath\ programmatically requires 
a spatial layout algorithm that ensures that tables, predicates, and nested multi-layer canvas partitions are drawn unambiguously. To improve readability, it should also reduce edge crossings and edge bendiness.
For initial work in that direction, please see our optimization model approach called \textsc{STRATISFIMAL LAYOUT}~\cite{DiBartolomeo2022StratisfimalLayoutModular}.

\begin{figure}[t]
\centering	
\begin{subfigure}[b]{.23\linewidth}		
\begin{lstlisting}
SELECT not exists
 (SELECT *
 FROM Sailor s
 WHERE not exists
  (SELECT b.bid
  FROM Boat b, Reserves r
  WHERE b.color = 'red' 
  AND r.bid = b.bid
  AND r.sid = s.sid))
\end{lstlisting}
\vspace{-6mm}
\caption{}
\label{sql:Sailor_sentences_universal}
\end{subfigure}	
\hspace{10mm}
\begin{subfigure}[b]{.34\linewidth}
    \includegraphics[scale=0.5]{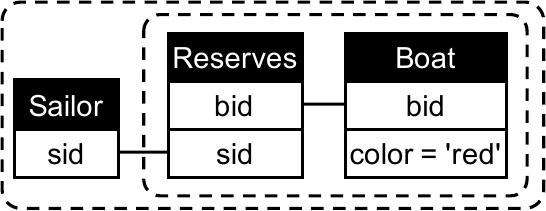}
\vspace{4mm}
\caption{}
\label{Fig_Sailor_all_Sailors_some_red_boat}
\end{subfigure}	
\caption{\Cref{ex:allsailors}: 
All sailors reserve some red boat.}
\label{fig:Sailor_sentences_universal}
\end{figure}

\section{Relational Query Patterns}
\label{sec:structureisomorphism}
\label{SEC:STRUCTUREISOMORPHISM}

\Cref{ex:intro_allboats}
illustrated that---while two languages may well have the same logical expressiveness---one of them 
may have more ways to represent ``logical patterns'' than the other.
We are interested in making this intuition more precise and \emph{establishing a language-independent formalism} 
that captures the so-far vague notion of
a relational query pattern.
The formalism should allow us to study 
the
``relative pattern expressiveness'' of languages, i.e.:
\emph{Can languages $\mathcal{L}_2$ express all patterns that language $\mathcal{L}_1$ 
can?}
We will then apply this formalism to the 
non-disjunctive fragment
and compare the four previously-defined relational query languages 
by their relative abilities 
to represent ``the same set of patterns'' 
as other languages.

In the following, we often need to distinguish between a query as the \emph{query expression} 
(the actual syntax in a particular query language) and a query as a \emph{logical function} that maps 
a set of input tables to an output table. If we need to be precise, we 
refer to the function implied by a query as the \emph{query semantics}
and the actual syntax as the \emph{query expression}. 
We use the word \emph{signature} to refer to an ordered argument list as the input to a function, 
and use bracket notation for indexing. 
For example, the signature of $f(x,y)$ is $\mathcal{S} = (x,y)$, and the first element is $\mathcal{S}[1] = x$.

\subsection{Defining relational query patterns}
\label{sec:defining_patterns}

\introparagraph{Intuition}
Our goal is to define relational patterns in a way that allows us to analyze and compare 
\emph{any relational query languages} irrespective of their syntax.
Our idea is to formalize patterns based on the only common symbols in queries across languages: 
references to the \emph{input relations} from the database.
Since every
relational query language needs to use input tables, the resulting formalism generalizes.
Intuitively, we will define two queries to be \emph{pattern-isomorphic}\footnote{Recall that an \emph{isomorphism} is a reversible \emph{structure-preserving} mapping between two structures.
For it to be reversible, it needs to be
surjective
(each element in the target is mapped to)
 and injective
(different elements in the source map to different elements in the target)
\cite{Gallier:2011ul}.
We use the term \emph{pattern-isomorphic} (instead of structure-isomorphic)
since our focus is on particular relational structures we refer to as patterns.
} if there is 
a one-to-one correspondence that pairs each table in one
query with a table in the other query that ``plays the same semantic role.''
This means that when applying identical changes to these paired
input tables (e.g.\ inserting a tuple), both queries will make identical changes to their outputs.
However, for queries with multiple occurrences of the same
input table (also called self-joins),
we need to 
treat such repeated occurrences of the same input table as if they were \emph{independent tables}.
We will refer to such repeated table occurrences as   
``\emph{table references}.''

For example, consider $q = R - \big( \pi_A R \times S \big)$.
The semantics of this query expression is a function $q(R, S)$ that maps input relations $R$ and $S$ to an output table,
and its signature would be
just its relational input $(R, S)$.
However, we will not be interested in the signature of a query semantics, 
but rather the \emph{signature of a query expression}, 
since we need to capture that two occurrences of $R$ play different semantic roles
in the query. 
In order to capture these different roles, we assign unique names
to each table reference, resulting in what we call the
\emph{dissociated query} $q' = R_1 - \big( \pi_A R_2 \times S \big)$.
We then formally define the \emph{relational query pattern} of $q$ 
as the logical function ${q}'(R_1, R_2, S)$ expressed by the dissociated query $q'$.
Notice that our definition of ``dissociation'' is inspired by, yet slightly different from 
its original use in the context of probabilistic inference~\cite{DBLP:journals/tods/GatterbauerS14}
and the complexity of resilience and causal responsibility~\cite{DBLP:journals/pvldb/FreireGIM15}.

\introparagraph{Formalization}
To make our definitions precise across relational query languages
with different syntax,
we need to unambiguously refer to the individual occurrences of
relational input tables in a given query expression, irrespective of the language used.

\begin{definition}[Query signature]
	A \emph{table reference} in a query expression $q$ 
	is any existentially or universally quantified reference to an input table.
	The \emph{signature} $\S$ of $q$ is the ordered list of its table references.
\end{definition}

For example,
the symbol ``$R$'' is a table reference
in the SQL fragment ``\sql{FROM R as R1}'', 
the $\TRC$ fragment ``$\exists r_1 \in R$'',
the $\RA$ fragment ``$\pi_A R$'', and the 
$\Datalog$ fragment ``$R(x,\_)$''.\footnote{Existential quantification either happens explicitly as in $\TRC$, or implicitly as in $\RA$.} 
In contrast, 
the symbol ``$R$'' is not a table reference 
in the SQL fragment ``\sql{WHERE R=1}''
as it is part of a reference to an attribute of a previously-defined table variable
and not part of an existentially-quantified statement.
The signature of a conjunctive $\SQL$ query with FROM clause
``\sql{FROM R as R1, R as R2, S}''
is $\S = (R, R, S)$.

\begin{definition}[Dissociated query]
A dissociation of a query expression $q$ with signature $\S$
is a modified query $q'$ with 
$\S$ being replaced with a table signature $\S'$
of same size (i.e.\ $|\S'| = |\S|$),
where
every table in $\S'$ has a different name,
and
every table $\S'[i]$
has the same schema as table $\S[i]$
for all $i \in [|\S|]$.
\end{definition}

We call $\S'$ the dissociated signature of $q$.
It is easy to dissociate a query by simply replacing duplicate names
in $\S$ with fresh names. For simplicity, we will use subscripts
when dissociating tables.

\begin{example}[Dissociation]
\label{ex:formalism}
The RA query 
$q = R - \big( \pi_A R \times S \big)$
has signature $\S = (R, R, S)$ with two of the three table references referring to the same input table $R$.
Replacing the signature $\S$ with a dissociated signature
$\S'=(R_1, R_2, S)$ leads to the dissociated query
$q' = R_1 - \big( \pi_A R_2 \times S \big)$.
Since the dissociated tables $R_1, R_2$ inherit the schema information from table $R$, 
the dissociated query is still a valid query
and represents a new relational function $q'(R_1, R_2, S)$ 
that maps three \emph{different} input tables to an output.
\end{example}

The intuition behind this formalism is that the dissociated query
defines a function that maps a set of \emph{table references} (not just a set of input tables) to an output table.
Thus, \emph{the dissociated query is a semantic definition}
of a relational query pattern across different relational query languages.
Two queries use the same query pattern if their dissociated queries are logically equivalent,
up to renaming and reordering of the input tables.

\begin{definition}[Relational pattern]
	Given a query expression $q$ with 
 	signature $\S$.
	The \emph{relational pattern} of $q$ is the logical function defined by its
	dissociated query $q'(\S')$.
\end{definition}

\begin{definition}[Pattern isomorphism]
	\label{def:isomorphism}
	Given two logically-equivalent queries $q_1$ and $q_2$
	with signatures $\S_1$ and $\S_2$,
	and dissociated queries $q_1'(\S'_1)$
	and
	$q_2'(\S'_2)$, respectively.
	The queries are \emph{pattern-isomorphic} iff 
	$q_1'(\S'_1) = 	q_2'(\pi(\S'_1))$
	for some permutation $\pi$.
	In that case, 
	we call the 
	bijection $\S_1[i] \mapsto \S_2[\pi(i)]$
	between the query signatures
	a \emph{pattern-preserving mapping}.
\end{definition}

\begin{example}[Patterns]
Next, consider the $\Datalog$ query
\begin{align*}
	I(x,y) &\datarule R(x,\_), S(y).\\
	Q_1(x,y) &\datarule R(x,z), \neg I(x,y).
\end{align*}
with signature $\S_1 \!=\! (R, S, R)$.
Then its dissociated query is 
\begin{align*}
	I(x,y)  &\datarule R_{\h{1}}(x,\_), S_{\h{2}}(y).\\
	Q'(x,y) &\datarule R_{\h{3}}(x,z), \neg I(x,y).
\end{align*}
with signature $\S_1' = (R_1, S_2, R_3)$.
Notice that $Q'$ defines a logical function $Q'(R_1, S_2, R_3)$ mapping two input tables with the same schema as $R$
and an input table $S$ to a binary output table.

Next, consider the RA query $q$ 
from \cref{ex:formalism}
with signature $\S_2=(R, R, S)$.
Notice that above $\Datalog$ query $Q$ 
and this $\RA$ query $q$
are pattern-isomorphic
since their dissociated queries define the same logical function up to permutation in the signatures:
$Q_1'(R_1, S_2, R_3) = q'(R_3, R_1, S_2) = q'(\pi(R_1, S_2, R_3))$ for permutation $\pi = (2, 3, 1)$.
Thus the mapping 
$(\S_1[1], \S_1[2], \S_1[3])
\mapsto
(\S_2[2], \S_2[3], \S_2[1])$
is a \emph{pattern-preserving mapping}.
\end{example}

\introparagraph{Complexity of deciding pattern isomorphism}
Deciding whether two relational queries are pattern-isomorphic 
is undecidable, in general
(we need to determine whether two queries are equivalent, both before and after dissociation).
This follows from Trakhtenbrot's theorem stating that the problem of validity in first-order logic on finite models is undecidable,
and thus also the logical equivalence of relational queries
(see, e.g., the reduction in \cite{ArenasBLMP:DatabaseTheory}).
However, we get a one-sided guarantee: if we can determine whether two queries are logically equivalent, 
then we can also determine whether they are pattern-isomorphic.
In practice, equivalence of relational queries can often be determined,
even for sophisticated SQL queries with grouping and aggregation evaluated over bags or sets~\cite{DBLP:journals/pvldb/ChuMRCS18}.

\subsection{Discussion with illustrating example}
\label{sec:intuitionforpatterndefinition}

We next give an
example of queries that have different query patterns
\emph{although they are logically equivalent and have the same table signature}.
This detailed example motivates to a large extent why we define query patterns 
based on the \emph{dissociated} signature.

\begin{example}[Different patterns]
\label{ex:querystructureiso1}
Consider
table $R(A,B)$ and the two $\Datalog$ queries
$Q_1(R)$ and $Q_2(R)$ with
\begin{align*}
	Q_1(x) 	&\datarule \h{R}(x,\_), \h{R}(x,\_). \\
	Q_2(x)	&\datarule \h{R}(x,y),\h{R}(\_,y).
\end{align*}
Both queries are logically equivalent to $Q(x) \datarule \h{R}(x, \_)$, 
and thus also logically equivalent to each other.
However, $Q_1$ and $Q_2$ represent different patterns:
$Q_1$ never uses the second attribute of $R$ whereas $Q_2$ uses that attribute to join both occurrences of $R$.
This difference becomes even more apparent in SQL:
\cref{Fig_isomorphism_sql_1}
would even work if $R$ was unary, whereas \cref{Fig_isomorphism_sql_2} requires $R$ to be at least binary.

We next show that table dissociation allows us to formally distinguish the two
patterns. 
First, notice that both queries have two occurrences
of $R$ as table references, and hence we need to associate each individual appearance to the ``role'' it
plays semantically in the query.
We achieve this by first dissociating the two occurrences of $R$ into two fresh input tables (with the same schema).
The two dissociated queries are 
$Q_1'(R_1, R_2)$
and
$Q_2'(R_3, R_4)$
with
\begin{align*}
	Q_1'(x) &\datarule \h{R_1}(x,\_),\h{R_2}(x,\_). \\
	Q_2'(x)	&\datarule \h{R_3}(x,y), \h{R_4}(\_,y).
\end{align*}
It is easy to verify that neither of the two possible mappings
$h_1 = \{(R_1, R_3), (R_2, R_4)\}$ and
$h_2 = \{(R_1, R_4),$ $(R_2, R_3)\}$,
preserves logical equivalence for the dissociated queries.
For example, $R_1(1,2), R_2(1,3)$ returns an answer for $Q_1'$ but not for $Q_2'$, 
under neither $h_1$ nor $h_2$.

However, $Q_1$ is pattern-isomorphic to the $\TRC$ query $q_3(R)$ with
\begin{align*}
	&\{ q_3(A) \mid \exists r_1 \in \h{R}, r_2 \in \h{R} 
	[q.A = r_1.A \wedge r_1.A = r_2.A] \}	\\
\intertext{
To see that, notice that its dissociated query $q_3'(R_5, R_6)$ with
}
	&\{ q_3'(A) \mid \exists r_1 \in \h{R_5}, r_2 \in \h{R_6}
	[q.A = r_1.A \wedge r_1.A = r_2.A] \}	
\end{align*}
allows the isomorphism
$h_3 = \{(R_1, R_5), (R_2, R_6)\}$
from $Q_1'$ to $q_3'$
that preserves logical equivalence.
By the same arguments, $Q_1$ is pattern-isomorphic to the SQL query in 
\cref{Fig_isomorphism_sql_1},
and $Q_2$ is pattern-isomorphic to the SQL query in
\cref{Fig_isomorphism_sql_2}.

\end{example}

\begin{figure}[t]
\centering
\begin{subfigure}[b]{.185\linewidth}
\begin{lstlisting}
SELECT DISTINCT R1.A
FROM R R1, R R2
WHERE R1.A = R2.A
\end{lstlisting}
\vspace{-7mm}
\caption{}
\label{Fig_isomorphism_sql_1}
\end{subfigure}
\begin{subfigure}[b]{.32\linewidth}
	\centering
	\includegraphics[scale=0.4]{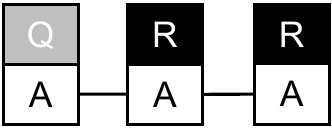}
	\vspace{-1mm}
	\caption{}
	\label{Fig_isomorphism_def_1}
\end{subfigure}
\hspace{45mm}
\begin{subfigure}[b]{.185\linewidth}
	\vspace{1mm}	
\begin{lstlisting}
SELECT DISTINCT R1.A
FROM R R1, R R2
WHERE R1.B = R2.B
\end{lstlisting}
\vspace{-7mm}
\caption{}
\label{Fig_isomorphism_sql_2}
\end{subfigure}
\begin{subfigure}[b]{.32\linewidth}
	\centering
	\vspace{1mm}
	\includegraphics[scale=0.4]{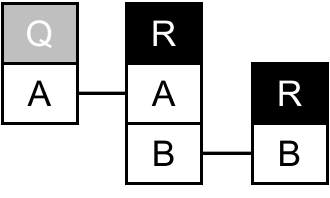}
	\vspace{-1mm}
	\caption{}
	\label{Fig_isomorphism_def_2}
\end{subfigure}
\caption{
\Cref{ex:querystructureiso1}: 
Two queries (a) and (c) with identical signatures that are logically equivalent but \emph{not pattern-isomorphic}.
Their associated \diagrams\ are shown in (b) and (d), respectively.}
\label{Fig_isomorphism_def}
\end{figure}

By design, our definition excludes views and intermediate tables such as Intensional Database Predicates in $\Datalog$ 
from the definition of table references.
To see why, consider a query  
returning nodes that form the starting point of a length-3 directed path:
\begin{align*}
	Q_1(x) 	&\datarule \h{E}(x,y), \h{E}(y,z), \h{E}(z,w).
\end{align*}
The following $\Datalog$ query uses
the same logical pattern (find three edges that join and keep the starting node),
even though it defines the intermediate intensional database predicate $I$:
\begin{align*}
	I(y)  	&\datarule \h{E}(y,z), \h{E}(z,w). \\
	Q_1'(x) &\datarule \h{E}(x,y), I(y).
\end{align*}

\subsection{Comparing the ``pattern-expressiveness''
of relational languages}
\label{sec:patternexpressiveness}

We next add the final definition needed to formally compare relational query languages based on their relative abilities to represent relational query patterns.

\begin{definition}[Representation equivalence]
	We say that a query language $\L_2$ can \emph{pattern-represent} a query language $\L_1$
	(written as $\L_1 \subseteq^\rep \L_2$) 
	iff 
	for every legal query expression $q_1 \in \L_1$ 
	there is a pattern-isomorphic query $q_2 \in \L_2$.
	We call a query languages $\L_2$ \emph{pattern-dominating} another language $\L_1$ 
	(written as $\L_1 \subsetneq^\rep \L_2$)
	iff
	$\L_1 \subseteq^\rep \L_2$
	but
	$\L_1 \not \supseteq^\rep \L_2$.
	We call $\L_1, \L_2$ \emph{representation equivalent}
	(written as $\L_1 \equiv^\rep \L_2$)
	iff 
	$\L_1 \subseteq^\rep \L_2$
	and
	$\L_1 \supseteq^\rep \L_2$,
	i.e., both language can represent the same set of relational patterns.	
\end{definition}

We are now ready to state our result on the hierarchy of pattern expressiveness 
of the
non-disjunctive fragment of the four languages defined earlier (\cref{sec:nondisjunctivefragment})
and our proposed relational diagrammatic representation \diagramsNDnomath\  (\cref{sec:QV}):

\begin{theorem}[Representation hierarchy]
	\label{th:representations}
	\label{TH:REPRESENTATIONS}%
	$\NDRA \subsetneq^\rep \DatalogND \subsetneq^\rep \NDTRC 
	\equiv^\rep \NDSQL
	\equiv^\rep \diagramsND$ (see \cref{Fig_Representation_proof_main}).
\end{theorem}

\iflabelexists{appendix:beginning}
{The proofs is provided in \cref{appendix:proofofhierarchy}.
In addition, \cref{appendix:antijoins}
shows how the separation between $\NDRA$ and $\DatalogND$ disappears after adding
the antijoin operator to the basic operators of $\NDRA$,
while the separation from $\NDTRC$ remains.
}
{In addition to containing all proofs, our optional online appendix~\cite{relationaldiagrams:links}
shows how the separation between $\NDRA$ and $\DatalogND$ disappears after adding
the antijoin operator to the basic operators of $\NDRA$,
while the separation from $\NDTRC$ remains.}
The proof demonstrates that \emph{relational calculus has relational patterns that cannot be expressed in relational algebra}.
The important consequence is that \textit{$\NDRA$, $\DatalogND$ or any diagrammatic language modeled after them 
would not be a suitable target language for helping users understand all existing relational query patterns}
(including those used by $\NDSQL$).
Our related work (\cref{sec:relatedWork})
shows that most existing visual query representations are modeled after relational algebra in that they model data flowing between relational operators, which implies they cannot faithfully 
represent 
all relational query patterns from $\NDTRC$
or $\NDSQL$.

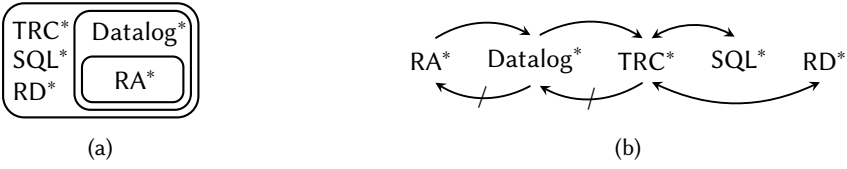
\begin{figure}[t]
\centering
\begin{subfigure}[b]{0.31\linewidth}
	\centering
	\begin{tikzpicture}[scale=1.0]
	\draw
	  (1.05,1.1) node[anchor=west] {$\DatalogND$}
	  (1.35,0.5) node[anchor=west] {$\NDRA$}  
	  (0.0,1.15) node[anchor=west] {$\NDTRC$}
	  (0.0,0.75) node[anchor=west] {$\NDSQL$}  
	  (0.0,0.35) node[anchor=west] {$\NRD$};        
	\draw[black,rounded corners=8,thick]
		(0,0) rectangle (2.6, 1.5);
	\draw[black,rounded corners=6,thick]
	    (0.95,0.1) rectangle (2.5,1.4);
	\draw[black,rounded corners=4,thick]
	    (1.05,0.2) rectangle (2.4,0.8);	
	\end{tikzpicture}
    \caption{}
\end{subfigure}	
\hspace{1mm}
\begin{subfigure}[b]{0.66\linewidth}
	\centering
	\begin{tikzpicture}[>=stealth, shorten >= 2pt, shorten <= 2pt, line width=0.25mm]
	  \node (RA) {$\NDRA$};
	  \node[right=2mm of RA] (Datalog) {$\DatalogND$};
	  \node[right=2mm of Datalog] (TRC) {$\NDTRC$};
	  \node[right=2mm of TRC] (SQL) {$\NDSQL$};  
	  \node[right=2mm of SQL] (RD) {$\NRD$};    
    
	  \draw[->] (RA.north) to[bend left] 
		(Datalog.north);
	  \draw[->] (Datalog.south) to[bend left] 
	  	node[pos=0.5, sloped] {\small$/$}
		(RA.south);
	  \draw[->] (Datalog.north) to[bend left] 
		(TRC.north);
	  \draw[->] (TRC.south) to[bend left] 
	   	node[pos=0.5, sloped] {\small$/$}
		(Datalog.south);
	  \draw[<->] (TRC.north) to[bend left] 
		(SQL.north);
	  \draw[<->] (TRC.south) to[bend right=25] 
		(RD.south);	
	\end{tikzpicture}
    \caption{}
\end{subfigure}	
\caption{
\cref{th:representations}: (a) A diagram summarizing the representation hierarchy between the non-disjunctive fragments of 4 query languages
and $\diagramsND$ (shown as $\NRD$).
(b) Directions of pattern-preservation (and non-preservation) used in the proofs.
}
\label{Fig_Representation_proof_main}
\end{figure}

\subsection{Similar patterns across schemas}
\label{sec:pattern_schemas}

We next extend the notion of pattern equivalence to allow comparing queries across \emph{different schemas}.
We call this concept ``\emph{pattern similarity}'' 
and define it as a Boolean condition: two queries either have a similar pattern or not.
The intuition is simple and best illustrated with the two queries from \cref{Fig_intro_textbook}:
As written those queries are not logically equivalent and thus they can't be pattern-isomorphic.
However, if we first replace the table and attribute names from $q_1$ with table names from $q_2$ in a reversible (thus bijective) way, 
then the thus modified query $q_1'$ would be pattern-isomorphic to $q_2$.

More formally, call a \emph{schema mapping} $\lambda$ from query $q_1$ to $q_2$,
a bijective mapping that
replaces table names, attribute names, constants, and attribute order appearing in $q_1$ with those from $q_2$.

\begin{definition}[Similar Patterns]
	\label{def:similarity}
	Given two queries $q_1$ and $q_2$.
	The queries use a \emph{similar pattern} iff
	there is a schema mapping $\lambda$ from $q_1$ to $q_2$ s.t.\
	$\lambda(q_1)$ and $q_2$ are pattern-isomorphic.
\end{definition}

\begin{example}
	For our example from \cref{Fig_intro_textbook}, consider
	a mapping $\lambda$ that replaces
	`Sailor' with `SX',
	`Reserves' with `SPX',
	`Boat' with `PX',
	`sname' with `sname',
	`sid' with `sno', and
	`bid' with `pno'.	
	Then the thus modified query $\lambda(q_1)$ is pattern-isomorphic with $q_2$. 
\end{example}

\section{Relational completeness}
\label{sec:completeness}
\label{SEC:COMPLETENESS}
\label{sec:unionForDiagrams}

\begin{figure}[t]
\centering
\begin{subfigure}[b]{0.175\linewidth}		
\begin{lstlisting}
SELECT DISTINCT R.A
FROM R
WHERE not exists
 (SELECT *
 FROM S
 WHERE not exists
  (SELECT *
  FROM R AS R2
  WHERE (R2.B = S.B
  OR R2.C = S.C)
  AND R2.A = R.A))
\end{lstlisting}
\vspace{5.5mm}
    \caption{}
    \label{fig:disjunctionSQL}
\end{subfigure}
\hspace{3mm}
\begin{subfigure}[b]{0.175\linewidth}		
\begin{lstlisting}
SELECT DISTINCT R.A
FROM R
WHERE not exists
 (SELECT *
 FROM S
 WHERE not exists
  (SELECT *
  FROM R AS R2
  WHERE R2.B = S.B)
  AND R2.A = R.A)
 AND not exists
  (SELECT *
  FROM R AS R2
  WHERE R2.C = S.C
  AND R2.A = R.A))
\end{lstlisting}
\vspace{-4.5mm}
    \caption{}
    \label{fig:disjunctionSQL2}
\end{subfigure}	
\hspace{1.5mm}	
\begin{subfigure}[b]{0.27\linewidth}
    \includegraphics[scale=0.42]{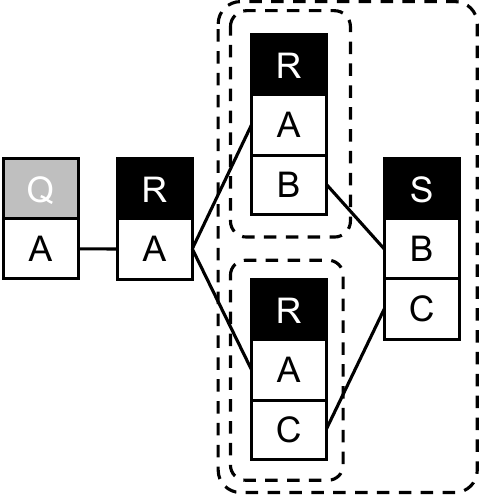}
	\vspace{3mm}
    \caption{}
    \label{Fig_disjunction_QV}
\end{subfigure}	
\hspace{1mm}
\begin{subfigure}[b]{0.12\linewidth}
\begin{lstlisting}
(SELECT 
 DISTINCT R.A
FROM R)
UNION
(SELECT 
 DISTINCT S.A
FROM S)
\end{lstlisting}
\vspace{4.5mm}
	\caption{}
	\label{Fig_simplestunion_SQL}
\end{subfigure}	
\hspace{1.5mm}
\begin{subfigure}[b]{0.125\linewidth}
    \includegraphics[scale=0.39]{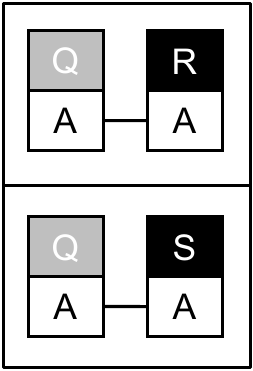}
	\vspace{8.5mm}
    \caption{}
    \label{Fig_simplestunion}
\end{subfigure}	
\caption{Illustrations for \Cref{ex:disjunction1} on replacing disjunctions:
(a) $\SQL$ with disjunctions,
(b) logically-equivalent (yet not representation-equivalent) $\NDSQL$ statement, and
(c) \diagrams.
Illustrations for \Cref{ex:disjunction2} on creating the union of queries:
(d) union of $\NDSQL$ statements, and
(e) \diagrams\ with union cells.}
\label{fig:disjunction}
\end{figure}

To make \diagrams\ relationally complete, we now 
remove the non-disjunction restrictions and
allow disjunctions and unions in
all four relational query languages (\cref{sec:nondisjunctivefragment}).
This means we must also add a corresponding syntactic device to \diagrams\
that achieves logical equivalence to the other
relational query languages.
Unfortunately, this means that \diagrams\ are no longer
representation-equivalent to $\TRC$.
Can this be addressed in the future by a better diagram design? 
Based on the current understanding of the inherent limits of diagrams to express disjunctive
information~\cite{shin_1995,Shin:2002} 
(see also the colored car example in 
\iflabelexists{subsec:disjunchard}{\cref{subsec:disjunchard}}{the online appendix \cite{relationaldiagrams:links}}),
such an extension would require adding \emph{non-diagrammatic abstractions} (i.e.\ ``syntactic devices'').

The syntactic device that makes \diagrams\ relationally complete is 
inspired by the representation of disjunction in $\Datalog$.
It was also proposed by Peirce in his discussion of Euler diagrams~\cite[4.366]{peirce:1933}
(see also \cite[sect.~2.3.1]{shin_1995}):
we allow placing several $\diagramsND$ on the same canvas, each in a separate \emph{union cell}.
Each cell of the canvas then represents only conjunctive information, 
yet the relation among the different cells is disjunctive (a union of the outputs).

We next illustrate with two examples logical transformations that are not pattern-preserving 
but that guarantee relational completeness. 
These transformations, together with union cells,
make \diagrams\
\emph{relationally complete}: 
every query expressible in full $\RA$, safe $\TRC$, $\DatalogN$, or 
our prior $\NDSQL$ fragment extended by union and
disjunctions of predicates\footnote{Extend the grammar from
\cref{table:supported_grammar} 
with one additional rule: 
\textsf{P::= `('P OR P`)'}
and making adjustments for allowing the \textsf{UNION} clause between non-Boolean queries.
}
can then be represented 
as a logically-equivalent \diagram.
The first example shows how to avoid disjunctions if they are not at the root level.
The second shows how to replace disjunctions in the root by unions of queries.

\begin{example}[Replacing disjunctions]\label{ex:disjunction1}
Consider the $\SQL$ query from \cref{fig:disjunctionSQL}
which contains a disjunction and is not in $\NDSQL$. 
Using De Morgan's Law with quantifiers 
$\neg \exists r \in R[A \vee B] = 
\neg (\exists r \in R[A] \vee \exists r \in R[B]) =
\neg \exists r \in R[A] \wedge \neg \exists r \in R[B]$, 
we can first reformulate the conditions including disjunction as DNF, 
and then distribute the quantifier over the conjuncts. 
This leads to a disjunction-free query, yet leads to an increased number of table references:
\begin{align*}
\phantom{=} &\{ q(A) \mid 
	\exists r \in R [q.A \equal r.A \wedge \neg (\exists s \in S\\
	&\hspace{10.5mm}[\h{\neg (\exists r_2 \in R [ }(r_2.B \equal s.B \;\h{\vee}\; r_2.C \equal s.C) \wedge r_2.A \equal r.A \h{] )} ])] \}
\\
= &\{ q(A) \mid 
  \exists r \in R [q.A \equal r.A \wedge \neg (\exists s \in S\\
	&\hspace{10.5mm}[\h{\neg (\exists r_2 \in R [}r_2.B \equal s.B \wedge r_2.A \equal r.A \h{])} \;\h{\wedge}\;	\\
	&\hspace{11.7mm}\h{\neg (\exists r_3 \in R [}r_3.C \equal s.C \wedge r_3.A \equal r.A \h{])} ]] \}	
\end{align*}
\cref{fig:disjunctionSQL2} shows this query as representation-equivalent $\NDSQL$ query,
and \cref{Fig_disjunction_QV} as 
\diagram.
\end{example}

\begin{example}[Union of queries]
	\label{ex:disjunction2}
Consider two unary tables $R(A)$ and $S(A)$ and the $\TRC$ query
\begin{align*}
	\{q(A) \mid \exists r\in R[q.A = r.A] \;\h{\vee}\; \exists s\in S[q.A = s.A]\}
\end{align*}
We can write this query as a union of disjunction-free $\NDTRC$ queries:
\begin{align*}
	\h{\{}q(A) \mid \exists r\in R[q.A = r.A]\h{\}} 
	\;\h{\cup}\; 
	\h{\{}q(A) \mid \exists s\in S[q.A = s.A]\h{\}} 
\end{align*}
\Cref{Fig_simplestunion_SQL} shows a pattern-isomorphic $\SQL$ query, 
and
\cref{Fig_simplestunion} shows it as \diagrams\ with two separate $\diagramsND$ queries,
each in a separate union cell, and each with the same attribute signature in the output table.
This query cannot be rewritten without the union operator in $\RA$,
nor $\diagramsND$ without union cells.
\end{example}

The additional validity criterion for multiple union cells follows the conditions of union or disjunction in the \emph{named perspective}~\cite{DBLP:books/aw/AbiteboulHV95}
of query languages:
for disjunction in $\TRC$, 
each operand needs to have the same arity, and the mapping between them is achieved by reusing the same variables.%

\begin{definition}[Validity---extending \cref{def:validRelationalDiagram})]
\label{def:validRelationalDiagramCont}\hspace{.1px} %
\begin{enumerate}
    \setcounter{enumi}{5}%
    \item The output tables in multiple cells for the same query need to have the same name and same set of attributes.
\end{enumerate}
\end{definition}

\begin{theorem}[Completeness]
\label{th:completeness}
\diagrams\ ($\diagramsND$ extended with union cells) are relationally complete.
\end{theorem}

The proof is
\iflabelexists{appendix:proofoflogicalexpressivness}
{in \cref{app:proofrelationalcompleteness}.}
{in the optional online appendix~\cite{relationaldiagrams:links}.}
It uses the earlier proven logical expressiveness of 
\diagramsNDnomath\
and the fact that disjunctions can either be rewritten with DeMorgan or be pushed to the root.
It also immediately follows that \diagramsNDnomath\ (without union cells) can already express any logical statement.

\begin{corollary}[Completeness]
\label{co:completeness}
Any logical statement in first-order logic can be expressed by a logically-equivalent \diagramsNDnomath.
\end{corollary}

\section{Two Applicability Studies}
\label{sec:studies}

\subsection{\diagrams\ for Textbook Queries}
\label{SEC:TEXTBOOKANALYSIS}
\label{sec:textbookanalysis}

\begin{wrapfigure}{r}{0.6\textwidth}
\vspace*{-2mm}
\includegraphics[scale=0.23]{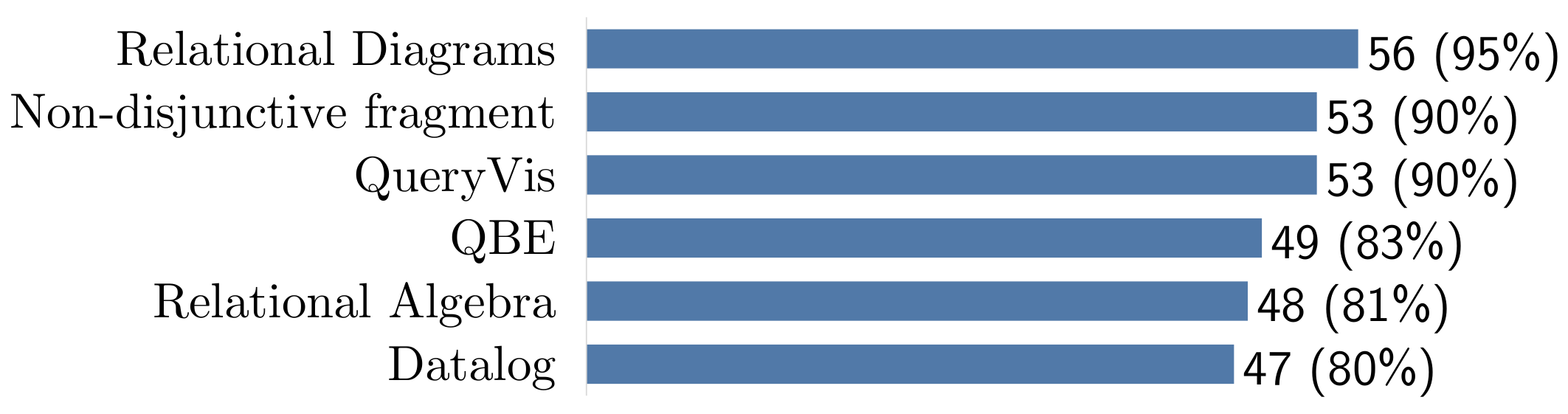}
\caption{\cref{sec:textbookanalysis}: Fraction among 59 queries from 5 textbooks with pattern-isomorphic representations in listed languages.}
\label{Fig_Textbook_Analysis}
\vspace*{-2mm}
\end{wrapfigure}

We analyzed the proportion of relational calculus queries encountered in learning scenarios 
that have pattern-isomorphic representations in either
\diagrams, $\RA$, $\Datalog$, 
$\queryviz$~ \cite{DanaparamitaG2011:QueryViz}, or 
$\QBE$~\cite{DBLP:journals/ibmsj/Zloof77}.
For that purpose, we identify 59 queries across 5 popular textbooks with sections on relational calculus
\cite{cowbook:2002, DBLP:books/mg/SKS20, Elmasri:dq, date2004introduction, ConnollyBegg:2015}.

Among those 59 queries, the number of queries that have pattern-isomorphic representations are
56 (94.9\%) for \diagrams,
53 (89.8\%) for \queryviz,
49 (87.5\%) for \QBE,
48 (85.7\%) for \RA, and
47 (79.7\%) for \Datalog.
The fraction for $\queryviz$ happens to be identical to the \emph{non-disjunctive fragment}.
Standard $\Datalog$ cannot express disjunctions in the body of a query and thus performs worse than $\RA$.\footnote{Modern variants of Datalog exist that can express disjunctions in the body, such as Souffle~\cite{souffle}, but those are not standard.}
For $\QBE$, notice that $\QBE$ 1) can express disjunctions within the same relations, yet
2) also requires the same safety conditions as $\Datalog$. Furthermore, 
theta joins require the use of a non-diagrammatic conditions box~\cite[Appendix C]{Elmasri:dq}.
$\RA$ extended with a primitive antijoin operator covers the same fraction as $\QBE$.
More details and all queries are given in 
\iflabelexists{appendix:textbook}
{\cref{appendix:textbook}}
{the online appendix \cite{relationaldiagrams:links}}.

\subsection{Controlled user study}
\label{sec:userstudy}
\label{SEC:USERSTUDY}

We conducted a controlled experiment on Amazon Mechanical Turk (MTurk)~\cite{AMT}
to evaluate the utility of \diagrams\ for recognizing patterns.
Our study investigates 3 main questions:
(1) Can $\SQL$ users {\emph{identify common relational query patterns faster}} using \diagrams\ than $\SQL$? 
(2) \emph{Can participants identify patterns faster over time}, thus can users learn the patterns under repeated exposure to the same patterns?
(3) Do participants have a similar \emph{accuracy} (i.e.\ a comparable numbers of correct responses) using \diagrams\ or $\SQL$? 
We chose $\SQL$ as a baseline for comparison because we expect that fewer workers on MTurk understand $\TRC$.

\introparagraph{Open practices}
Following best practices in user design, we preregistered the study design on OSF 
before collecting the data \cite{relationaldiagrams:links}.
All code for generating the stimuli, the stimuli, the tutorial provided to participants, the resulting data (pilot $n=13$, study $n=50$), the analysis code, and changes from the preregistration are available on OSF \cite{relationaldiagrams:links}.
More details on the study design and procedure are provided \iflabelexists{appendix:controlled-experiment}
{in \cref{appendix:controlled-experiment}}
{on arXiv \cite{relationaldiagrams:links}}.

\begin{figure}[t]
    \centering	
	\hspace{-2mm}
    \includegraphics[scale=0.36]{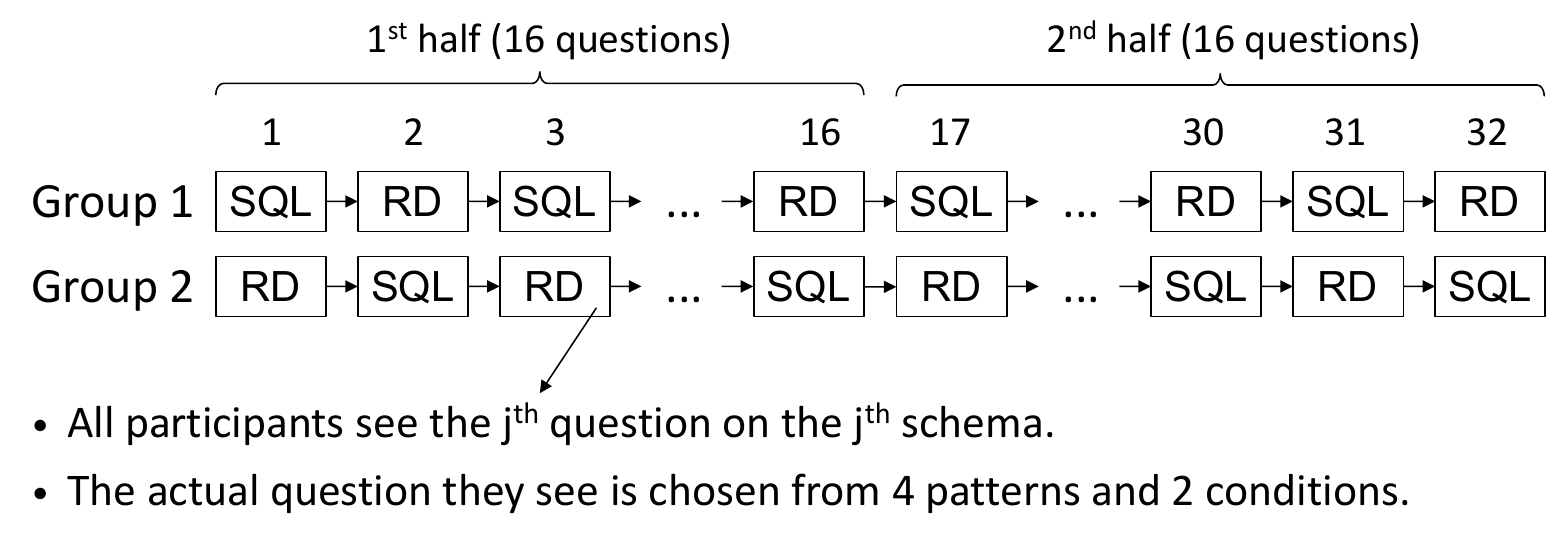}
    \caption{
	Illustration of the randomization and counterbalancing in our within-subjects study design.}
    \label{Fig_study_design}
\end{figure}

\introparagraph{Counterbalanced within-subjects study with randomization}
We asked participants 32 questions, each having them identify which of 4 relational query patterns was presented.
The exposure for each participant alternates between two conditions (\diagrams\ and formatted $\SQL$ text).
Each question uses a different relational schema found or adapted from common textbooks.
We used counterbalancing and randomization to reduce ordering effects.
Concretely, we start half the participants using SQL (group 1)
and the rest using \diagrams\ (group 2), 
after which participants alternate conditions with each question (see \Cref{Fig_study_design}).
Moreover, we randomize the order in which patterns are presented 
such that each pattern-condition combination appeared twice in the first 16 questions 
(1$^\textrm{st}$ half) and twice in the second 16 questions (2$^\textrm{nd}$ half).
The result is that each participant 
sees each of 4 patterns,
in each of 2 conditions,
exactly 2 times
($16 = 4\times 2 \times 2$),
in each of the 2 halves (first and second 16 questions).
This randomization helps reduce cheating as well as order effects.
This setup necessitated creating 256 different stimuli 
(32 schemas $\times$ 4 patterns $\times$ 2 conditions),
creation of which we semi-automated.
The 4 patterns, illustrated with the sailors-reserve-boats schema, 
were:
\begin{enumerate}[noitemsep, nolistsep]
    \item Find \textit{sailors} who have \textit{reserved} \textbf{some} \textit{boat}.
    \item Find \textit{sailors} who have \textbf{not} \textit{reserved} \textbf{any} \textit{boat}.
    \item Find \textit{sailors} who have \textbf{not} \textit{reserved} \textbf{all} \textit{boats}.
    \item Find \textit{sailors} who have \textit{reserved} \textbf{all} \textit{boats}.
\end{enumerate}
We chose these patterns because 
we are interested in how \diagrams\ can be used in educational settings,
and they represent 4 different query structures.
In particular, double-negation from pattern (4) is challenging for novice users to understand and is easy to misinterpret \cite{Luk1986ELFS}.
Moreover, pattern (4) does not have a pattern-isomorphic representation in \RA.

Several prior user studies 
\cite{Reisner1975:HumanFactors,DBLP:journals/csur/Reisner81,Leventidis2020QueryVis}
have shown that diagrammatic representations of queries can help users understand queries faster.
Key innovations in our design are:
(1) We repeatedly expose participants to 4 identical patterns across 32 different schemas, allowing us to semi-automatically design 256 questions 
instead of a few hand-curated ones (each participant saw only 32, one for each schema).
(2) Our questions are balanced across the first 16 and second 16 questions, which allows us to track learning over time.
We are not aware of prior study design that allowed studying learning in an online study.
(3) Our setup is randomized and parameterized, which creates a space of $2 \cdot 2540^4$ possible treatments, i.e., participants are unlikely to see the same question sequence, reducing the chance of cheating.
(4) We used monetary incentives for both time and accuracy, inspired by our recent work on $\queryviz$~\cite{Leventidis2020QueryVis},

\begin{figure}[t]
\centering	
\begin{subfigure}[b]{0.49\linewidth}	
    \includegraphics[scale=0.3]{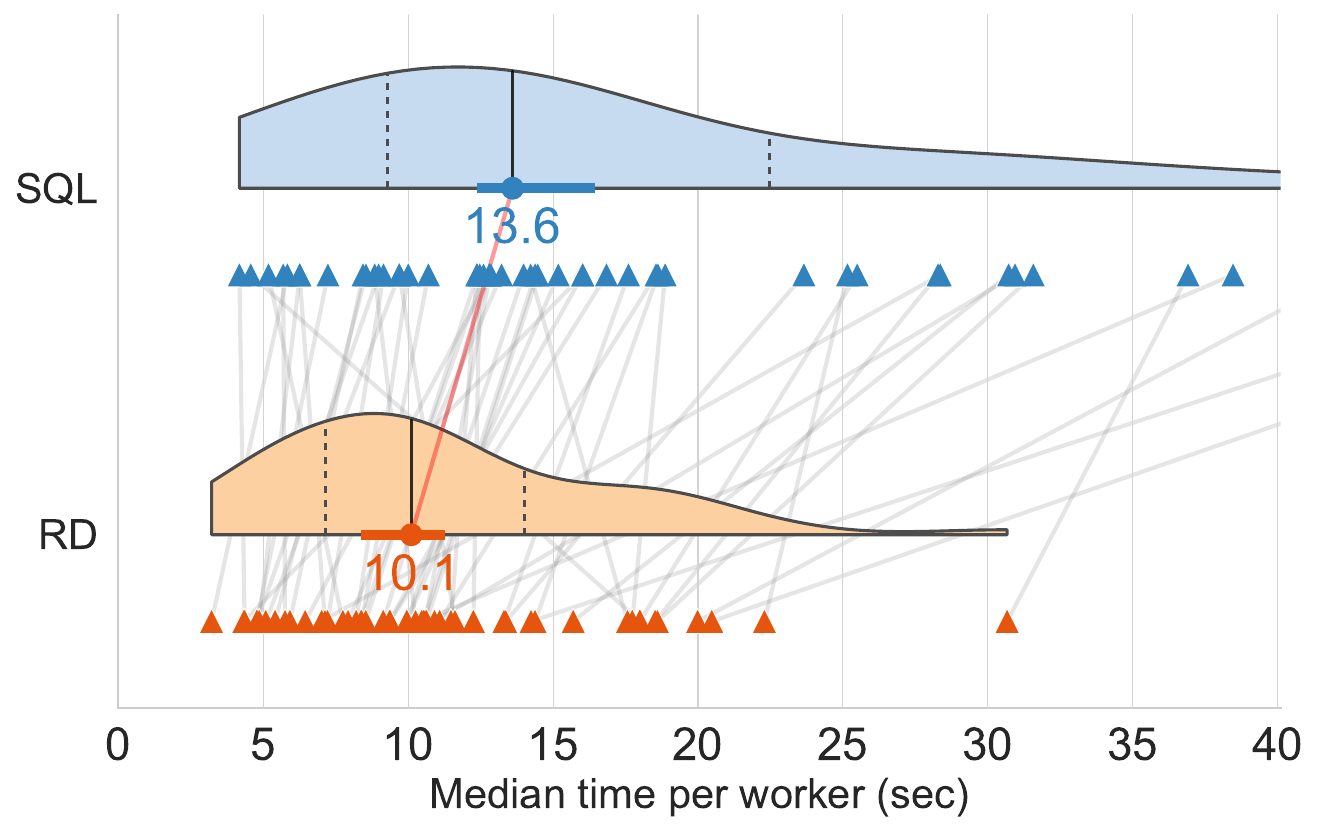}
    \includegraphics[scale=0.3]{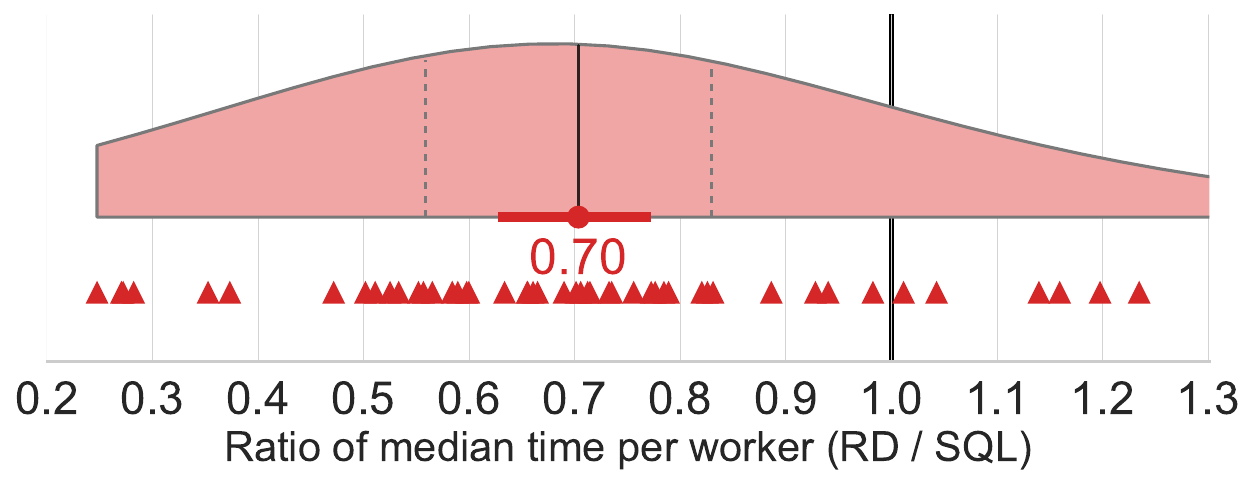}	
	\vspace{-2mm}	
    \caption{Result 1. (Speed)}
    \label{q1_figure1_variant1}
\end{subfigure}
\begin{subfigure}[b]{0.49\linewidth}	
    \centering	
    \includegraphics[scale=0.3]{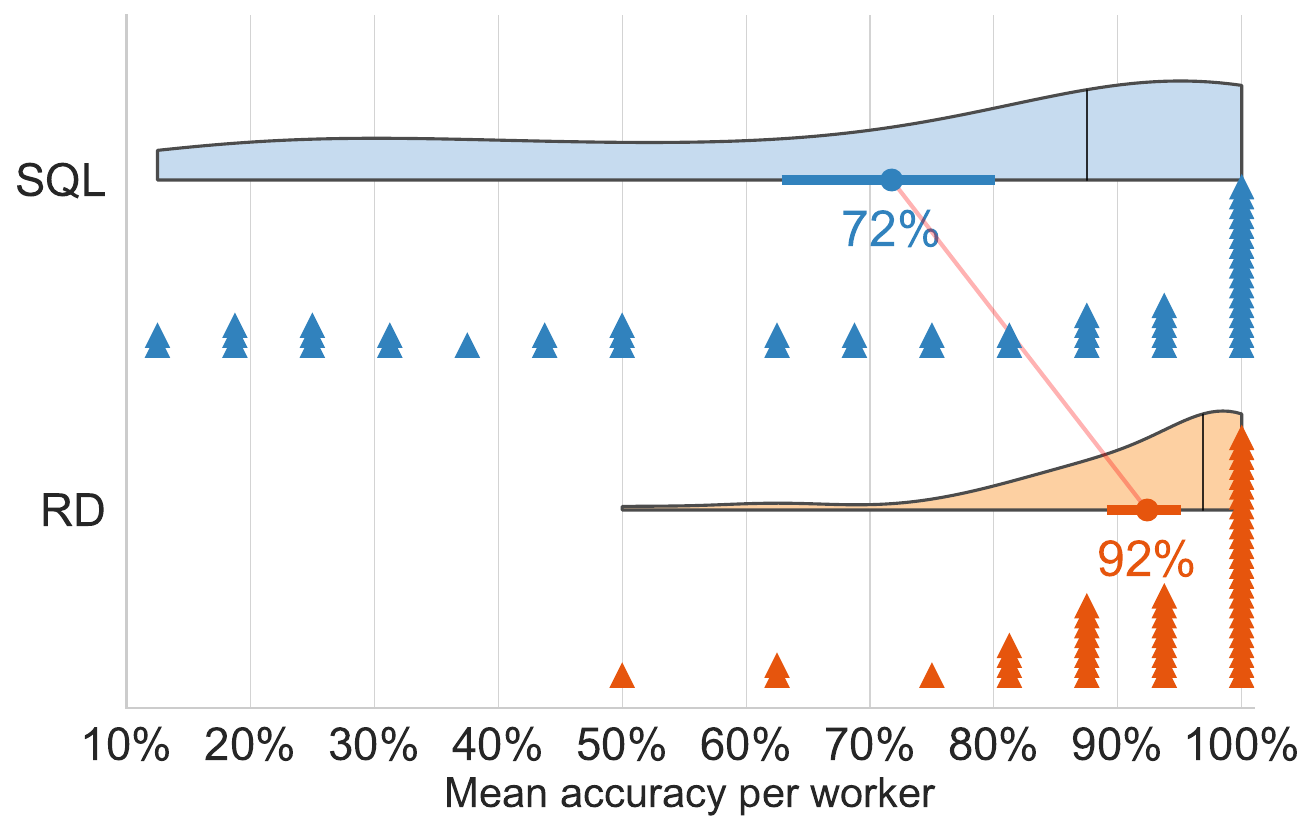}
    \includegraphics[scale=0.3]{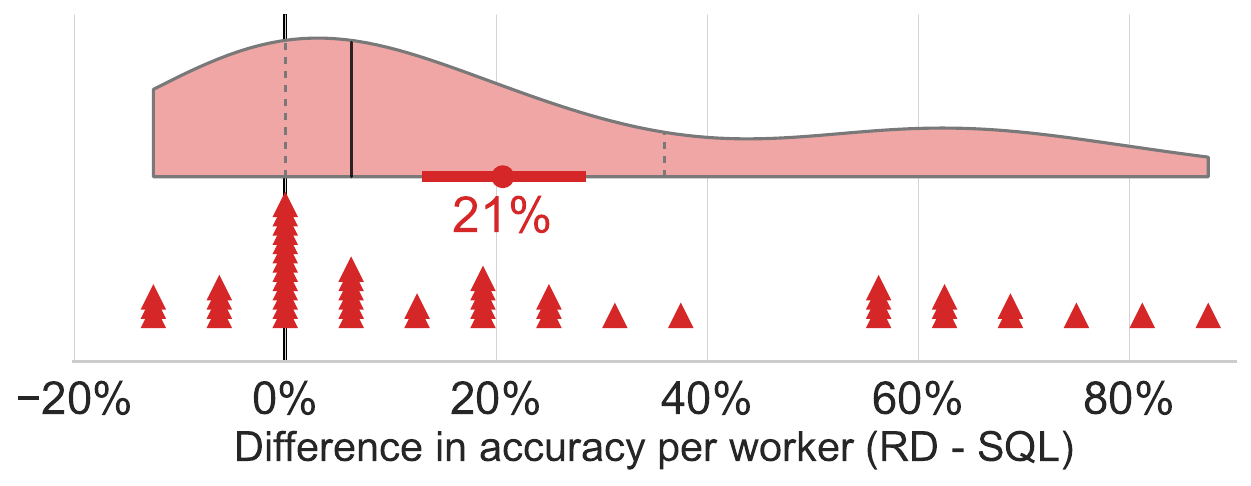}	
	\vspace{-2mm}	
    \caption{Result 3. (Accuracy)}
    \label{q4_figure_variant1}
\end{subfigure}
\begin{subfigure}[b]{0.7\linewidth}	
    \centering	
	\vspace{1mm}	
    \includegraphics[scale=0.31]{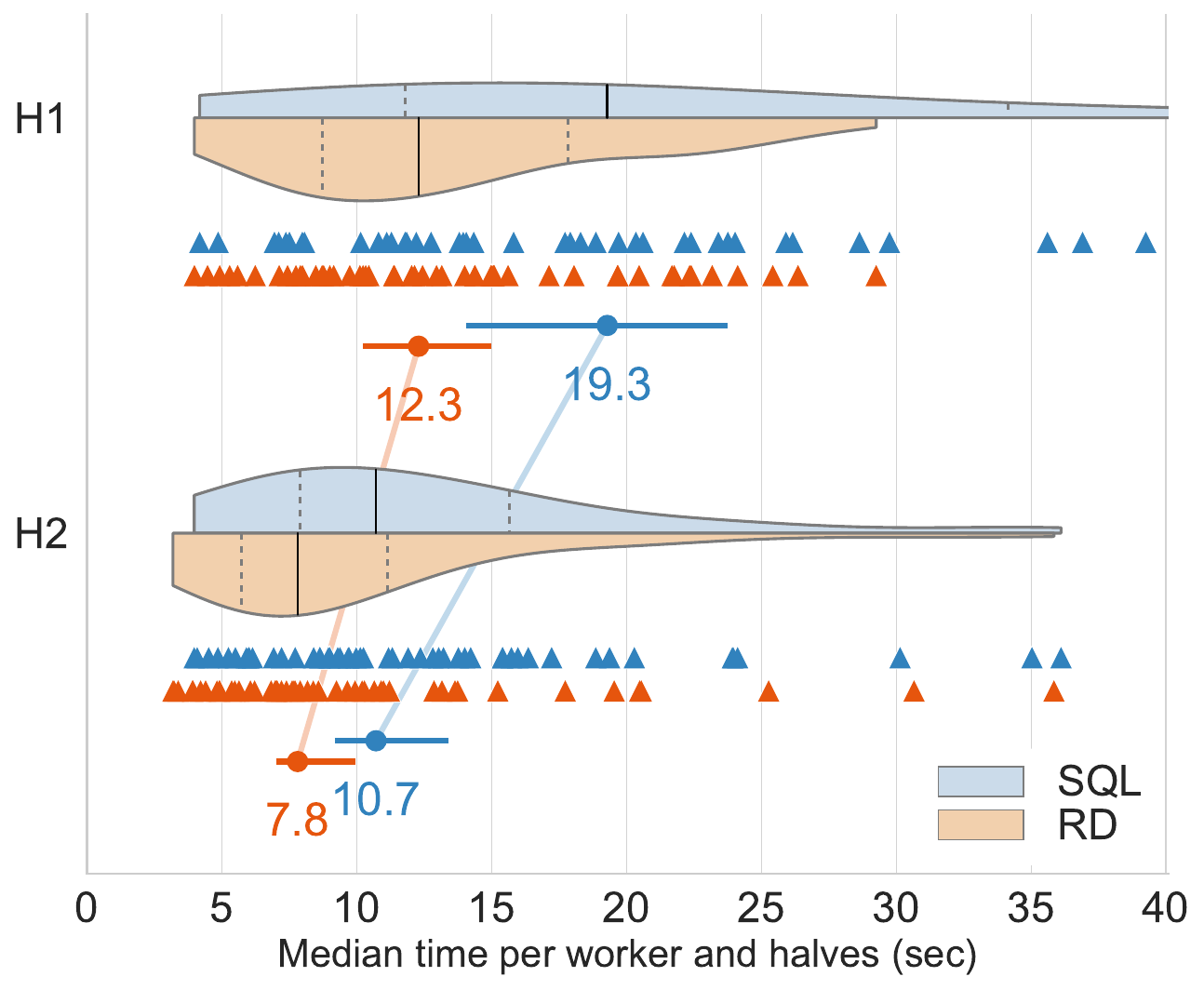}
	\vspace{-2mm}
    \caption{Result 2. (Learning)}
    \label{q2_figure_variant1}
\end{subfigure}
\caption{User study:
Triangles show median times per condition or mean accuracy for each of the $n=50$ successful participants.
Violin plots show the data distribution, the median with a solid line, and the 25\% and 75\% quantiles with dashed lines.
Error bars show the 95\% BCa bootstrapped confidence intervals (CI) around the mean or median.
Lines connect related marks.
\diagram\ is abbreviated here by RD.}
\label{fig:studyResults}
\end{figure}

\introparagraph{Participants}
We conducted an $n=13$ pilot study in the lab.
After registering our study on OSF~\cite{relationaldiagrams:links}
and receiving approval from our Institutional Review Board (IRB), we began recruiting participants on MTurk~\cite{AMT}. 
Participants needed to have at least 500 completed tasks approved by requesters and at least 97\% of their completed tasks approved.
For us to approve their task, they needed to have at least 50\% accuracy (i.e.\ answer at least 16 of our 32 questions correctly).
Thus a participant who answered every question in \SQL correctly and every question in \diagrams\ incorrectly (or v.v.)
would be included.
Among the 120 task submissions, only 58  were approved.
Our preregistration specified 50 participants, so we dropped 8 records and kept the counterbalancing intact by using the data from the first 25 participants who started in each condition.

\introparagraph{Tutorial}
Participants were given a self-paced 8-page tutorial on \diagrams. The tutorial introduced our basic visual notations by showing SQL examples and their diagrams.
The mean (resp.\ median) time spent on the tutorial and consent form was approximately 6.33 (respectively 3.5) minutes. 
The tutorial is available in our supplemental material \cite{relationaldiagrams:links}.

\introparagraph{Analysis}
(1) As a within-subjects study, we first determined the \emph{per-participant median time} for each condition.
We used the median (instead of mean) for time because it is robust to outliers, although the median often requires more participants for the same statistical power.
We computed the relative time $\diagrams / \SQL$ needed per participant and again calculated the \emph{median across} all participants.
Here we used the median of the ratios as the median creates an unbiased estimator (in contrast to the mean of ratios, 
see \iflabelexists{appendix:controlled-experiment}
{\cref{appendix:controlled-experiment}}
{the online appendix \cite{relationaldiagrams:links}} for details).
(2) Likewise, we computed the median per-participant time for each condition 
spent on the 1$^\textrm{st}$ half (16 questions), on the 2$^\textrm{nd}$ half (16 questions),
and the median ratio of the $2^\textrm{nd}/1^\textrm{st}$ times.
(3) We also computed the per-participant 
accuracy
for each condition and their difference.
Then, across all participants, we calculated the mean of the differences in accuracy.
Here we used the mean (instead of median) since the values are bounded within $[0,1]$ (i.e.\ there are no outliers)
and mean is more appropriate for discrete values like accuracies (i.e.\ $16/16$, $15/16$, ...).
We analyzed these mean/median effect sizes \cite{cumming2013understanding, dragicevik2018CanWeCallMean} and used bias-corrected and accelerated (BCa) 95\% confidence intervals (CIs) to show their range of plausible values \cite{efron1987better, Dragicevic2016FairStatistical}.

\introparagraph{Results}
We summarize 3 key takeaways.
The executed Python notebook has more details \cite{relationaldiagrams:links}.

\resultbox{\begin{resultW}(\textbf{Speed}) 
We have strong evidence that
participants were meaningfully faster at identifying patterns using \diagrams\ than $\SQL$:
median ratio $\diagrams /\SQL = $ 0.70, 95\% CI [0.63, 0.77].
\end{resultW}}

\noindent
Our choice of visualization is a variant of Raincloud plots~\cite{Allen:2019:Raincloudplots}
and is inspired from recent work in the visualization literature~\cite{Correll:2023:RaincloudPlots} 
discussing various ways to juxtapose multiple visualizations (``clouds + rain + lightning'') in the same chart for increasing information content.
In that framework, each of our charts consists of
($i$) \emph{density plots} that show an overview of the shape of the distribution (the ``cloud''),
($ii$) unjittered \emph{dot plots} that show the raw data (the ``rain'':  here we deviate from \cite{Correll:2023:RaincloudPlots} in using triangles instead of circles which, in our opinion, are more easily countable due to their visible vertices), and
($iii$) \emph{95\% confidence intervals} that provide summary statistics (the ``lightning'').
Furthermore, whenever we compare alternative modalities (``repeated measures''), 
we also use (4) \emph{paired plots} with lines connecting summary statistics and/or raw data.

\hyperref[{q1_figure1_variant1}]{Figure~\ref*{q1_figure1_variant1} (top)}
uses a paired plot to show the individual median per-participant times (and overall median across participants together with confidence interval)
for both \SQL (13.61, 95\% CI [12.37, 16.43] in blue on the top) 
and \diagrams\ (10.11 95\% CI [8.38, 11.26] in orange on the bottom). 
\hyperref[{q1_figure1_variant1}]{Fig.~\ref*{q1_figure1_variant1} (bottom)}
shows the per-participant ratios between median times.
Notice that the 95\% confidence interval of the overall median [0.63, 0.77] does not overlap 1.00,
which gives strong evidence for our conclusions.

\resultbox{\begin{resultW}(\textbf{Learning}) 
Participants got meaningfully faster during the study in both conditions.
\end{resultW}}

\noindent
\Cref{q2_figure_variant1} shows the individual times for H1 ($1^\textrm{st}$ half, i.e.\ the first 16 questions) 
and H2 ($2^\textrm{nd}$ half, i.e.\ the last 16 questions),
for both \SQL (in blue on the top) and \diagrams\ (in orange on the bottom),
together with medians and CIs.
We see that the overall trend (\diagrams\ being faster than \SQL) is repeated across both halfs,
and additionally that learning is taking place (participants need less time in H2 than in H1 in both conditions).
The median ratios H1$/$H2 we used for inference (not shown but in our supplemental material) are
0.71, 95\% CI [0.63, 0.79]
for \diagrams, and
0.70, 95\% CI [0.51, 0.79]
for $\SQL$.

\resultbox{\begin{resultW}(\textbf{Accuracy}) 
Participants were considerably more often correct with \diagrams\ than with $\SQL$: 
mean difference in accuracy $\diagrams-\SQL = $ 21\%, 95\% CI [13\%, 29\%].
\end{resultW}}

\noindent
\hyperref[{q4_figure_variant1}]{Figure~\ref*{q4_figure_variant1} (top)}
shows that the per-participant accuracies and the overall mean accuracies were meaningfully higher 
with \diagrams\ than with \SQL.
Notice that each participant answered 16 questions in each condition, thus possible scores are 
$16/16=1$, $15/16\approx 0.94$, etc.
Thus accuracy per user and modality is discretized in multiples of $1/16$ (in contrast to completion time, which is a continuous value and differs, even if slightly, between any two users).
We thus use stacked triangles akin to a Wilkinson dot plot~\cite{Wilkinson:1999:Dotplot}
to avoid overplotting and show individual data points.
\hyperref[{q4_figure_variant1}]{Figure~\ref*{q4_figure_variant1} (bottom)}
shows the per-participant difference in mean accuracy.
As the 95\% CI of the overall mean [13\%, 29\%] does not overlap 0, we have strong 
evidence for our conclusion.

\introparagraph{Participant comments} 
Participants could optionally write feedback at the end of the study. 
Few participants did, but those who did were encouraging, such as,
``I found your diagrams very helpful in understanding the queries. At first I didn't get it, but after staring at the diagrams for a few minutes it clicked and everything became super simple. I saw the patterns and it became just looking for the correct pattern to know which query was being used.''

\section{Related Work}
\label{sec:relatedWork}
\label{SEC:RELATEDWORK}

\subsection{Peirce's beta Existential Graphs}
\label{sec:Peirce}

\diagrams\ represent nested quantifiers
similarly as the influential and widely-studied \emph{Existential Graphs} 
by Charles Sanders Peirce~\cite{peirce:1933,Roberts:1992,Shin:2002} for expressing logical statements
(i.e.\ Boolean queries).
Peirce's graphs come in two variants called alpha and beta.
Alpha graphs represent propositional logic, and beta graphs represent first-order logic (FOL). 
Both variants use so-called \emph{cuts} to express negation (similar to our negation boxes),
and beta graphs use a syntactical element called the \emph{Line of Identity} (LI) to denote 
both \emph{the existence of objects} and \emph{the identity between objects}.

\introparagraph{Differences}
The four key differences of beta graphs vs.\ \diagrams\ are:
(1) beta graphs can only represent sentences and not queries;
(2) beta graphs cannot represent constants, thus selections cannot be modeled and instead require dedicated predicates;
(3) beta graphs can only represent identity predicates (and no comparisons); and
(4) Lines of Identity (LIs) in beta graphs have multiple meanings (existential quantification and identity between objects)
and are a primary symbol.\footnote{Every beta graph has lines, and graphs with lines but no predicates have meanings.
See, e.g., the 
definition in \cite[p.\ 41]{Shin:2002}.}
This \emph{function overload of LIs} can make reading the graphs ambiguous.
We, in contrast, have predicates inspired by $\TRC$. Lines only connect two attributes and have no loose ends. 
Interpreting a graph as a $\TRC$ formula is straightforward and can be summarized in a simple set of rules (recall \cref{sec:QV}).
We discuss this important conceptual difference in detail in \iflabelexists{appendix:beginning}
{\cref{app:PeirceDetails}.}
{the online appendix~\cite{relationaldiagrams:links}.}

\subsection{\texorpdfstring{$\queryviz$}{\queryviz}}%
\label{sec:queryvis1_vs_2}

Some of our design decisions are similar
to an earlier query representation called 
$\queryviz$ \cite{DanaparamitaG2011:QueryViz,Leventidis2020QueryVis,gatterbauer2011databases}.
In $\queryviz$ diagrams, grouping boxes are used to group all tables within a local scope, 
i.e., \emph{for each individual query block}. 
Those boxes thus cannot show their respective nesting, and an additional symbol of directed arrows is needed to ``encode'' the nesting.
The high-level consequence of those design decisions is that 
(1) $\queryviz$ does not guarantee to unambiguously visualize nested queries with nesting depth $\geq 4$ 
(please see
\iflabelexists{appendix:queryvis}
{\cref{appendix:queryvis}.}
{our online appendix \cite{relationaldiagrams:links}}
for a minimum example),
(2) each grouping box needs to contain at least one relation 
(thus $\queryviz$ cannot represent the query in \cref{Fig_TRC_vs_RD}),
and (3) $\queryviz$ cannot represent general Boolean sentences
(e.g.\ the sentence ``All sailors have reserved some red boat'').
Thus $\queryviz$ is not sound and not relationally complete, 
even for the disjunctive fragment.

\subsection{Other relationally-complete formalisms}
\label{sec:DFQL}

\iflabelexists{appendix:DFQL}
{\cref{appendix:DFQL}}
{The online appendix \cite{relationaldiagrams:links}}
compares \diagrams\ to other related visualizations like
DFQL (Dataflow Query Language) \cite{DBLP:journals/iam/ClarkW94,DBLP:journals/vlc/CatarciCLB97}.
On a high level, all visual formalisms
that we are aware of and that were proven to be relationally complete
(including those listed in
\cite{DBLP:journals/vlc/CatarciCLB97})
are at their core visualizations of relational algebra operators.
This applies even to the more abstract \emph{graph data structures (GDS)} from \cite{DBLP:conf/vdb/Catarci91} and the later \emph{graph model (GM)} from \cite{DBLP:journals/is/CatarciSA93}, which are related to our concept of \emph{query representation}.
The key difference is that GDS and GM are formulated inductively based on mappings onto operators of relational algebra. 
They thus mirror dataflow-type languages where visual symbols (directed hyperedges) represent operators like \emph{set difference} connecting two relational symbols, leading to a new third symbol as output.
Even QBE~\cite{DBLP:journals/ibmsj/Zloof77}
uses the query pattern from $\RA$ and $\DatalogN$ of implementing relational division (or universal quantification)
in a dataflow-type, sequential manner.
Similarly, SIEUFERD~\cite{DBLP:conf/sigmod/BakkeK16}, a direct manipulation spreadsheet-like interface, 
uses direct translation of relational algebra operators to prove SQL-92 completeness. 
This translation involves expressing set difference with outer joins and ``IS NULL'' conditions. 
We have proved that there are simple queries in relational calculus 
that cannot be represented in relational algebra with the same number of relational symbols. 
Thus any visual formalism based on relational algebra 
cannot represent the full range of relational query patterns.

\subsection{Other diagrammatic and non-diagrammatic query representations}
\label{sec:otherrelatedwork}

Visual SQL~\cite{DBLP:conf/er/JaakkolaT03} is 
a visual query language that also supports query visualization. 
With its focus on query specification, it maintains the one-to-one correspondence to SQL,
and syntactic variants of the same query lead to different representations.
SQLVis~\cite{DBLP:conf/vl/MiedemaF21} shares motivation with \queryviz. 
Similar to Visual SQL, it places a stronger focus on the actual syntax of a SQL query 
and syntactic variants like nested EXISTS queries change the visualization,
and join conditions are expressed as text.
StreamTrace \cite{DBLP:conf/chi/BattleFDBCG16}
focuses on visualizing temporal queries
with workflow diagrams and a timeline.
It is an example of visualizations for spatiotemporal domains and not the logic behind general relational queries.
DataPlay~\cite{DBLP:journals/pvldb/AbouziedHS12,DBLP:conf/uist/AbouziedHS12} 
allows a user to specify their query by interactively modifying a query tree with quantifiers and observing changes in the matching/non-matching data.
It does not have a union operator and is thus not relationally complete.
For a more detailed discussion we refer to two recent tutorials on visual representations of relational queries~\cite{DBLP:journals/pvldb/Gatterbauer23,ICDE:2024:diagrammatic:tutorial}.

\section{Conclusions and Future Work}

We motivated a criterion called \emph{pattern-isomorphism} that captures the patterns across relational languages 
and gave evidence for its importance in designing diagrammatic representations.
We formulated the non-disjunctive fragments of $\DatalogN$, $\RA$, safe $\TRC$, and corresponding $\SQL$ (interpreted under set semantics)
that naturally generalize conjunctive queries to nested queries with negation.
We prove that this important fragment 
allows a rather intuitive and, in hindsight, natural diagrammatic representation that can preserve the query pattern used across all four languages.
We further prove that 
this formalism, extended with a representation of union, is complete for full safe relational calculus (though not pattern-preserving)
and showed via user studies strong evidence that this diagrammatic representation allows users to understand query patterns faster and more accurately than SQL, even with minimal training.

Finding a \emph{pattern-preserving} diagrammatic representation for disjunction 
and even more general features of SQL 
(such as grouping and aggregates) 
is an open problem. 
For example, it is not clear how to achieve an intuitive and principled \emph{diagrammatic} representation for arbitrary nestings of disjunctions, such as
``$R.A \!<\! S.E \wedge (R.B \!<\! S.F \vee R.C \!<\! S.G)$''
or 
``$(R.A \!>\! 0 \wedge R.A \!<\! 10) \vee (R.A \!>\! 20 \wedge R.A \!<\! 30)$'' with minimal additional notations.
Grounded in a long history of diagrammatic representations of logic,
we gave intuitive arguments for why visualizing disjunctions is inherently more difficult than conjunctions,
with some experts believing it is not possible~\cite{shin_1995,Shin:2002} 
unless one adds \emph{non-diagrammatic abstractions}.

\begin{acks}
This work was supported in part by
a Khoury seed grant program,
and the National Science Foundation (NSF) under award numbers IIS-1762268, IIS-2145382,
and IIS-1956096.
It was conducted in part while WG was visiting the Simons Institute for the Theory of Computing.
We like to thank
Mirek Riedewald for helpful comments on an early version of this paper,
and Jan Van den Bussche for insightful comments about an earlier version of the proof of separation lemma 20.
\end{acks}

\clearpage

\bibliographystyle{ACM-Reference-Format}

\bibliography{queryvis.bib}

\clearpage
\appendix
\section{Nomenclature}
\label{appendix:beginning}

\begin{table}[h]
\centering
\small
\begin{tabularx}{\linewidth}{@{\hspace{0pt}} >{$}l<{$} @{\hspace{2mm}}X@{}} %
\hline
\textrm{Symbol}		& Definition 	\\
\hline
    \hline
	\S				& \emph{Table signature} or ordered list of its table references of a query \\
	q\				& Dissociated query starting from query $q$ \\	
	\L_1 \subseteq^\rep \L_2					
					& Language $\L_2$ can \emph{pattern-represent} language $\L_1$, i.e., 
					$\L_2$ can represent all relational query patterns of language $\L_1$. \\
	\L_1 \not \subseteq^\rep \L_2					
					& Language $\L_2$ can not \emph{pattern-represent} language $\L_1$, i.e., 
					there are query patterns in language $\L_1$ that $\L_2$ cannot represent \\
	\L_1 \subsetneq^\rep \L_2
					& Language $\L_2$ \emph{pattern-dominates} language $\L_1$, i.e., 
					$\L_2$ can represent all relational query patterns of language $\L_1$
					and $\L_2$ can represent relational query patterns that $\L_1$ cannot.\\
	\L_1 \equiv^\rep \L_2
					& Languages $\L_1$ and $\L_2$ are \emph{representation equivalent}, i.e.,
					they can express the identical set of relational query patterns.\\
\hline
\end{tabularx}
\end{table}

\section{Additional example for \texorpdfstring{\cref{SEC:INTRODUCTION}}
{Section \ref{SEC:INTRODUCTION}}: Limits of Relational Algebra}

We give an alternative minimal example that illustrates the limits of relational algebra in expressing even simple relational query patterns.

\begin{example}[RA vs.\ Datalog]
\label{ex:intro}
Consider the $\DatalogN$ (non-recursive Datalog with negation) query in \cref{Fig_RA_vs_Datalog_g}, which returns all tuples in $R(A,B)$ 
whose attribute $B$ does not appear in the unary table $S(B)$.
The query references each of the input tables $R$ and $S$ exactly once. 
As the proof of 
our later \cref{th:representations}
shows, basic Relational Algebra ($\RA$) cannot
express this query by referencing each of the tables $R$ and $S$ only once. 
\Cref{Fig_RA_vs_Datalog_b,Fig_RA_vs_Datalog_e} show
two logically-equivalent queries in $\RA$,
each of which references the table $R(A,B)$ twice.
We also added equivalent $\DatalogN$ queries, which for those two $\RA$ expressions use the same ``query pattern'' (a concept we will formalize later). 
Intuitively (and we prove this later more formally), $\DatalogN$ can express strictly more query patterns than RA;
it has a higher ``pattern-expressiveness'' despite having the same logical expressiveness.
We believe that any diagrammatic language for illustrating and reasoning about query patterns used in queries should be able to express the full range of possible patterns across existing
relational query languages (such as the one in \cref{Fig_RA_vs_Datalog_g}).
It follows that any diagrammatic representation of relational queries that relies on a one-to-one mapping with the operators of RA 
\emph{cannot} represent the full spectrum of query patterns of relational queries.
\end{example}

\cref{ex:intro} illustrated that 
query languages with equal expressiveness 
may not necessarily be equally able to express the same range of logical patterns---they are not \emph{representation-equivalent}.
In other words, there may be queries have have no patterns-preserving translations into the other language.

\begin{figure}[t]
    \centering
    \begin{subfigure}[b]{.32\linewidth}
        \centering
        \hspace{-6mm}\textbf{$\DatalogN$}
    \end{subfigure}
    \hspace{-0.6mm}
    \begin{subfigure}[b]{.24\linewidth}
        \centering
        \hspace{2mm}\textbf{Relational}\\
        \hspace{2mm}\textbf{Algebra (RA)}
    \end{subfigure}
    \hspace{0mm}
    \begin{subfigure}[b]{.33\linewidth}
        \centering
        \textbf{\diagrams}
    \end{subfigure}
    \begin{subfigure}[b]{.32\linewidth}
    	\begin{align*}
    		I(x,y)	& \datarule R(x,\_),  S(y). 		\\[-1mm]
    		Q(x,y) 	& \datarule R(x,y), \neg I(x,y).
    	\end{align*}
    	\vspace{-4mm}
        \caption{}
    	\label{Fig_RA_vs_Datalog_a}
    \end{subfigure}
    \hspace{-0.6mm}
    \begin{subfigure}[b]{.24\linewidth}
    	\begin{align*}
    		R - \big( \pi_A R \times S \big)
    	\end{align*}
    	\vspace{-4mm}
        \caption{}
    	\label{Fig_RA_vs_Datalog_b}
    \end{subfigure}
    \hspace{0mm}
    \begin{subfigure}[b]{.33\linewidth}
    	\vspace{1mm}
        \includegraphics[scale=0.36]{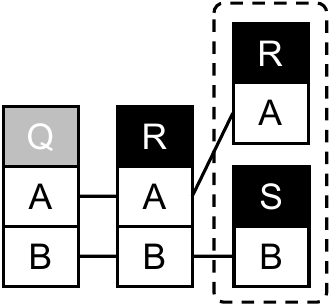}
    	\vspace{-2mm}
        \caption{}
    	\label{Fig_RA_vs_Datalog_c}
    \end{subfigure}
    \begin{subfigure}[b]{.32\linewidth}
    	\begin{align*}
    		I(y)	& \datarule R(\_,y), \neg S(y). 		\\[-1mm]
    		Q(x,y) 	& \datarule R(x,y), I(y).
    	\end{align*}
    	\vspace{-4mm}
        \caption{}
    	\label{Fig_RA_vs_Datalog_d}
    \end{subfigure}
    \hspace{2mm}
    \begin{subfigure}[b]{.24\linewidth}
    	\begin{align*}
    		R \Join \big( \pi_B R - S \big)
    	\end{align*}
    	\vspace{-4mm}
        \caption{}
    	\label{Fig_RA_vs_Datalog_e}
    \end{subfigure}
    \hspace{0mm}
    \begin{subfigure}[b]{.33\linewidth}
        \includegraphics[scale=0.36]{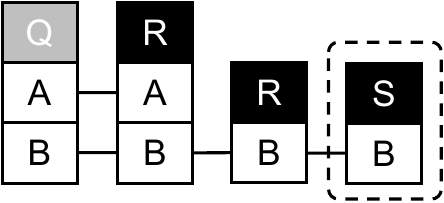}
    	\vspace{-2mm}
        \caption{}
    	\label{Fig_RA_vs_Datalog_f}
    \end{subfigure}
    \begin{subfigure}[b]{.32\linewidth}
    	\begin{align*}
    		Q(x,y) 	& \datarule R(x,y), \neg S(y).
    	\end{align*}
    	\vspace{-4mm}
        \caption{}
    	\label{Fig_RA_vs_Datalog_g}
    \end{subfigure}
    \hspace{1mm}
    \begin{subfigure}[b]{.26\linewidth}
    \centering
    	\h{\Huge\frownie{}}
    	\vspace{1mm}
        \caption{}
    	\label{Fig_RA_vs_Datalog_h}
    \end{subfigure}
    \hspace{-0.6mm}
    \begin{subfigure}[b]{.33\linewidth}
        \includegraphics[scale=0.36]{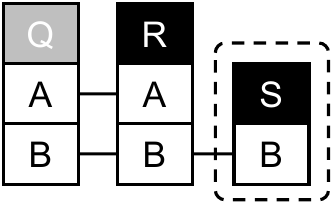}
    	\vspace{-2mm}
        \caption{}
    	\label{Fig_RA_vs_Datalog_i}
    \end{subfigure}
    \caption{
    Given tables R(A,B) and S(B). 
	The first column (\subref{Fig_RA_vs_Datalog_a}, \subref{Fig_RA_vs_Datalog_d}, \subref{Fig_RA_vs_Datalog_g}) shows three logically-equivalent $\DatalogN$ queries that use three different ``query patterns.'' 
    The first two (\subref{Fig_RA_vs_Datalog_a}, \subref{Fig_RA_vs_Datalog_d}) can also be expressed in Relational Algebra (RA) (\subref{Fig_RA_vs_Datalog_b}, \subref{Fig_RA_vs_Datalog_e}), whereas the third pattern (\subref{Fig_RA_vs_Datalog_g})
    cannot be expressed in RA,
    i.e.\ it is not possible
    to write a logically-equivalent query in basic RA that references only one occurrence of $R$ and $S$ each.
	Our paper makes these notions of query patterns precise.
    The third column (\subref{Fig_RA_vs_Datalog_c}, \subref{Fig_RA_vs_Datalog_f}, \subref{Fig_RA_vs_Datalog_i}) shows \diagrams\ that 
    use the same ``query patterns''
    as the $\DatalogN$ queries in the first column.
    To the best of our knowledge, our proposed \diagrams\ are the first diagrammatic representation of relational queries that are 
	($i$) unambiguous, 
	($ii$) relationally-complete and 
	($iii$) able to represent the full range of query patterns for union of non-disjunctive relational calculus.
    }
    \label{Fig_RA_vs_Datalog}
\end{figure}

\section{Proof For \texorpdfstring{\cref{sec:nondisjunctivefragment}, \cref{th:equivalence}}
{Section \ref{sec:nondisjunctivefragment}, Theorem \ref{th:equivalence}}}%
\label{appendix:proofoflogicalexpressivness}

\begin{proof}[Proof of \cref{th:equivalence}]
We prove each of the directions in turn: 

\centerline{
\begin{tikzpicture}[>=stealth, shorten >= 2pt, shorten <= 2pt, line width=0.25mm]
  \node (RA) {$\NDRA$};
  \node[right=5mm of RA] (Datalog) {$\DatalogND$};
  \node[right=5mm of Datalog] (TRC) {$\NDTRC$};
  \node[right=5mm of TRC] (SQL) {$\NDSQL$};  
    
  \draw[->] (RA.north) to[bend left] 
  	node[midway,above] {\hyperref[th:equivalence:ra-datalog]{(1)}} 
	(Datalog.north);
  \draw[->] (Datalog.south) to[bend left] 
  	node[midway,below] {\hyperref[th:equivalence:datalog-ra]{(2)}} 
	(RA.south);
  \draw[->] (Datalog.north) to[bend left] 
  	node[midway,above] {\hyperref[th:equivalence:datalog-trc]{(3)}} 
	(TRC.north);
  \draw[->] (TRC.south) to[bend left] 
  	node[midway,below] {\hyperref[th:equivalence:trc-datalog]{(4)}} 
	(Datalog.south);
  \draw[<->] (TRC.north) to[bend left] node[midway,above]
  	{\hyperref[th:equivalence:trc-sql]{(5)}}
	(SQL.north);
\end{tikzpicture}
}

\begin{pfparts}
\item[\namedlabel{th:equivalence:ra-datalog}{\underline{(1) $\NDRA \rightarrow \DatalogND$}}:]
    The proof for this direction is an easy induction on the size of the algebraic expression.
    It is a minor adaptation of the translation from $\RA$ to $\DatalogN$ proposed by Ullman~\cite{Ullman1988PrinceplesOfDatabase},
    yet it also pays attention to the restricted fragment 
    and keeps the numbers of atoms constant during the translation.
    Formally, we show that if an $\NDRA$ expression has $i$ occurrences of operators, 
    then there is a $\DatalogND$ program that produces the value of the expression 
    as the relation for one of its Intensional Database (IDB) predicates. 
    
    The basis is $i = 0$, that is, a single operand. 
    If this operand is a given relation $R$, then $R$ is an Extensional Database (EDB) relation and thus ``available'' without the need for any rules. 
    For the induction, consider an expression whose outermost operator is one of 5 operators: 
    Cartesian product $\times$, selection $\sigma$, theta join $\Join_{c}$, 
    projection $\pi$, or difference $-$.    
	Notice that union $\cup$ is missing. 
	Also, the rename operator is trivial since $\DatalogN$ uses position information instead of names.
    
    Case 1: $Q = E_1 \times E_2$:
    Let $\NDRA$ expressions $E_1$ and $E_2$ have $\DatalogND$ predicates $e_1$ and $e_2$ whose rules define their relations, and assume their relations are of arities $d$ and $m$, respectively. 
    Then define $q$, the predicate for $Q$, by:
    \begin{align*}
    	q(x_1, \ldots , x_{d + m}) \datarule e_1(x_1, \ldots , x_d),  e_2(x_{d + 1}, \ldots , x_{d + m}).
    \end{align*}

    Case 2: $Q = \sigma_c E$:
    By restricting our language from $\RA$ to $\NDRA$, we only allow selections $\sigma_c(\varphi)$ where the condition $c$ is
    a conjunction of simple selections $c = c_1 \wedge c_2 \wedge \cdots$, i.e.,
    each selection $c_i$ is of the form 
    $\sigma_{A_{i1} \theta A_{i2}}$ (join predicate) or
    $\sigma_{A_{i1} \theta v}$ (selection predicate).
    Let $e$ be a $\DatalogND$ predicate whose relation is the same as the relation for $E$, 
    and suppose $e$ has arity $d$. Then the rule for $Q$ is:
    \begin{align*}
    	q(x_1, \cdots , x_d) \datarule e_1(x_1, \cdots , x_d), c_{\theta}.
    \end{align*}
    where $c_{\theta}$ is a conjunction of join predicates $x_i \theta x_j$
    and selection predicates $x_i \theta v$
    where $x_i$ and $x_j$ index the attributes
	$A_i$ and $A_j$.

    Case 3: $Q = E_1 \Join_{c} E_2$:
    While the join operator is not a basic operator of relational algebra, 
    built-in predicates are, in practice, commonly expressed directly through join conditions.
    This case follows immediately from cases 1 and 2, and the definition of joins as
    $Q = E_1 \Join_{c} E_2 = \sigma_c (E_1 \times E_2)$.

    Case 4: $Q = \pi_{A_{i1}, \ldots, A_{id}}(E)$:
    Let $E$ 's relation have arity $m$, and let $e$ be a predicate of arity $m$ whose rules produce the relation for $E$. 
    Then the rule for the predicate $q$ corresponding to expression $Q$ is:
    \begin{align*}
    	q(x_{i_1}, \ldots, x_{i_d}) \datarule e(x_1, \ldots , x_m).
    \end{align*}

    Case 5: $Q = E_1 - E_2$: 
    We know by definition of the set difference that $E_1$ and $E_2$ must have the same arities. Assume those to be $d$,
    and that there are predicates $e_1$ and $e_2$ whose rules define their relations to be the same as the relations for $E_1$ and $E_2$, respectively.
    Then we use rule:
    \begin{align*}
    	q(x_1, \cdots , x_d) \datarule e_1(x_1,  \ldots , x_d),  \neg e_2(x_1, \ldots, x_d).
    \end{align*}
    to define a predicate $q$ whose relation is the same as the relation for $Q$. 
    We can easily check that safety for this rule is fulfilled as all variables appearing in the negated $e_2$ also appear in the positive $e_1$.

\item[\namedlabel{th:equivalence:datalog-ra}{\underline{(2) $\DatalogND \rightarrow \NDRA$}}:]
Typical textbook translations from $\DatalogN$ to $\RA$,
such as the one by Ullman \cite[Th 3.8, Alg 3.2, Alg 3.6]{Ullman1988PrinceplesOfDatabase}
need to compute the active domain by projecting all EDB relations onto each of their components 
and then taking the union of these projections and the set of constants appearing in the rules. 
Since $\NDRA$ lacks the union operator, 
we cannot create the active domain from a union of all constants used in the database.    

Let $\DatalogND$ program $\mathcal{P}$ be a collection of safe, nonrecursive Datalog rules, possibly with negated subgoals. 
By the safety condition, every variable that appears anywhere in the rule must appear in some nonnegated, relational subgoal of the body,
or must be bound by an equality (or a sequence of equalities) to a variable of such an ordinary predicate or to a constant~\cite{DBLP:journals/tkde/CeriGT89}.
From the definition of $\DatalogND$, each IDB appears in exactly one rule as the head.
Then, for each IDB predicate $q$ of $\mathcal{P}$, there is an expression $Q$ of relational algebra 
that computes the relation for $q$.
Since $\mathcal{P}$ is nonrecursive, we can order the predicates according to a topological sort of the dependency graph; 
that is, if $q_1$ appears as a subgoal in a rule for $q_2$, then $q_1$ precedes $q_2$ in the order. 

If a rule contains built-in predicates in the body (join predicates $x_i \theta x_j$ or selection predicates $x_i \theta v$), 
the translation first focuses on the body without predicates and then applies a selection $\sigma_c$ where the selection condition $c$
is a conjunction of the built-in predicates.

To express negated subgoals in the body, we need to use the set difference, and this requires us to complement negated subgoals with additional attributes.
Concretely, take a general Datalog rule with built-in predicates:
\begin{align*}
	q(\vec x) \datarule p_1(\vec x_1), \ldots, p_k(\vec x_k),  \neg n_1(\vec y_1), \ldots,  \neg n_m(\vec y_m), c_{\theta}.
\end{align*}
From the safety conditions of this rule, we know that
all variables in the built-in predicates $c_{\theta}$ need to appear in positive atoms.
Similarly, all variables in negated atoms also need to appear in positive atoms:
$\bigcup_1^k \vec x_i \supseteq \bigcup_1^m \vec y_i$. 
Let $\vec z$ be the set of complementing attributes, i.e., the attributes that only appear in positive atoms:
$\vec z = \bigcup_1^k \vec x_i - \bigcup_1^m \vec y_i$.

Let $P_i$ and $N_i$ be the $\NDRA$ expressions corresponding to $\DatalogND$ predicates $p_i$ and $n_i$.
If $\vec z = \emptyset$, then define $Q'$ as
\begin{align}
	&(P_1 \Join \ldots \Join P_k) - (N_1 \Join \ldots \Join N_m) 	\notag \\
\intertext{Otherwise define $Q'$ as} 
	&(P_1 \Join \ldots \Join P_k) - \big((N_1 \Join \ldots \Join N_m) \times \pi_{\vec z}(P_1 \Join \ldots \Join P_k) \big)
	\label{eq:translation_Datalog_RA_problem}
\end{align}
Finally define $Q = \pi_{\vec A}(\sigma_{\theta} (Q'))$ with 
$\vec A = ({A_{i1}, \ldots, A_{id}})$ representing the set of attributes indexed by $\vec x$.
Then this expression $Q$ translates one single rule with built-in predicates into a valid relational algebra expression $Q$ without union or disjunction.
Since every IDB predicate $q$ appears in only one rule, we do not need union nor disjunctions 
even if multiple rules are translated.
Since all variables in the built-in predicate $c_{\theta}$ need to appear in $Q'$, 
the selection $\sigma_{\theta}$ can be correctly applied on $Q'$.
It then follows by induction, on the order in which the IDB predicates are considered, that each has a relation defined by some expression 
in $\NDRA$.

\item[\namedlabel{th:equivalence:datalog-trc}{\underline{(3) $\DatalogND \rightarrow \NDTRC$}}:]
    
We consider a general Datalog rule:
\begin{align*}
	q(\vec x) \datarule p_1(\vec x_1), \ldots, p_k(\vec x_k), \neg n_1(\vec y_1), \ldots, \neg n_m(\vec y_m), c_{\theta}.
\end{align*}
Here $c_{\theta}$ is a conjunction of built-in predicates that adhere to the standard safety conditions \cite{DBLP:journals/tkde/CeriGT89}.
From those safety conditions, 
we know that
all variables in the built-in predicates $c_{\theta}$ need to appear in positive atoms,
and all variables in negated atoms also need to appear in positive atoms.

The rule then translates into a $\TRC$ fragment
\begin{align*}
	\{q(\vec A) \mid 
		&p_1 \in P_1, \ldots, p_k \in P_k[c_\textrm{out} \wedge c_{p} \\
		&\wedge \neg(\exists n_1\in N_1, \ldots, n_m \in N_m[c_\textrm{in}])]\}
\end{align*}
Here $\vec A$ is a set of attributes that correspond to the variables returned by the Datalog rule 
(from safety conditions, only attributes from the positive relations can be returned),
$c_\textrm{out}$ is a conjunction of equality predicates linking attributes from the output table $q$ to attributes from the input tables $P_i$,
$c_p$ is a conjunction of equality predicates and comparison predicates specified by $c_{\theta}$ 
between the positive relations or constants, and
$c_\textrm{in}$ is a conjunction of equality predicates between exactly one negative relation
and either a positive relation or a constant.
Notice that this translation guarantees that all predicates (including those in $c_\textrm{in}$) are \emph{guarded}.

\item[\namedlabel{th:equivalence:trc-datalog}{\underline{(4) $\NDTRC \rightarrow \DatalogND$}}:]
    In this translation, 
    we start from the canonical representation of $\NDTRC$
    (\cref{sec:TRC})
    where a set of existential quantifiers is always preceded by the negation operator 
    (except for the table variables at the root of the query).
    This implies that we can decompose any query in $\NDTRC$ 
    and write it as nested query components, each delimited by the scope of one negation operator.
    Each query component is then of the form:
    \begin{align*}
    	& \{ q(\vec A) \mid p_1 \in P_1, \ldots, p_k \in P_k[c_\textrm{out} \wedge c_{p} \\
    	&\hspace{10mm}	\wedge \neg q_1(\vec A_1)
    					\wedge \ldots \wedge \neg q_m(\vec A_m)] \}
    \end{align*}
    Here $\vec A$ is a set of attributes that correspond to the variables returned by the query
    (or, equivalently, variables that are passed to a nested query that determine whether that nested query is true or false),
    $c_\textrm{out}$ is a conjunction of equality predicates linking attributes from the output table $q$ 
    to attributes from the local tables $P_i$,
    $c_p$ is a conjunction of equality and comparison predicates between the positive relations or constants, and
    $\vec A_j$ are attributes from the input tables $P_1, \ldots, P_k$ used in the nested query $q_j$.

    Notice that for safe queries, only attributes from the positive relations can be returned,
    i.e., the output attributes need to be connected via equality predicates specified in $c_\textrm{out}$ 
	to attributes from $P_1, \ldots, P_k$.
    However, \emph{nested queries do not need to be safe}, and output attributes can be 
	($i$) connected directly to further nested queries, 
	or ($ii$) connected to input tables $P_1, \ldots, P_k$ via built-in instead of equality predicates.
    This last point is the main complication we need to take care of during the translation.
    We proceed in two steps:

    (1) First assume that each query is safe. 
    Then each subquery can be immediately translated into a separate rule by induction on the nesting hierarchy from the inside out.
    The basis of the induction is the leaf queries which can't contain nested queries and thus are of the form:
    \begin{align*}
    	& \{ q(\vec A) \mid p_1 \in P_1, \ldots, p_k \in P_k[c_\textrm{out} \wedge c_{p}] \}
    \end{align*}
    A leaf query is translated into
    \begin{align*}
    	q(\vec x) \datarule P_1(\vec x_1), \ldots, P_k(\vec x_k), c_{\theta}.
    \end{align*}
    where $\vec x$ are attributes chosen from the relations $P_i$ as specified in $c_\textrm{out}$,
    and $c_{\theta}$ is 
    the conjunction of comparison predicates between the positive relations $P_i$.
    
    For the induction step, assume that each nested $q_i(\vec A_i)$ is safe and translated into a rule $q_i$
    Then the safe query $q$
    \begin{align*}
    	& \{ q(\vec A) \mid p_1 \in P_1, \ldots, p_k \in P_k[c_\textrm{out} \wedge c_{p} \\
    	&\hspace{10mm}	\wedge \neg q_1(\vec A_1)
    					\wedge \ldots \wedge \neg q_m(\vec A_m)] \}
    \end{align*}
    is translated into a rule
    \begin{align*}
    	q(\vec x) \datarule P_1(\vec x_1), \ldots, P_k(\vec x_k), \neg N_1(\vec y_1), \ldots, \neg N_m(\vec y_m), c_{\theta}.
    \end{align*}
    where $\vec x$ are attributes chosen from the positive relations $P_i$ as specified in $c_\textrm{out}$,
    $c_{\theta}$ is a conjunction of comparison predicates between the positive relations $P_i$,
    and $\vec y_j$ are chosen from the variables 
	$\vec x = \bigcup_1^k \vec x_i$
	used in the positive relations $P_1, \ldots, P_k$.

    (2) Next assume that a nested query is valid yet not safe
    for either of the two previously-stated reasons:
    ($i$) either some $\neg q_j(\vec A_j)$ uses an attribute $q.A_{it}$ from the output $q(\vec A)$ directly;
    or ($ii$) some predicate in $c_\textrm{out}$ connects 
	an attribute $p_j.A_{is}$ 
	from $P_1, \ldots, P_k$ 
	to an output attribute $q.A_{it}$ 
	with a built-in predicate $p_j.A_{is} \theta q.A_{it}$ instead of an equality predicate.
    In both cases, we can make this query safe by first adding 
	\emph{an additional existentially-quantified table} $p_{k+1} \in P_{k+1}$ to the query
	that has an attribute $A_{i}'$ containing the domain of $q.A_{it}$.
	In our translation, we use the same table that is used in the outer query to connect to 
	and therefore to ``bound'' $q.A_{it}$ to the active domain.
	We then have to add appropriate predicates:
    for case ($i$), we add the equality predicate 
	$p_{k+1}.A_{i}' = q.A_{it}$ 
	to $c_\textrm{out}$ 
	and replace the attribute $q.A_{it}$ previously used in $q_j(\vec A_j)$ 
	with instead 
	$p_{k+1}.A_{i}'$;
    for case ($ii$), we also 
	add the equality predicate 
	$p_{k+1}.A_{i}' = q.A_{it}$
	to $c_\textrm{out}$ 
	and replace the previously used predicate 
	$p_j.A_{is} \theta q.A_{it}$
	with instead
	$p_j.A_{is} \theta p_{k+1}.A_{i}'$.
	After thus making the subquery safe, we can use the translation described earlier.
	
	It follows that every query in $\NDTRC$
	can be translated into a logically equivalent query in $\DatalogND$.

    We next illustrate both cases for when we need to add existentially-quantified tables in turn.

\begin{example}[All quantification in $\DatalogND$]
    We illustrate the translation for case ($i$) above with the following query:
    \begin{align*}
    	&\{ q(A) \mid \exists r \in R [q.A \equal r.A  \wedge \neg (\exists s \in S[	\\
    	&\hspace{10mm}	
    		\neg (\exists r_2 \in R 	
    		[r_2.B \equal s.B \wedge r_2.A \equal r.A])])] \}
    \end{align*}
	It represents relational division and is 
    shown as 
   	\diagrams, $\NDSQL$, $\NDRA$, and $\DatalogND$ in \cref{fig:SQL_equivalence,Fig_equivalence},
	and will also be re-used in \cref{ex:relationalDivision}.
    Based on our extended safety condition for $\NDTRC$
    \cref{def:anchor}, 
    all predicates are guarded, i.e.,
    they contain at least one attribute of a table that is existentially quantified 
    inside the same negation scope as that predicate.
    Those guarded attributes are: 
    {$r.A$}   in $r.A \equal q.A$,
    {$r_2.B$} in $r_2.B \equal s.B$, and
    {$r_2.A$} in $r_2.A \equal r.A$.
    
    Rewriting the query based on its recursive nested negation hierarchy gives us 3 query components:
    \begin{align*}
    	\{ q(A)   &\mid  \exists r \in R [q.A \equal r.A \wedge \neg (q_1(r.A))] \}	\\
    	\{ q_1(A) &\mid  \exists s \in S
    		[\h{\neg} (q_2(\h{q_1.A}, s.B) ) ] \} \\
    	\{q_2(A, B)  &\mid  \exists r_2 \in R 	
    		[q_2.B \equal r_2.B \wedge r_2.A \equal q_2.A] \}
    	\hspace{13mm}
    \end{align*}
	Now notice that $q_1$ is not safe because $q_1.A$ is used within the negated scope 
	$\neg (q_2(\h{q_1.A}, s.B))$
	without being existentially quantified (or 
	somehow ``bound'' to an element in the active domain)
	within $q_1$.
	In other words, the recursive call $q_2(q_1.A, s.B )$ 
	passes a predicate through the call hierarchy from $q_2$ directly to $q$ without being ``bound'' 
	while passing through $q_1$.

    We can make $q_1$ safe (or equivalently ``bound'' the predicate using $q_1.A$) 
    by adding another table $r_3 \in R$ in $q_1$ that accepts and hands over that attribute in the call hierarchy:
    \begin{align*}
    	\{ q(A)   &\mid  \exists r \in R [q.A \equal r.A \wedge \neg (q_1(r.A))] \}	\\
    	\{ q_1(A) &\mid \exists s \in S, \h{\exists r_3 \in R}
    		[\h{q_1.A = r_3.A} \wedge \h{\neg} q_2(\h{r_3.A}, s.B)] \} \\
    	\{q_2(A, B)  &\mid \exists r_2 \in R 	
    		[q_2.B \equal r_2.B \wedge r_2.A \equal q_2.A]] \}	
    \end{align*}
    This rewritten query now allows a direct translation into $\DatalogND$ from the inside out:
    \begin{align*}
    	Q_2(x, y) 	& \datarule R(x,y). 		\\
    	Q_1(x) 	& \datarule R(x,\_), S(y), \neg Q_2(x,y). 		\\
    	Q(x) 	& \datarule R(x,\_), \neg Q_1(x).
    \end{align*}
\end{example}

\begin{example}[Built-in predicates in $\DatalogND$]
    We next illustrate the translation for case $(ii)$ of a built-in predicate with
	the following query asking for values from $R$ for which no smaller value appears in $S$:
    \begin{align*}
    	& \{ q(A) \mid \exists r \in R[q.A \equal r.A
    		\wedge \neg (\exists s \in S [	s.A < r.A])]\}	
    \end{align*}
    This query is additionally used in \cref{ex:DatalogLimits} as $Q_3$ and illustrated in
	\cref{Fig_DatalogLimits_k,Fig_DatalogLimits_l,Fig_DatalogLimits_m,Fig_DatalogLimits_n,Fig_DatalogLimits_o}.
    Rewriting the query based on its recursive nested negation hierarchy gives us 2 query components:
    \begin{align*}
    	\{ q(A)   &\mid  \exists r \in R [q.A \equal r.A \wedge \neg (q_1(r.A))] \}	\\
    	\{ q_1(A) &\mid  \exists s \in S [\h{s.A < q_1.A}] \} 
    	\hspace{13mm}
    \end{align*}
	Now notice that $q_1$ is not safe 
	because $q_1.A$ is connected to an existentially-quantified table $s \in S$ only via a built-in predicate
	$s.A < q_1.A$ instead of an equality predicate.

    We can make $q_1$ safe
    by adding another table $r_2 \in R$ in $q_1$:
    \begin{align*}
    	\{ q(A)   &\mid  \exists r \in R [q.A \equal r.A \wedge \neg (q_1(r.A))] \}	\\
    	\{ q_1(A) &\mid  \exists s \in S, \h{r_2 \in R} [\h{s.A < r_2.A} \wedge \h{r_2.A = q_1.A}] \} 
    	\hspace{13mm}
    \end{align*}
    This rewritten query (also illustrated in
	\cref{Fig_DatalogLimits_p,Fig_DatalogLimits_q,Fig_DatalogLimits_r,Fig_DatalogLimits_s,Fig_DatalogLimits_t}) 
	now allows a direct translation into $\DatalogND$ 
	from the inside out:
    \begin{align*}
    	Q_1(x)	& \datarule R(x),  S(y), x\!>\!y. 		\\
    	Q(x) 	& \datarule R(x), \neg Q_1(x).
    \end{align*}
\end{example}

\begin{figure}[t]
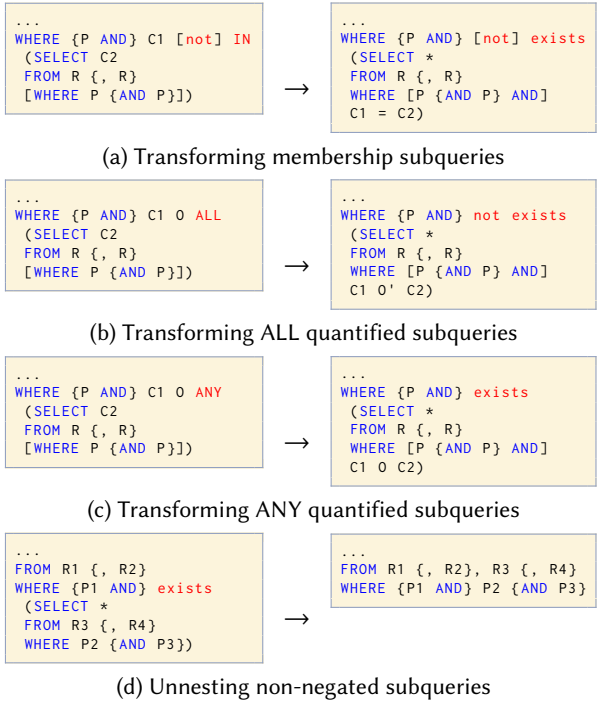

\centering
\begin{subfigure}[b]{1.0\linewidth}	
\centering
\begin{minipage}{32mm}		
\begin{lstlisting}
...
WHERE {P AND} C1 [not] IN
 (SELECT C2
 FROM R {, R}
 [WHERE P {AND P}])
\end{lstlisting}
\vspace{2.4mm}
\end{minipage}
$\hspace{3mm}\rightarrow\hspace{3mm}$
\begin{minipage}{33mm}		
\begin{lstlisting}
...
WHERE {P AND} [not] exists
 (SELECT *
 FROM R {, R}
 WHERE [P {AND P} AND]
 C1 = C2)
\end{lstlisting}
\end{minipage}
\vspace{-7mm}
\caption{Transforming membership subqueries}
\label{fig:canonicalquery_a}
\end{subfigure}	
\begin{subfigure}[b]{1.0\linewidth}	
\centering
\begin{minipage}{32mm}			
\begin{lstlisting}
...
WHERE {P AND} C1 O ALL
 (SELECT C2
 FROM R {, R}
 [WHERE P {AND P}])
\end{lstlisting}
\vspace{2.4mm}
\end{minipage}
$\hspace{3mm}\rightarrow\hspace{3mm}$
\begin{minipage}{33mm}	
\begin{lstlisting}
...
WHERE {P AND} not exists
 (SELECT *
 FROM R {, R}
 WHERE [P {AND P} AND]
 C1 O' C2)
\end{lstlisting}
\end{minipage}
\vspace{-7mm}
\caption{Transforming ALL quantified subqueries}
\label{fig:canonicalquery_c}
\end{subfigure}	
\begin{subfigure}[b]{1.0\linewidth}	
\centering		
\begin{minipage}{32mm}		
\begin{lstlisting}
...
WHERE {P AND} C1 O ANY
 (SELECT C2
 FROM R {, R}
 [WHERE P {AND P}])
\end{lstlisting}
\vspace{2.4mm}
\end{minipage}
$\hspace{3mm}\rightarrow\hspace{3mm}$
\begin{minipage}{33mm}		
\begin{lstlisting}
...
WHERE {P AND} exists
 (SELECT *
 FROM R {, R}
 WHERE [P {AND P} AND]
 C1 O C2)
\end{lstlisting}
\end{minipage}
\vspace{-7mm}
\caption{Transforming ANY quantified subqueries}
\label{fig:canonicalquery_e}
\end{subfigure}	
\hspace{6mm}
\begin{subfigure}[b]{1.\linewidth}		
\centering		
\begin{minipage}{32mm}			
\begin{lstlisting}
...
FROM R1 {, R2}
WHERE {P1 AND} exists
 (SELECT *
 FROM R3 {, R4}
 WHERE P2 {AND P3})
\end{lstlisting}
\end{minipage}
$\hspace{3mm}\rightarrow\hspace{3mm}$
\begin{minipage}{33mm}	
\begin{lstlisting}
...
FROM R1 {, R2}, R3 {, R4}
WHERE {P1 AND} P2 {AND P3}
\end{lstlisting}
\vspace{7.2mm}
\end{minipage}
\vspace{-7mm}
\caption{Unnesting non-negated subqueries}
\label{fig:canonicalquery_g}
\end{subfigure}	
\caption{\hyperref[th:equivalence:trc-sql]{Part (5)}
in the proof of \cref{th:equivalence}: 
$\NDSQL$ queries can be brought into a canonical form by 
replacing 
membership subqueries (a) and
quantified subqueries (b, c)
with existential subqueries that push join predicates into the local scope of the nested query,
and then unnesting non-negated subqueries (d).}
\label{fig:canonicalquery}
\end{figure}

\begin{figure*}[t]
\centering
\begin{subfigure}[b]{.22\linewidth}		
\begin{align*}
	&	\{ q(A) \mid \exists r \in R, \exists s \in S [\\
	&	q.A = r.A \wedge r.B = s.B]\}	
\end{align*}
\vspace{-4mm}
    \caption{}
\end{subfigure}
\hspace{4.5mm}
\begin{subfigure}[b]{.19\linewidth}		
\begin{lstlisting}
SELECT DISTINCT R.A
FROM R, S
WHERE R.B = S.B
\end{lstlisting}
\vspace{0.5mm}
    \caption{}
    \label{fig:SQLvariety_b}
\end{subfigure}
\hspace{12mm}
\begin{subfigure}[b]{.19\linewidth}
\begin{lstlisting}
SELECT DISTINCT R.A
FROM R 
WHERE exists
 (SELECT *
 FROM S
 WHERE R.B = S.B)
\end{lstlisting}
\vspace{-7mm}
\caption{}
\end{subfigure}
\\
\begin{subfigure}[b]{.19\linewidth}
\vspace{1mm}
\begin{lstlisting}
SELECT DISTINCT R.A
FROM R 
WHERE R.B in
  (SELECT S.B
  FROM S)
\end{lstlisting}
\vspace{-7mm}
\caption{}
\end{subfigure}
\hspace{12mm}
\begin{subfigure}[b]{.19\linewidth}		
\begin{lstlisting}
SELECT DISTINCT R.A
FROM R 
WHERE R.B = any
  (SELECT S.B
  FROM S )
\end{lstlisting}
\vspace{-7mm}
\caption{}
\end{subfigure}
\hspace{12mm}
\begin{subfigure}[b]{.19\linewidth}
    \includegraphics[scale=0.41]{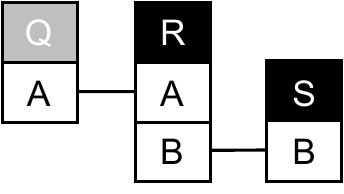}
	\vspace{-1mm}
    \caption{}
    \label{Fig_R_S_join}
\end{subfigure}	
\hspace{20mm}
\begin{subfigure}[b]{.22\linewidth}		
\begin{align*}
	&	\{ q(A) \mid \exists r \in R[q.A = r.A \,\wedge \\
	&	 \neg (\exists s \in S [r.B = s.B])]\}	
\end{align*}
\vspace{-4mm}
    \caption{}
\end{subfigure}
\hspace{-1mm}
\begin{subfigure}[b]{.19\linewidth}
	\phantom{test}
\end{subfigure}		
\hspace{12mm}	
\begin{subfigure}[b]{.19\linewidth}
\begin{lstlisting}
SELECT DISTINCT R.A
FROM R 
WHERE not exists
 (SELECT *
 FROM S
 WHERE R.B = S.B)
\end{lstlisting}
\vspace{-7mm}
\caption{}
\label{fig:SQLvariety_h}
\end{subfigure}
\\
\begin{subfigure}[b]{.19\linewidth}
\begin{lstlisting}
SELECT DISTINCT R.A
FROM R 
WHERE R.B not in
 (SELECT S.B
 FROM S)
\end{lstlisting}
\vspace{-7mm}
\caption{}
\end{subfigure}
\hspace{12mm}
\begin{subfigure}[b]{.19\linewidth}		
\begin{lstlisting}
SELECT DISTINCT R.A
FROM R 
WHERE R.B <> all
 (SELECT S.B
 FROM S)
\end{lstlisting}
\vspace{-7mm}
\caption{}
\end{subfigure}	
\hspace{12mm}
\begin{subfigure}[b]{.19\linewidth}
    \includegraphics[scale=0.41]{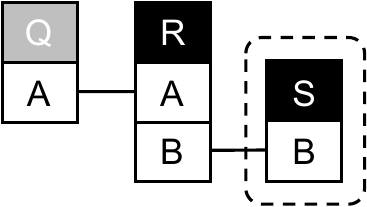}
	\vspace{-1mm}
    \caption{}
    \label{Fig_R_S_neg_join}
\end{subfigure}
\hspace{30mm}
\begin{subfigure}[b]{.22\linewidth}		
\begin{align*}
	&	\{ q(A) \mid \exists r \in R[q.A = r.A \,\wedge \\
	&	 \neg (\exists s \in S [r.B < s.B])]\}	
\end{align*}
\vspace{-4mm}
    \caption{}
\end{subfigure}
\hspace{0mm}
\begin{subfigure}[b]{.19\linewidth}
	\phantom{test}
\end{subfigure}		
\hspace{12mm}	
\begin{subfigure}[b]{.19\linewidth}
\begin{lstlisting}
SELECT DISTINCT R.A
FROM R 
WHERE not exists
 (SELECT *
 FROM S
 WHERE R.B < S.B)
\end{lstlisting}
\vspace{-7mm}
\caption{}
\label{fig:SQLvariety_m}
\end{subfigure}
\\
\begin{subfigure}[b]{.19\linewidth}
\vspace{1mm}
\begin{lstlisting}
SELECT DISTINCT R.A
FROM R 
WHERE not (R.B < any
 (SELECT S.B
 FROM S))
\end{lstlisting}
\vspace{-7mm}
\caption{}	
\end{subfigure}		
\hspace{12mm}
\begin{subfigure}[b]{.19\linewidth}		
\begin{lstlisting}
SELECT DISTINCT R.A
FROM R 
WHERE R.B >= all
 (SELECT S.B
 FROM S)
\end{lstlisting}
\vspace{-7mm}
\caption{}
\end{subfigure}	
\hspace{12mm}
\begin{subfigure}[b]{.19\linewidth}
    \includegraphics[scale=0.41]{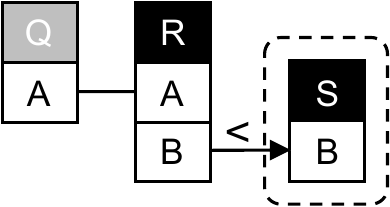}
	\vspace{-6mm}
    \caption{}
    \label{Fig_R_S_inequality_join}
\end{subfigure}
\caption{\Cref{ex:SQLvsTRC}: SQL has a redundant syntax if interpreted under set semantics (``\sql{SELECT DISTINCT}''), 
binary logic (tables contain no null values) 
and compared with TRC.
Here, queries (a)--(e), queries (g)--(j), and queries (l)--(n) are equivalent.
On the right, (f), (k), and (o) show the three corresponding \diagrams\ (\cref{sec:QV}) that abstract away the syntactic variants 
and focus on the logical patterns of the queries.
$\SQL$ queries (b), (h), (m) are canonical and isomorphic to the $\TRC$ queries.
}
\label{fig:SQLvariety}
\end{figure*}

\item[\namedlabel{th:equivalence:trc-sql}{\underline{(5) $\NDTRC \leftrightarrow \NDSQL$}}:]

We prove equivalence in three steps:
We first reduce the syntactic variety of $\NDSQL$,
then define a canonical form,
and finally prove a one-to-one mapping between that canonical $\NDSQL$ and canonical $\NDTRC$.

1. Starting from the grammar in \cref{table:supported_grammar}, 
we first transform ``membership subqueries'' 
(\cref{fig:canonicalquery_a})
and ``quantified subqueries'' 
(\cref{fig:canonicalquery_c,fig:canonicalquery_e})
into equivalent ``existential subqueries.''
Here \sql{O'} is the complement operator of \sql{O} (for example ``<'' for ``>='')
and \sql{C1} and \sql{C2} represent different columns or attributes.

2. Analogous to the canonical form of $\NDTRC$,
we pull existential quantifiers of tables (table variables defined in \sql{FROM} clauses)
as early as possible such that they either appear in the root query, or
directly following a \sql{not exists} (\cref{fig:canonicalquery_g}).
This reduction is deterministic, and
every valid $\NDTRC$ query is equivalent to exactly one canonical $\NDTRC$ query.

3. The resulting canonical $\NDSQL$ is now in a direct 1-to-1 correspondence to $\NDTRC$,
and the translation between $\NDSQL$ and $\NDTRC$ 
is then the matter of translating the different syntactic expressions between the two languages:
From our grammar (\cref{table:supported_grammar}), \sql{SELECT DISTINCT C \{, C\} } is equivalent to the output definition in $\NDTRC$
$\{q(\vec A) \mid \ldots \}$,
each \sql{FROM R \{, R\} } defines the existentially-quantified tuple variables $\exists \vec r \in \vec R[...]$,
each \sql{not exists(SELECT * FROM R \{, R\} ...)} corresponds to negated existentially-quantified tuple variables,
$\neg(\exists \vec r \in \vec R [...])$,
and the syntax of predicates is identical.
The Boolean variants are similar yet the resulting $\NDTRC$ is a logical statement and
thus without curly braces for set identifiers.
\qedhere
\end{pfparts}
\end{proof}

\begin{example}[$\NDSQL$ vs.\ $\NDTRC$]
\label{ex:SQLvsTRC}
\Cref{fig:SQLvariety} shows three different non-disjunctive queries in $\NDTRC$, various syntactic variants of $\NDSQL$, and \diagrams.
$\NDSQL$ queries (\subref{fig:SQLvariety_b}, \subref{fig:SQLvariety_h}, \subref{fig:SQLvariety_m}) are canonical and in a direct 1-to-1 relationship with $\NDTRC$.
\end{example}

\section{More explanations for \texorpdfstring{\cref{sec:QV}}{Section \ref{sec:QV}}}
\label{sec:UMLchoices}

A table can be represented by any visual grouping of its attributes (see \cref{Fig_Table_Attributes_Set_Visualization} for examples).
Our choice in \diagrams\ is to use the typical UML convention of representing tables as rectangular boxes
with the table name on top and attribute names below in separate rows
(\cref{Fig_Table_Attributes_Set_Visualization_a}).
Any alternative choice may affect the readability and usability of \diagrams,
yet does not affect their semantics and pattern expressiveness.
Our focus in this paper is pattern expressiveness, not usability.

\begin{figure}
\centering
\begin{subfigure}[b]{0.15\linewidth}
	\centering
    \includegraphics[scale=0.42]{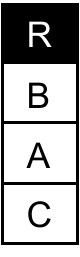}
	\vspace{-1mm}
    \caption{}
    \label{Fig_Table_Attributes_Set_Visualization_a}
\end{subfigure}
\begin{subfigure}[b]{0.18\linewidth}
	\centering	
    \includegraphics[scale=0.42]{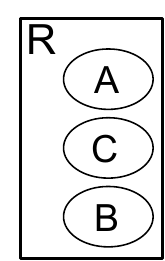}
    \caption{}
    \label{Fig_Table_Attributes_Set_Visualization_b}
\end{subfigure}
\begin{subfigure}[b]{0.23\linewidth}
	\centering	
    \includegraphics[scale=0.42]{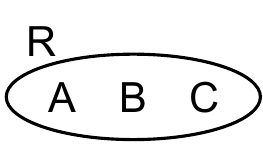}
    \caption{}
    \label{Fig_Table_Attributes_Set_Visualization_c}
\end{subfigure}
\begin{subfigure}[b]{0.17\linewidth}
	\centering	
    \includegraphics[scale=0.42]{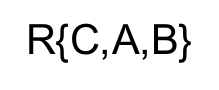}
    \caption{}
    \label{Fig_Table_Attributes_Set_Visualization_d}
\end{subfigure}
\begin{subfigure}[b]{0.21\linewidth}
	\centering	
    \includegraphics[scale=0.42]{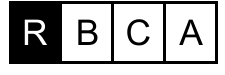}
    \caption{}
    \label{Fig_Table_Attributes_Set_Visualization_e}
\end{subfigure}
\caption{\Cref{sec:UMLchoices}:
A few alternative ways to visualize a table and its set of attributes as a group of nodes.
Inspired by a familiar UML convention for ER diagrams, we chose (a).
}
\label{Fig_Table_Attributes_Set_Visualization}
\end{figure}

\section{More examples for \texorpdfstring{\cref{sec:sentences}}
{Section \ref{sec:sentences}}: logical sentences (or Boolean queries)}
\label{sec:BooleanExamples}

\begin{figure}[t]
\centering	
\begin{subfigure}[b]{.21\linewidth}		
\begin{lstlisting}
SELECT DISTINCT S.sname
FROM Sailor S
WHERE not exists
 (SELECT *
 FROM Boat B
 WHERE B.color = 'red'
 AND not exists
  (SELECT *
  FROM RESERVES R
  WHERE R.bid = B.bid
  AND R.sid = S.sid))
\end{lstlisting}
\vspace{-3.3mm}
\caption{}
\label{SQL_Sailor_exists_Sailor_all_red_boats_text}
\end{subfigure}	
\hspace{13mm}
\begin{subfigure}[b]{.21\linewidth}		
\begin{lstlisting}
SELECT exists
 (SELECT *
 FROM Sailor S
 WHERE not exists
  (SELECT *
  FROM Boat B
  WHERE B.color = 'red'
  AND not exists
   (SELECT *
   FROM RESERVES R
   WHERE R.bid = B.bid
   AND R.sid = S.sid)))
\end{lstlisting}
\vspace{-6mm}
\caption{}
\label{SQL_Sailor_exists_Sailor_all_red_boats}
\end{subfigure}	
\hspace{20mm}
\begin{subfigure}[b]{.42\linewidth}
	\vspace{1mm}	
    \includegraphics[scale=0.45]{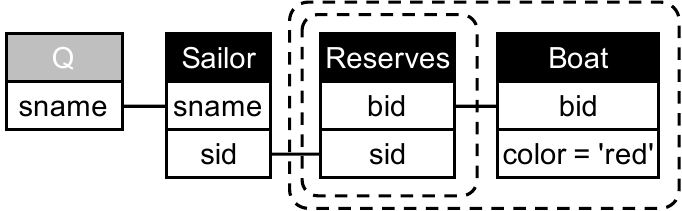}
    \caption{}
    \label{Fig_Sailor_Sailors_all_red_boats}
\end{subfigure}	
\hspace{2mm}
\begin{subfigure}[b]{.35\linewidth}
	\vspace{1mm}	
    \includegraphics[scale=0.45]{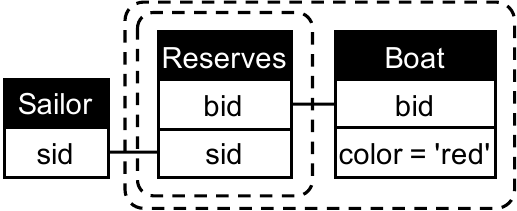}
	\vspace{0.5mm}
    \caption{}
    \label{Fig_Sailor_exists_Sailor_all_red_boats}
\end{subfigure}	
\caption{\Cref{ex:sailorsallredboats}: Sailors reserving all red boats.}
\label{fig:Sailor_sentences}
\end{figure}

\begin{example}[Sailors reserving all red boats]
	\label{ex:sailorsallredboats}
Consider the sailor database \cite{cowbook:2002} that models sailors reserving boats:
Sailor(sid, sname, rating, age),
Reserves(sid, bid, day),
Boat(bid,bname,color),
and the query
``Find sailors who reserved all red boats:''
\begin{align}
\begin{aligned}
\!\!\!\!\!\!\{q(\sql{sname}) \mid 
	& \exists s \in \sql{Sailor} [q.\sql{sname}=s.\sql{sname} \ \wedge \\
	&\neg (\exists b \in \sql{Boat} [b.\sql{color} = \sql{'red'} \wedge \\
	&\neg (\exists r \in \sql{Reserves} [r.\sql{bid} = b.\sql{bid} \wedge r.\sql{sid} = s.\sql{sid}])]) \}
\end{aligned}
\end{align}

Contrast it with the logical statement
``There is a sailor who reserved all red boats.''
In $\NDTRC$, the difference is achieved 
by leaving away curly brackets and any mentions of the output table
(highlighted for a different example in green color in \cref{Fig_TRC_vs_RD_a}):
\begin{align}
\begin{aligned}
	&\exists s \in \sql{Sailor}[ \\
	&\neg (\exists b \in \sql{Boat} [b.\sql{color} = \sql{'red'} \wedge 
	\\
	& \neg (\exists r \in \sql{Reserves} [r.\sql{bid} = b.\sql{bid} \wedge r.\sql{sid} = s.\sql{sid}])])
\end{aligned}
\label{trc:existssailorsallredboats}
\end{align}
Similarly, \diagrams\
loses the output table
(contrast \cref{Fig_Sailor_Sailors_all_red_boats}
with \cref{Fig_Sailor_exists_Sailor_all_red_boats} and their respective SQL statements).
\end{example}

\begin{figure}[t]
\centering	
\begin{subfigure}[b]{.18\linewidth}		
\begin{lstlisting}
SELECT not
 (not exists
  (SELECT *
  FROM R
  WHERE R.A = 1)
 AND not exists
  (SELECT *
  FROM R R2
  WHERE R2.A = 2)))
\end{lstlisting}
\vspace{-6mm}
\caption{}
\label{SQL_Disjunction}
\end{subfigure}	
\hspace{8mm}
\begin{subfigure}[b]{.19\linewidth}
    \includegraphics[scale=0.4]{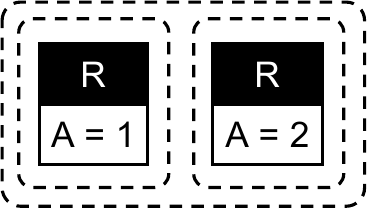}
\vspace{4mm}
\caption{}
\label{Fig_Disjunction}
\end{subfigure}	
\hspace{3mm}
\caption{\Cref{ex:disjunction}: $\exists r \in R [R.A = 1 \vee R.A = 2]$.}
\label{fig:disjunction_oneex}
\end{figure}

We give an example that shows that in order to express sentences (instead of queries), 
and to be relationally complete (in that we would like to be able to express all logical sentences), 
we actually would not have to introduce the visual union.
This is in stark contrast to the union at the root being \emph{necessary} for queries.

\begin{example}[Disjunctions]
	\label{ex:disjunction}
Consider the simplest disjunction
\begin{align*}
	& \exists r \in R [R.A = 1 \;\h{\vee}\; R.A = 2]
\end{align*}
We can remove the disjunction with a double negation:
\begin{align*}
	& \exists r \in R [R.A = 1] \;\h{\vee}\; \exists r \in R[R.A = 2] \\
	& \h{\neg}(\h{\neg}(\exists r \in R [R.A = 1]  \;\h{\vee}\; \exists r \in R[R.A = 2])) \\	
	& \h{\neg}(\h{\neg}(\exists r \in R [R.A = 1]) \;\h{\wedge}\; \h{\neg}(\exists r \in R[R.A = 2])) 
\end{align*}

The first 4 steps of the translation in \cref{sec:fromTRCtoRD} still work
and leads to \cref{Fig_Disjunction}.
For $\SQL$, the query uses the second new rule to express double negation before the first \sql{FROM} clause.
\end{example}

\section{Proof For \texorpdfstring{\cref{sec:structureisomorphism}, \cref{th:representations}}{Section \ref{sec:structureisomorphism}, Theorem \ref{th:representations}}}
\label{appendix:proofs4}
\label{appendix:proofofhierarchy}

\begin{figure}[t]
\centering
\begin{subfigure}[b]{0.66\linewidth}
	\centering
	\begin{tikzpicture}[>=stealth, shorten >= 2pt, shorten <= 2pt, line width=0.25mm]
	  \node (RA) {$\NDRA$};
	  \node[right=2mm of RA] (Datalog) {$\DatalogND$};
	  \node[right=2mm of Datalog] (TRC) {$\NDTRC$};
	  \node[right=2mm of TRC] (SQL) {$\NDSQL$};  
	  \node[right=2mm of SQL] (RD) {$\NRD$};    
    
	  \draw[->] (RA.north) to[bend left] 
	  	node[midway,above] {\hyperref[th:pattern:ra-datalog]{(1)}} 
		(Datalog.north);
	  \draw[->] (Datalog.south) to[bend left] 
	  	node[pos=0.5, sloped] {\small$/$}
		node[midway,below] {\hyperref[th:pattern:datalog-ra]{(2)}} 
		(RA.south);
	  \draw[->] (Datalog.north) to[bend left] 
	  	node[midway,above] {\hyperref[th:pattern:datalog-trc]{(3)}} 
		(TRC.north);
	  \draw[->] (TRC.south) to[bend left] 
	   	node[pos=0.5, sloped] {\small$/$}
	  	node[midway,below] {\hyperref[th:pattern:trc-datalog]{(4)}} 
		(Datalog.south);
	  \draw[<->] (TRC.north) to[bend left] 
	  	node[midway,above] {\hyperref[th:pattern:trc-sql]{(5)}}
		(SQL.north);
	  \draw[<->] (TRC.south) to[bend right=20] 
		node[midway,below] {\hyperref[th:pattern:trc-rd]{(6)}} 
		(RD.south);	
	\end{tikzpicture}
    \caption{}
\end{subfigure}	
\hspace{1mm}
\begin{subfigure}[b]{0.31\linewidth}
	\centering
	\begin{tikzpicture}[scale=1.0]
	\draw
	  (1.05,1.1) node[anchor=west] {$\DatalogND$}
	  (1.35,0.5) node[anchor=west] {$\NDRA$}  
	  (0.0,1.15) node[anchor=west] {$\NDTRC$}
	  (0.0,0.75) node[anchor=west] {$\NDSQL$}  
	  (0.0,0.35) node[anchor=west] {$\NRD$};        
	\draw[black,rounded corners=8,thick]
		(0,0) rectangle (2.6, 1.5);
	\draw[black,rounded corners=6,thick]
	    (0.95,0.1) rectangle (2.5,1.4);
	\draw[black,rounded corners=4,thick]
	    (1.05,0.2) rectangle (2.4,0.8);	
	\end{tikzpicture}
    \caption{}
\end{subfigure}	
\caption{\Cref{appendix:proofs4}: 
Directions used in proof for \cref{th:representations} (a) and resulting representation hierarchy (b).
We use two separations lemmas for 
{\hyperref[th:pattern:datalog-ra]{(2)}} 
and
{\hyperref[th:pattern:trc-datalog]{(4)}} 
in \cref{sec:separationlemmas}.
}
\label{Fig_Representation_proof}
\end{figure}
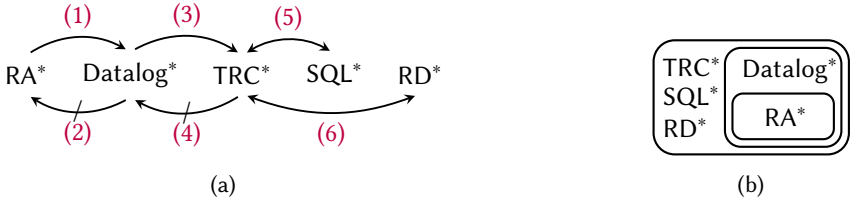

The most interesting parts of this proof 
are two separation results
(\cref{Fig_Representation_proof}).
We bring those in two separate lemmas first.

\subsection{Two Separation Lemmas}
\label{sec:separationlemmas}

\begin{lemma}[$\NDRA \not \supseteq^\rep \DatalogND$]
	\label{lem:RA_separation}
	The following $\DatalogND$ query 
	over schema $R(A,B), S(B)$
	has no pattern-isomorphic query in $\NDRA$:
	\begin{align}
		Q(x,y)	 & \datarule R(x,y), \neg S(y)
		\label{eq:intro_appendix_example_Datalog}
	\end{align}	
\end{lemma}

\begin{proof}[Proof \cref{lem:RA_separation}]
	We show that query \cref{eq:intro_appendix_example_Datalog}
	(also used in our earlier \cref{ex:intro}) 
	has no pattern-isomorphic query in $\NDRA$.
	The simple intuition is that 	
	the binary \emph{minus operator from $\RA$ requires the same arity} of the two input relations.
	Thus one cannot apply the minus operator directly to combine $R$ and $S$ as in $\DatalogND$. 
	Any possible sequence that includes a minus thus either uses the minus on 1 attribute, or 2 attributes 
	(or more attributes).
	We show that any such option requires at least 3 table instances.

	Case 1: Minus on 2 (or more) attributes: 
	Having 2 (or more) attributes for the minus requires us to increase the arity of the right side and thus $S$. 
	This in turn requires a cross product with the domain from $R.A$ before the minus as in the following translation:
	\begin{align*}
		R - (\pi_A R \times S) 		
	\end{align*}
	This in turn increases the number of input table instances used from 2 to at least 3, 
	which prevents a pattern-preserving representation.

	Case 2: Minus on 1 attribute: 
	Having 1 attribute on the minus requires us to increase the arity \emph{after applying the minus} (because our output has arity 2). 
	This in turn unavoidably increases the number of input table instances to at least 3,
	which again prevents a pattern-preserving representation.
	An example translation first uses a projection on $R$ before the minus (to reduce the left input to arity 1) and then a subsequent join again with $R$ after the minus:
	\begin{align*}
		R \Join_B (\pi_B R - S)		
	\end{align*}

	It follows that any $\NDRA$ expression representing
	the $\DatalogND$ expression 
	\cref{eq:intro_appendix_example_Datalog} needs at least 3 references to input tables
	and thus cannot preserve the representation from \cref{eq:intro_appendix_example_Datalog}.
\end{proof}

\begin{lemma}[$\DatalogND \not \supseteq^\rep \NDTRC$]
	\label{lem:Datalog_separation}
	The following $\NDTRC$ query 
	over a schema $T(A), R(A,B), S(B)$
	has no pattern-isomorphic query in $\DatalogND$:
	\begin{align}
	\begin{aligned}
	\{ q(A) \mid 
		&\exists t \in T 	[q.A \equal t.A \wedge \neg (\exists s \in S[	\\
		&\neg (\exists r \in R [r.B \equal s.B \wedge \h{r.A \equal t.A}])])] \}
	\end{aligned}
	\label{TRC:division_appendix_simpler}
	\end{align}
\end{lemma}

\begin{figure}[t]
\centering
\hspace{1mm}
\begin{subfigure}[b]{.3\linewidth}
    \includegraphics[scale=0.4]{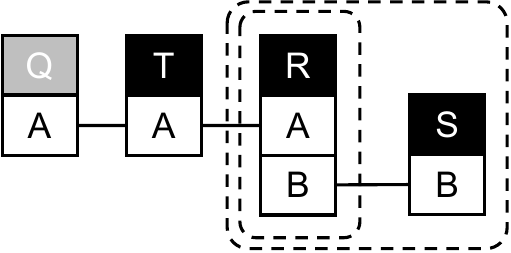}
	\vspace{2mm}
    \caption{}
    \label{Fig_relational_division_1_proof}
\end{subfigure}
\hspace{5mm}
\begin{subfigure}[b]{.32\linewidth}
    \includegraphics[scale=0.4]{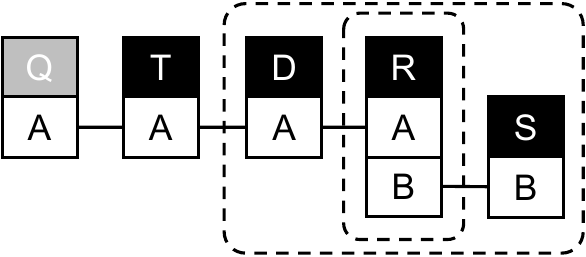}
	\vspace{2mm}
    \caption{}
    \label{Fig_relational_division_2_proof}
\end{subfigure}
\caption{Illustration for \cref{lem:Datalog_separation} 
to show that $\DatalogND \not \supseteq^\rep \NDTRC$, 
which forms
\hyperref[th:pattern:trc-datalog]{part (4)}
of the proof of \cref{th:representations}.}
\label{fig:division_proof}
\end{figure}

\begin{proof}[Proof \cref{lem:Datalog_separation}]
The $\NDTRC$ query \cref{TRC:division_appendix_simpler} 
returns attribute values from $T.A$ that co-occur in $R$ with all attribute values from $S.B$.
We will prove that $\DatalogND$ cannot isomorphically represent this query 
(which corresponds to the dissociated query \cref{TRC:division1} for relational division from \cref{ex:relationalDivision}).

The intuitive reason is that
the predicate ``$t.A \equal r.A$'' represents a join
across two different negation scopes
(\Cref{Fig_relational_division_1_proof} shows that query as \diagram\ where that join 
predicate ``$t.A \equal r.A$''
crosses two negation boxes).
The safety condition of $\DatalogND$  
requires that each variable occurring in a negated atom also occurs in at least one non-negated atom 
\emph{of the same rule}~\cite{DBLP:journals/tkde/CeriGT89}.
As such, it can express negation only \emph{one rule at a time}
(each rule only allows the application of one negation),
and it cannot represent the join predicate $r.A \equal t.A$ across \emph{two} negations without an additional domain table to fulfill the safety condition. 
The standard translation thus uses two separate rules,
each of which can express maximally one negation scope,
and
\emph{each such rule needs to fulfill the safety conditions}.
The first rule finds all the $A$ values that do not co-occur with all $S.B$ values. 
The second rule then finds the complement against the domain from $T.A$:
\begin{align*}
\begin{aligned}
	I(x) 	& \datarule D(x), S(y), \neg R(x,y) 		\\
	Q(x) 	& \datarule T(x), \neg I(x) 		
\end{aligned}
\end{align*}	
Notice that the first ``extra'' atom $D(x)$ 
is needed for the safety condition of $\DatalogND$ 
that requires that each variable occurring in a negated atom also occur in at least one non-negated atom \emph{of the same rule}~\cite{DBLP:journals/tkde/CeriGT89}.
Here $D.A$ represents any domain that includes all values that appear in $T.A$
(it can be a subset or even empty only if $S.B$ is also empty).
Thus in general, we require $D.A \supseteq T.A$ for this Datalog query to be logically equivalent with our $\NDTRC$ query \cref{TRC:division_appendix_simpler}.\footnote{Assume for a moment that we did not require $D.A \supseteq T.A$ and instead used $D.A = R.A$. Then the query would incorrectly return $Q(0)$ for the example database $T(0), D(1), S(5), R(1,5)$.}
\Cref{Fig_relational_division_2_proof} shows this last query as \diagram.
This domain table cannot be $R$ or $S$ (which may contain different domains) and 
one occurrence of $T$ needs to be already used outside any negation).
We thus need to use an additional reference to $T$ 
as ``guard'' for that join predicate. 
Thus $\DatalogND$ requires adding an additional table reference.
Thus $\DatalogND$ cannot preserve the pattern from \cref{TRC:division_appendix_simpler}.

This was an intuition, and 
we make this more precise: 
We next show that the $\NDTRC$ query \cref{TRC:division_appendix_simpler} cannot be expressed
in $\DatalogND$ \emph{without some additional domain table} $D(x)$, 
which would have to be $T(x)$ if $T, R, S$ were the only tables available in a database.
In other words, it cannot be expressed with each of the 3 tables $T, R, S$ appearing only once.

First, observe that $q$ is positive monotone in $T$ and $R$ and negative monotone in $S$, 
i.e.\ adding tuples to $T$ or $R$ can never remove a tuple from the output, 
and adding tuples to $S$ can never add a tuple to the output.

Next consider a general Datalog program without built-in predicates. Here $p_i$ are positive atoms, and $n_j$ negative atoms:
\begin{align*}
	q_d(\vec x_d) &\datarule p_{d1}(\vec x_{d1}), \ldots, p_{d k_d}(\vec x_{d k_d}), 
		\neg n_{d1}(\vec y_{d1}), \ldots, \neg n_{d m_d} (\vec y_{d m_d}). \\
		\ldots \\
	q_1(\vec x_1) &\datarule p_{11}(\vec x_{11}), \ldots, p_{1 k_1}(\vec x_{1 k_1}), 
		\neg n_{11}(\vec y_{11}), \ldots, \neg n_{1 m_1} (\vec y_{1 m_1}). \\
	q_0(\vec x_0) &\datarule p_{10}(\vec x_{10}), \ldots, p_{0 k_0}(\vec x_{0 k_0}), 
		\neg n_{01}(\vec y_{01}), \ldots, \neg n_{0 m_0} (\vec y_{0 m_0}). 
\end{align*}

Recall from \cref{def:DatalogND} of $\DatalogND$
that 
($i$) every IDB appears in the head of exactly one rule (no disjunction),
and ($ii$) every IDB can be used maximally once in any body (no re-use of IDBs).
Furthermore ($iii$) we will assume every relation $T, R, S$ appears exactly once in the body and show this leads to a contradiction.

WLOG we only focus on ``canonical'' Datalog programs in which IDBs in the body can only appear in negated atoms. 
This is WLOG, because an IDB that appears in a positive atom (recall it can appear only once in the body) 
can always be replaced with the body of its defining rule and give an equivalent Datalog program with the same number of table occurrences.

Next, define the ``nesting depth'' $\mathit{ND}(q_i)$ of $q_i$ recursively as the number of rules to traverse to reach $q_0$, which has by definition depth 0. In other words, treat each rule as a hyperedge with atoms in head and body as vertices. 
Then $\mathit{ND}(q_i)$ is the length of the path (the number of hyperedges or rules to traverse) to reach $q_0$ starting from $q_i$.

Next, define the ``sign'' of an EDB in $q_i$ as \emph{positive} if is it appears as positive atom in $q_i$ 
and $(\mathit{ND}(q_i) \textrm{ mod } 2) = 0$, 
or as negative atom and $(\mathit{ND}(q_i) \textrm{ mod } 2) = 1$. 
Analogously for \emph{negative}. 

Recall that every relation and every IDB is used only once, and thus appears either as positive or negative. 
It follows from an induction argument on the nesting depth that the query $q_0$ is positive monotone in relation $R$
iff it appears with positive sign in the Datalog program. Analogously for negative monotone.

Our proof now proceeds by simple enumeration over all possible canonical $\DatalogND$ programs 
that use $T$, $R$, $S$ only once, consistent with their defined ``signs'' ($T$ and $R$ positive, and $S$ negative).

Case 1. Consider one rule 
$q_0(x) :- T(x), R(\vec z), \neg S(y)$. 
A single rule cannot express $q'$, thus we need at least 2 rules.

Case 2. The only way to have canonical rules up to nesting depth 2 with $T$ in $q_0$ and $R$ and $S$ appearing with a sign consistent with their expected monotonicity in $q_0$ is
\begin{align*}
	q_2(y) &\datarule R(x,y)\\
	q_1(y) &\datarule S(y), \neg q_2(y). \\
	q_0(x) &\datarule T(x), \neg q_1(y). 
\end{align*}
which is not equal to $q'$.
Thus, we can only have a Datalog rule with nesting depth 1.

Case 3. There is no way to have two canonical rules with nesting depth 1 with $T$ in $q_0$ and $R$ and $S$ appearing with a sign consistent with their expected monotonicity in $q_0$ since one rule needs to contain $S$ as a positive atom, and the other one would have to have $R$ as a negative atom, which is not possible without an additional atom to make this safe.
Thus we can only have two rules, one of which is nesting depth 1.

Case 4. Assume we have two canonical rules, one of nesting level 1. Since $R$ is positive, it can either appear in level 0 as positive or in level 1 negated. $R$ appearing in level 0 leads to a contradiction:
\begin{align*}
	q_1(y) &\datarule S(y). \\
	q_0(x) &\datarule T(x), R(\vec z), \neg q_1(y). 
\end{align*}

Case 5. The last option is of the form:
\begin{align*}
	q_1(x) &\datarule S(y), \neg R(x, y). \\
	q_0(x) &\datarule T(x), \neg q_1(x). 
\end{align*}
For that program to be safe, both $x$ and $y$ in $R$ need to be 
bound to an element in the active domain, 
which can only happen if unary $S$ is accompanied with another relation. Contradiction.
\end{proof}

\subsection{Proof Representation Hierarchy}

\begin{proof}[Proof of \cref{th:representations}]
We prove each of the directions in turn (\cref{Fig_Representation_proof}).
The logical equivalences already follow from the proof of \cref{th:equivalence}
in \cref{appendix:proofoflogicalexpressivness}.
We need to prove that certain directions are guaranteed to be pattern-preserving,
and for the other directions that do not preserve the structure in general, we give minimum counterexamples.

\begin{pfparts}
\item[\namedlabel{th:pattern:ra-datalog}{\underline{(1) $\NDRA \subseteq^\rep \DatalogND$}}:]
This direction follows immediately from the 
\hyperref[th:equivalence:ra-datalog]{proof part (1)} 
of \cref{th:equivalence} 
in \cref{appendix:proofoflogicalexpressivness}
by observing each of the mappings in the 5 cases to be pattern-preserving.

\item[\namedlabel{th:pattern:datalog-ra}{\underline{(2) $\NDRA \not \supseteq^\rep \DatalogND$}}:]
\Cref{lem:RA_separation}
showed that the set difference (or minus $-$) 
from $\NDRA$ cannot isomorphically represent negation from $\DatalogND$
if the complementing set of attributes is non-empty (see \cref{eq:translation_Datalog_RA_problem}).

\item[\namedlabel{th:pattern:datalog-trc}{\underline{(3) $\DatalogND \subseteq^\rep \NDTRC$}}:]
This direction follows immediately from the proof of \cref{th:equivalence} 
by observing the mappings of each Datalog rule to be pattern-preserving.

\item[\namedlabel{th:pattern:trc-datalog}{\underline{(4) $\DatalogND \not \supseteq^\rep \NDTRC$}}:]
\Cref{lem:Datalog_separation}
showed that $\DatalogND$ cannot isomorphically represent 
the $\NDTRC$ query \cref{TRC:division_appendix_simpler}
(which corresponds to the dissociated query \cref{TRC:division1} for relational division from \cref{ex:relationalDivision}).

\item[\namedlabel{th:pattern:trc-sql}{\underline{(5) $\NDTRC \equiv^\rep  \NDSQL$}}:]
This also follows immediately from the 
\hyperref[th:equivalence:trc-sql]{proof part (5)}
in \cref{appendix:proofoflogicalexpressivness}
which 
preserves 1-to-1 correspondences of the mappings
in either direction.

\item[\namedlabel{th:pattern:trc-rd}{\underline{(6) $\NDTRC \equiv^\rep \NRD$}}:]
This follows immediately from the 
step-by-step translations in
\cref{sec:fromTRCtoRD,sec:fromRDtoTRC} 
which keeps a 1-to-1 correspondence between table references
and thus form the proof.
\qedhere
\end{pfparts}
\end{proof}

\section{Extension for 
\texorpdfstring{\cref{SEC:STRUCTUREISOMORPHISM}}
{Section \ref{SEC:STRUCTUREISOMORPHISM}}:
Pattern expressiveness of 
\texorpdfstring{$\NDRA$}{\NDRAPlain} 
with additional operators}
\label{appendix:extensions}

\subsection{
\texorpdfstring{$\NDRA$}{\NDRAPlain}
with antijoins
(\texorpdfstring{$\NDRAA$}{\NDRAA})}
\label{appendix:antijoins}

We have so far focused on the pattern expressiveness of $\NDRA$ using only the basic algebraic operators except for the union.
We will now show that adding the antijoin operator can extend the expressiveness to the same as $\DatalogND$.

Given two relations $R$ and $S$,
and a conjunction $c$ of equality predicates between attributes from $R$ and $S$,
the antijoin $R \antijoin_c S$ (sometimes written as $R \overline{\ltimes}_c S$)
returns all tuples in $R$ 
that do not have any tuple in $S$ that joins with $R$ based on the equality predicates in $c$~\cite{DBLP:books/mg/SKS20}.
For example, $R \antijoin_{R.A=S.B} S$
outputs all tuples in $R$ 
that do not have any tuple in $S$ whose $s.B$ attribute value matches that tuples $R.A$ attribute value. 
It is formally defined as
$R \antijoin_c S = R - \pi_{\mathit{schema}(R)}(R \Join_c S)$.
It is customary to leave away the conditions $c$ if they are equijoins on identically named attributes and thus corresponding to 
a natural join. 
Thus,
$R \antijoin S = R - \pi_{\mathit{schema}(R)}(R \Join S)$.

\begin{figure}[t]
\centering
\begin{subfigure}[b]{0.4\linewidth}
	\centering
	\begin{tikzpicture}[>=stealth, shorten >= 2pt, shorten <= 2pt, line width=0.25mm]
	  \node (RA) {$\NDRAA$};
	  \node[right=3mm of RA] (Datalog) {$\DatalogND$};
    
	  \draw[->] (RA.north) to[bend left] 
	  	node[midway,above] {\hyperref[th:pattern:ra-datalog]{(1)}} 
		(Datalog.north);
	  \draw[->] (Datalog.south) to[bend left] 
		node[midway,below] {\hyperref[th:pattern:datalog-ra]{(2)}} 
		(RA.south);
	\end{tikzpicture}
	\vspace{-2mm}
    \caption{}
\end{subfigure}	
\hspace{5mm}
\begin{subfigure}[b]{0.35\linewidth}
	\centering
	\begin{tikzpicture}[scale=1.0]
	\draw
	  (1.05,1.0) node[anchor=west] {$\DatalogND$}
	  (1.35,0.5) node[anchor=west] {$\NDRAA$}  
	  (0.0,1.15) node[anchor=west] {$\NDTRC$}
	  (0.0,0.75) node[anchor=west] {$\NDSQL$}  
	  (0.0,0.35) node[anchor=west] {$\NRD$};        
	\draw[black,rounded corners=8,thick]
		(0,0) rectangle (2.6, 1.5);
	\draw[black,rounded corners=6,thick]
	    (0.95,0.1) rectangle (2.5,1.4);
	\end{tikzpicture}
    \caption{}
\end{subfigure}	
\caption{\Cref{appendix:antijoins}: 
Directions used in proof for \cref{th:NDRAA} (a) and resulting representation hierarchy (b).}
\label{Fig_NDRAA}
\end{figure}
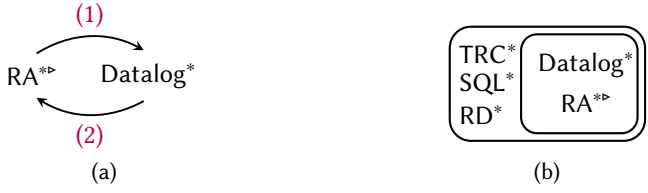

\begin{theorem}[$\NDRAA$]
\label{th:NDRAA}
$\NDRAA$ ($\NDRA$ extended with the antijoin operator $\antijoin$)
and $\DatalogND$
are re\-pre\-sentation-equivalent. 
\end{theorem}

\begin{proof}[\cref{th:NDRAA}]
We again prove each of the two directions in turn (\cref{Fig_NDRAA})
and build upon the proof of \cref{th:equivalence}
in \cref{appendix:proofoflogicalexpressivness}.

\begin{pfparts}	
\item[\namedlabel{th:pattern:raanti-datalog}{\underline{(1) $\NDRAA \subseteq^\rep \DatalogND$}}:]
Building upon \cref{appendix:proofoflogicalexpressivness},
we need to show that the translation of the antijoin is pattern-preserving.
	
Case 6: $Q = E_1 \antijoin_c E_2$: 
Assume that there are predicates $e_1$ and $e_2$ whose rules define their relations to be the same as the relations for $E_1$ and $E_2$.
We partition the attributes of $E_1$ and $E_2$ into three sets based on the equality conditions $c$:
those that appear in $E_1$ but not in $E_2$ (indexed by $\vec x$),
those that appear in both (indexed by $\vec y$),
and those that appear only in $E_2$ (indexed by $\vec z$).

Then we use two rules:
\begin{align*}
	q'(\vec y) &\datarule e_2(\vec y, \vec z). \\
	q(\vec x, \vec y) &\datarule e_1(\vec x, \vec y),  \neg q'(\vec y).	
\end{align*}
to define a first predicate $q'$ whose relation contains all the attribute values from $E_2$ that could join with $E_1$
and $q$ a second predicate that contains the result of the antijoin.
We can easily see that the two rules preserve the referenced input tables,
and that the safety for these rules is fulfilled as all variables appearing in the negated $q'$ also appear in the positive $e_1$.

\item[\namedlabel{th:pattern:datalog-raanti}{\underline{(2) $\DatalogND \rightarrow \NDRAA$}}:]
Building upon \cref{appendix:proofoflogicalexpressivness},
we show that the translation of 
rules with negated subgoals in the body can be achieved in a pattern-preserving way by using the antijoin.

Concretely, take a general Datalog rule with built-in predicates:
\begin{align*}
	q(\vec x) \datarule p_1(\vec x_1), \ldots, p_k(\vec x_k),  \neg n_1(\vec y_1), \ldots,  \neg n_m(\vec y_m), c_{\theta}.
\end{align*}
From the safety conditions of this rule, we know that
all variables in the built-in predicates $c_{\theta}$ need to appear in positive atoms.
Similarly, all variables in negated atoms also need to appear in positive atoms:
$\bigcup_1^k \vec x_i \supseteq \vec y_i, \forall i \in [m]$. 
Let $P_i$ and $N_i$ be the $\NDRA$ expressions corresponding to $\DatalogND$ predicates $p_i$ and $n_i$,
$\vec A$ represent the set of attributes indexed by $\vec x$,
and $c_i$ represent the conjunction of equality predicates implied by the set of re-used variables $\vec y_i$ for 
each negated predicate $n_i, i \in [m]$.
Then the following expression translates one single rule with built-in predicates 
into a valid $\NDRAA$ expression in a pattern-preserving way:
\begin{align}
	\pi_{\vec A}(\sigma_{\theta} (\ldots(((P_1 \Join \ldots \Join P_k)\antijoin_{c1} N_1) \antijoin_{c2} N_2) \ldots \antijoin_{c_m} N_m))
\end{align}
Since all variables in the built-in predicate $c_{\theta}$ need to appear in $\bigcup_1^k \vec x_i$, 
the selection $\sigma_{\theta}$ can be correctly applied on $Q'$.
It then follows by induction on the order in which the IDB predicates are considered that each has a relation defined by some expression 
in $\NDRA$.
\qedhere
\end{pfparts}
\end{proof}

\begin{example}
Using the antijoin operator, the $\DatalogND$ query 
$Q(x,y)	\datarule R(x,y), \neg S(y)$,
from \cref{ex:intro} 
can be translated in a pattern-preserving way into the $\NDRAA$ expression
$R \antijoin S$
(or 
$R \antijoin_{R.B=S.B} S$
with explicit join conditions)
\end{example}

\begin{example}[Relational division with antijoins]
\label{ex:divisions_with_antijoins}
We can use the antijoin operator to also express relational division
\begin{align*}
\begin{aligned}
	I(x) 	& \datarule R(x,\_), S(y), \neg R(x,y). 		\\
	Q(x) 	& \datarule R(x,\_), \neg I(x). 		
\end{aligned}
\end{align*}	
in $\NDRAA$ as
\begin{align*}
	\pi_A R  \antijoin \pi_A\big( (\pi_A R \times S) \antijoin R \big)
\end{align*}
While standard $\SQL$ does not have a dedicated antijoin operator, 
it can model the operator
via a left join and ``IS NULL'' selection as shown in \cref{fig:SQL_antijoins}
Notice that the syntax is quite different from the $\SQL$ expressions shown elsewhere throughout this paper
(e.g.\ we require subqueries in the FROM clause because we require two ``IS NULL'' selections)
and, for equivalence, still requires that the input tables have no NULL values 
(relations with NULLs are only an intermediate representation to express the antijoin).
However, this syntax is then still pattern-isomorphic with 
relational division expressed above in $\Datalog$
and in \cref{Fig_equivalence1}
in various languages.
This example illustrates the powerful abilities of reasoning in terms of patterns across languages and various syntactic constructs.
\end{example}

\begin{figure}[t]
\centering
\begin{minipage}{0.23\linewidth}
\begin{lstlisting}
FROM R LEFT JOIN 
  (SELECT X.A
  FROM (SELECT R.A, S.B
       FROM R, S) AS X
  LEFT JOIN R
  ON (R.B = X.B
  AND R.A = X.A)
  WHERE R.A IS NULL) AS Y
ON R.A = Y.A
WHERE Y.A IS NULL  
\end{lstlisting}
\vspace{0mm}
\end{minipage}
\vspace{-5mm}
\caption{\Cref{ex:divisions_with_antijoins}: Relational division in $\SQL$ with antijoins (syntactically expressed and LEFT JOINS with ``IS NULL'' selections).}
\label{fig:SQL_antijoins}
\vspace{-1mm}
\end{figure}

\subsection{Adding relational division to 
\texorpdfstring{$\DatalogND$ }{\DatalogND}}
\label{sec:adding_relational_division}

\begin{figure}[t]
\centering
\begin{subfigure}[b]{.3\linewidth}
    \includegraphics[scale=0.4]{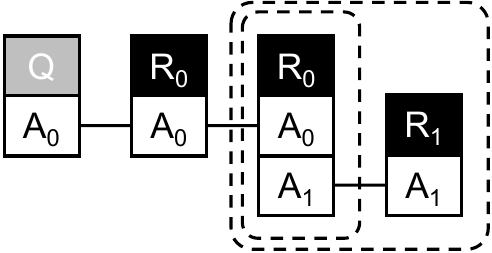}
    \caption{}
    \label{Fig_Separation_division_1}
\end{subfigure}
\hspace{9mm}
\begin{subfigure}[b]{.32\linewidth}
    \includegraphics[scale=0.4]{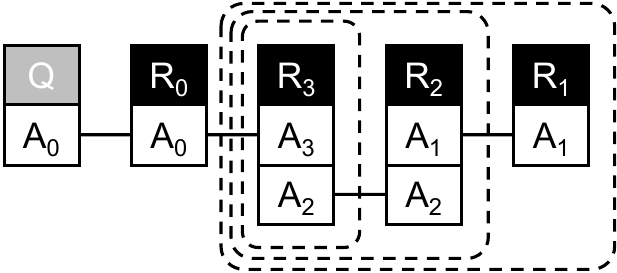}
    \caption{}
    \label{Fig_Separation_division_2}
\end{subfigure}
\caption{
Illustration for \cref{sec:adding_relational_division}.
}
\label{fig:division_does_not_help}
\end{figure}

We saw adding the antijoin operator to $\NDRA$ makes resulting language $\NDRAA$ equally pattern-expressive as $\DatalogND$.
We next show that adding the relational division to $\DatalogND$ 
does not make it as expressive as $\NDTRC$.
We achieve this by giving a $\NDTRC$ query $q_{S3}$ that requires two additional tables in the translation to 
$\DatalogND$ (and thus also to $\NDRA$ since $\DatalogND$ can pattern represent $\NDRA$).

\begin{align}
\begin{aligned}
\{ q_{S3}(A) 
	& \mid \exists r_0 \in R_0 	[q_{S3}.A_0 \equal r_0.A_0 \wedge 	\\
	& \: \neg \exists r_1 \in R_1[								\\
	& \: \neg \exists r_2 \in R_2[r_2.A_1 \equal r_1.A_1 \wedge \\
	& \: \neg \exists r_3 \in R_3[r_3.A_2 \equal r_2.A_2 \wedge r_3.A_3 \equal r_0.A_0])])] \}	
\end{aligned}
\label{TRC:separationquery3}
\end{align}
We assume a schema of only binary relations:
$R_0(A_0, B)$, 
$R_1(A_0, A_1)$, 
$R_2(A_1, A_2)$, 
$R_3(A_2, A_3)$.
\Cref{Fig_Separation_division_2} shows $q_{S3}$ in \diagrams.
Now assume that we add an additional ``operator'' to $\DatalogND$ that captures the semantics of relational division
(shown in \Cref{Fig_Separation_division_1}):
\begin{align}
\begin{aligned}
	I(x) 	& \datarule R_0(x,\_), R_1(\_, y), \neg R_0(x,y). 		\\
	Q_D(x) 	& \datarule R_0(x,\_), \neg I(x). 		
\end{aligned}
\label{datalog}
\end{align}	
We represent this ``operator'' as a function $f_D$ 
that we add to the constructs of the resulting language,
and show that the resulting language has not pattern-isomorphic representation of $q_{S3}$.
\begin{align}
\begin{aligned}
	Q_D(x) 	& \datarule f_D[R_0(x,y), R_1(z,y)].
\end{aligned}
\label{datalog}
\end{align}

First, notice that $f_D$ is monotone positive in $R_0$ and monotone negative in $R_1$.
$q_{S3}$ is monotone positive in $R_0$ and $R_2$ and monotone negative in $R_1$.

Next notice that the output schema of $f_D$ is unary, while the inputs to $f_D$ are binary. 
Thus $f_D$ cannot be composed. This prevents that the resulting query has two occurrences of $f_D$
as otherwise all 4 input tables $R_0, R_1, R_2, R_3$ would be used as inputs to the two occurrences of $f_D$
while we cannot express the query. It follows $f_D$ must occur exactly once (it must occur at least once since 
$q_{S3}$ cannot be expressed without $f_D$).

The proof now succeeds be checking all finitely many ways in which $q_{S3}$ could be expressed 
with one occurrence of $f_D$ while respecting the monotonicity constraints and realizing that none of these queries is equivalent to $q_{S3}$.

\subsection{Conjecture about limits of adding operators}

We leave it open whether $\RA$ can be extended to express the full range of relational patterns available in logical languages such as $\NDTRC$.
We hypothesize that this gap between procedural languages (either via an algebraic language like $\RA$ or declarative languages that are restricted to rule-by-rule evaluation such as $\Datalog$) and logical languages cannot be breached with any natural operator that can be readily implemented in a procedural way.

\section{More examples For \texorpdfstring{\cref{sec:structureisomorphism}}{Section \ref{sec:structureisomorphism}}}

We next illustrate,
with the help of the more complicated example of relational division,
that there is a pattern-preserving mapping from $\RA$ to $\TRC$, but not in the other direction.

\begin{figure*}[t]
\centering
\begin{subfigure}[b]{.19\linewidth}
\begin{lstlisting}
SELECT DISTINCT R.A
FROM R
WHERE not exists
 (SELECT *
 FROM S
 WHERE not exists
  (SELECT *
  FROM R AS R2
  WHERE R2.B = S.B
  AND R2.A = R.A))
\end{lstlisting}
\vspace{-7mm}
    \caption{}
    \label{fig:SQL_equivalence1}
\end{subfigure}
\hspace{6mm}
\begin{subfigure}[b]{.28\linewidth}
    \includegraphics[scale=0.4]{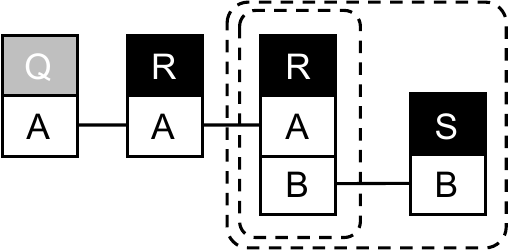}
	\vspace{2mm}
    \caption{}
    \label{Fig_relational_division_1}
\end{subfigure}
\hspace{38mm}
\phantom{x}
\begin{subfigure}[b]{.19\linewidth}
\vspace{1mm}
\begin{lstlisting}
SELECT DISTINCT R.A
FROM R
WHERE not exists
 (SELECT *
 FROM S, R AS R3
 WHERE R3.A = R.A
 AND not exists
  (SELECT *
  FROM R AS R2
  WHERE R2.B = S.B
  AND R2.A = R3.A))
\end{lstlisting}
\vspace{-7mm}
    \caption{}
    \label{fig:SQL_equivalence2}
\end{subfigure}
\hspace{6mm}
\begin{subfigure}[b]{.32\linewidth}
\vspace{1mm}	
    \includegraphics[scale=0.4]{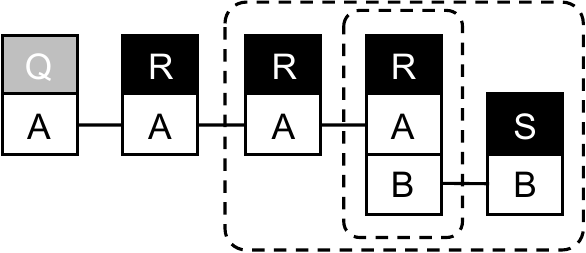}
	\vspace{2mm}
    \caption{}
    \label{Fig_relational_division_2}
\end{subfigure}
\hspace{2mm}
\begin{subfigure}[b]{.23\linewidth}
\vspace{1mm}	
\begin{lstlisting}
SELECT DISTINCT R.A
FROM R
WHERE R.A not in
 (SELECT R3.A
 FROM S, R AS R3
 WHERE (R3.A, S.B) not in
  (SELECT R2.A, R2.B
  FROM R AS R2))
\end{lstlisting}
\vspace{-7mm}
    \caption{}
    \label{fig:SQL_equivalence3}
\end{subfigure}
\hspace{35.5mm}
\vspace{-4mm}
\caption{\Cref{ex:relationalDivision}: Relational division in SQL (a)(c)(e) and as \diagrams\ (b)(d).
All 5 representations are \emph{logically equivalent}, but only the partitions
\{(a), (b)\}
and
\{(c), (d), (e)\} are also
\emph{pattern-isomorphic} (which is what we expect).
We illustrate in \cref{Fig_equivalence} how the tables in this query correspond to the equivalent queries in \diagrams, $\NDSQL$, $\NDRA$, $\DatalogND$, $\NDTRC$.
}
\label{fig:SQL_equivalence}
\end{figure*}

\begin{figure*}[t]
\centering	
	\begin{subfigure}[b]{\linewidth}
        \centering
	    \includegraphics[scale=0.36]{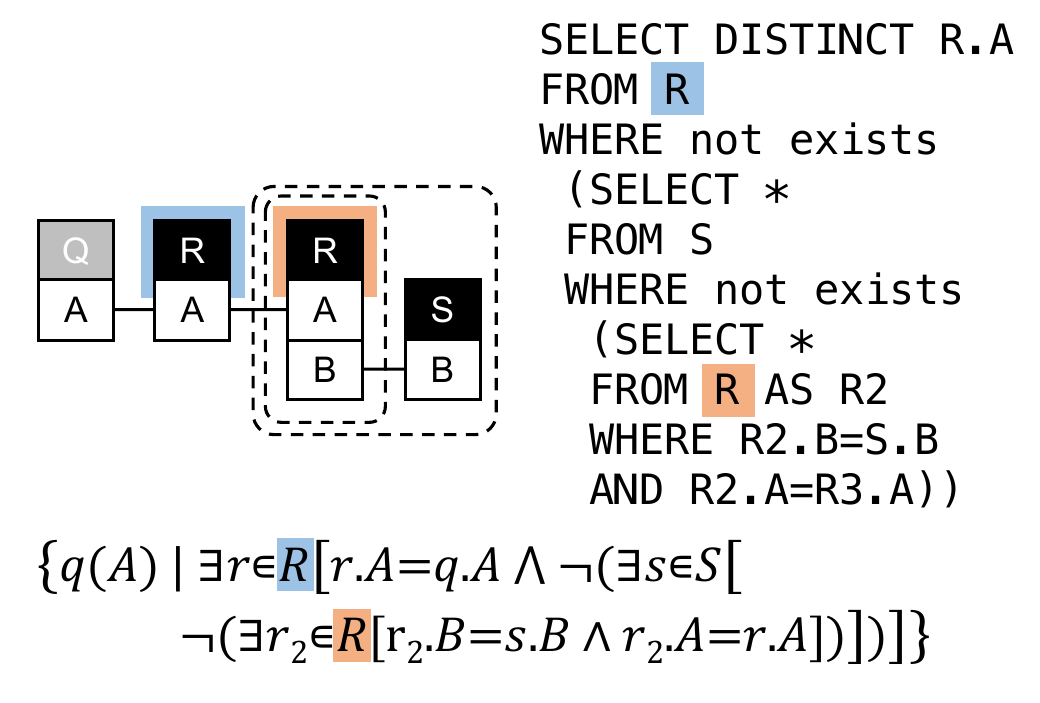}
		\vspace{-1mm}
	    \caption{}
	    \label{Fig_equivalence2}
	\end{subfigure}
    \\
	\begin{subfigure}[b]{\linewidth}
		\vspace{3mm}
        \centering
	    \includegraphics[scale=0.36]{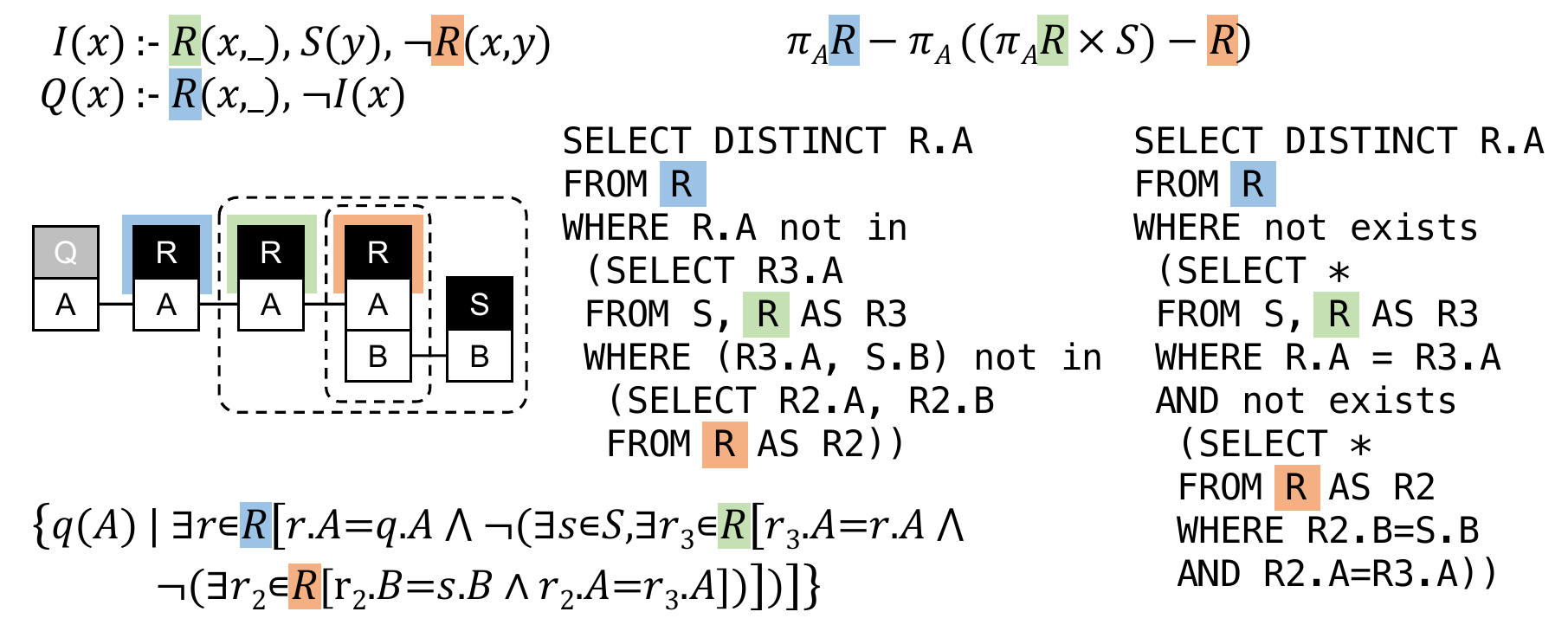}
		\vspace{-1mm}
	    \caption{}
	    \label{Fig_equivalence1}
	\end{subfigure}	
\caption{\Cref{ex:relationalDivision} and \cref{fig:SQL_equivalence} continued:
Two logically-equivalent sets (a) and (b) of relational division in 5 query languages 
(\diagrams, $\NDSQL$, $\NDRA$, $\DatalogND$, $\NDTRC$).
The queries in (a) use 2 occurrences of $R$, whereas the ones in (b) use 3 occurrences of $R$.
We highlight the two or three occurrences of $R$ across the different languages that can be mapped to each other 
according to our pattern isomorphism defined in \cref{def:isomorphism}.
We prove in \cref{appendix:proofs4} that for (a), there is no pattern-isomorphic representation in $\DatalogN$ (thus neither in $\RA$).
}
\label{Fig_equivalence}
\end{figure*}

\begin{example}[$\TRC$ and $\RA$ are not representation-equivalent]\label{ex:relationalDivision}
Assume a schema $R(A,B), S(B)$.
Consider the relational division asking for attribute values from $R.A$ that co-occur in $R$ with all attribute values from $S.B$.
The translation
into $\TRC$  is
\begin{align}
\begin{aligned}
\{ q(A) \mid 
	&\exists r \in R 	[q.A \equal r.A \wedge \neg (\exists s \in S[	\\
	&\neg (\exists r_2 \in R[r_2.B \equal s.B \wedge r_2.A \equal r.A])])] \}	
\end{aligned}
\label{TRC:division1}
\end{align}
The corresponding canonical $\SQL$ statement is shown in \cref{fig:SQL_equivalence1}.
Relational division expressed in primitive $\RA$ 
is
\begin{align}
	\pi_A R - \pi_A\big( (\pi_A R \times S) - R \big)
	\label{RA:division}
\end{align}
The translation into $\DatalogN$ uses two rules:
\begin{align}
\begin{aligned}
	I(x) 	& \datarule R(x,\_), S(y), \neg R(x,y). 		\\
	Q(x) 	& \datarule R(x,\_), \neg I(x). 		
\end{aligned}
\label{datalog}
\end{align}	
The atoms $R(x,\_)$ are needed for the safety condition of $\DatalogN$.
This translation is part of a standard proof for equivalence of expressiveness between $\RA$ and safe $\TRC$ in textbooks such as 
\cite{Ullman1988PrinceplesOfDatabase,DBLP:books/cs/Maier83,DBLP:books/aw/AbiteboulHV95}.

Now notice an arguably important difference between the three expressions:
$\TRC$ \cref{TRC:division1} uses the atom $R$ \emph{two} times, whereas
$\RA$ \cref{RA:division} and $\DatalogN$ \cref{datalog} use $R$ \emph{three} times.
It turns out that there is no way to represent relational division in primitive $\RA$ or $\DatalogN$ with only two occurrences of the $R$ symbol 
(see \cref{th:representations}).

There is, however, an alternative representation in $\TRC$ that preserves the $\RA$ structure with three occurrences of $R$:
\begin{align}
\begin{aligned}	
	\{ q(A) \mid \exists r \in R &[q.A \equal r.A \wedge 
		\neg (\exists s \in S, \exists r_3 \in R [r_3.A = r.A  		\\							
		&\hspace{0mm} \wedge \neg (\exists r_2 \in R[r_2.B \equal s.B \wedge r_2.A \equal r_3.A])])]  \}		\label{TRC:division2}
\end{aligned}
\end{align}
Notice that now there is a natural 1-to-1 correspondence between the atoms in 
$\TRC$ \cref{TRC:division2} and the atoms in $\RA$ \cref{RA:division}.
This correspondence is even more intuitive by mapping the correspondence between two logically-equivalent $\SQL$ statements
(\cref{fig:SQL_equivalence2,fig:SQL_equivalence3}) and $\RA$ \cref{RA:division}:
for example, lines 4--8 in \cref{fig:SQL_equivalence3} translate into 
the $\RA$ fragment ``$\pi_A\big( (\pi_A R \times S) - R \big)$'', 
which corresponds to the IDB predicate $\sql{Temp}(x)$ in $\DatalogN$ \cref{datalog}.

In other words, while all of these 7 queries are \emph{logically equivalent},
they partition into two disjoint sets that are ``\emph{pattern-isomorphic}'':
\begin{align*}
\textrm{Set 1} =
		& \{\RA \textrm{ \cref{RA:division}}, 
			\TRC \textrm{ \cref{TRC:division2}}, 
			\SQL \textrm{ \cref{fig:SQL_equivalence2}}, 
			\SQL \textrm{ \cref{fig:SQL_equivalence3}}, 
			\DatalogN \textrm{ \cref{datalog}} \} \\
\textrm{Set 2} =
		& \{\TRC \textrm{ \cref{TRC:division1}}, 
			\SQL \textrm{ \cref{fig:SQL_equivalence1}} \} 
\end{align*}	
This suggests that $\TRC$ and $\RA$ are not \emph{representation equivalent}.

We next prove the pattern-isomorphism between 
RA query 
\cref{RA:division}
and TRC query
\cref{TRC:division2} with our formalism.
First, write their dissociated queries
$q'_{\textrm{RA}}(R_1, R_2, S, R_3)$
and
$q'_{\textrm{TRC}}(R_1, S, R_2, R_3)$
with
\begin{align*}
	q'_{\textrm{RA}} = 
	\pi_A R_1 - \pi_A\big( (\pi_A R_2 \times S) - R_3 \big)
\end{align*}
\begin{align*}
	\{ q'_{\textrm{TRC}}(A) \mid \exists r \in R_1 &[q.A \equal r.A \wedge 
		\neg (\exists s \in S, \exists r_3 \in R_2 [r_3.A = r.A  		\\							
		&\hspace{0mm} \wedge \neg (\exists r_2 \in R_3[r_2.B \equal s.B \wedge r_2.A \equal r_3.A])])]  \}		
\end{align*}
Second, define the
homomorphism $h(R_1, R_2, S, R_3) = (R_1, S, R_2, R_3)$, which is injective and surjective and thus an isomorphism between the signature of the dissociated queries.
We can now easily verify that the dissociated queries are logically equivalent 
after composition with $h$:
$q'_{\textrm{RA}} \equiv q'_{\textrm{TRC}} \circ h$.
In other words:
\begin{align*}
q'_{\textrm{RA}}(R_1, R_2, S, R_3) \equiv
q'_{\textrm{TRC}}(h(R_1, R_2, S, R_3)) =
q'_{\textrm{TRC}}(R_1, S, R_2, R_3)
\end{align*}

\Cref{Fig_equivalence} illustrates the pattern isomorphism within two sets of queries with the color-highlighted $R$ atoms. 
Notice in 
\cref{Fig_equivalence2} the correspondences between the blue and orange highlighted tables $R$ across the SQL and the TRC statements and the \diagram.
Notice that the constraints between their ``A'' attributes (``R2.A=R.A'') make use of the nesting hierarchy: two levels of the ``not exists'' nesting hierarchy in SQL and two levels in the negation hierarchy in TRC.
Similarly, in \cref{Fig_equivalence1}
the two SQL variants use different syntactic constructs to represent the single negation hierarchy between R3 and R2. 
Then, see how Datalog represents this constraint without referring to the explicit attributes ``A'' but by positional reference and using the repeated variable $x$ together with ``not'' to represent the same logical constraint.
RA represents the same logical constraint by projecting attribute ``A'' from the green instance R on the left side of a cross product, before the set difference with the yellow instance $R$.
Despite the extreme syntactic variants and logical equivalence of all these queries, 
by defining individual pairwise isomorphisms between extensional tables, our
formalisms allow us to partition these queries into two sets within which the queries are pattern-isomorphic. 
\end{example}

\begin{example}[\cref{ex:allsailors} continued]
For completeness, we also translate here the statement (or Boolean query)
\emph{``All sailors reserve a red boat''}
into $\DatalogND$ and $\NDRA$. 
Boolean queries are usually not shown in $\RA$ but are straightforward:
The queries project away all attributes as last operation:
\begin{align*}
	\pi_{\emptyset} \big( \pi_{\mathit{sid}} \mathit{Sailor} 
	- \pi_{\mathit{sid}} (\mathit{Reserves} \Join (\sigma_{color='red'} \mathit{Boat}) ) \big)
\end{align*}
The translation into $\DatalogN$ uses the same pattern:
\begin{align*}
\begin{aligned}
	I(s) 	& \datarule \mathit{Reserves}(s, b, \_), \mathit{Boat}(b, \_, 'red'). 		\\
	Q	 	& \datarule \mathit{Sailor}(s,\_, \_, \_), \neg I(s). 		
\end{aligned}
\end{align*}	
Notice that all 4 languages use the same relational pattern for this query.
\end{example}

\section{More discussions for \texorpdfstring{\cref{sec:structureisomorphism}}{Section \ref{sec:structureisomorphism}}}

Our formalism is similar in spirit to \emph{edge-preserving graph homomorphisms}
that map two nodes in graph $G_1$ linked by an edge to two nodes in graph $G_2$ that are also linked by an edge.
In our pattern-preserving isomorphisms between queries,
the role of nodes is played by the table signatures in the queries
and the queries themselves play the role of the edges.
Notice also the difference in homomorphisms between conjunctive queries for determining query containment~\cite{Chandra:1977:OIC:800105.803397}: 
in that formalism, the role of nodes is played by variables (and constants)
and the relational atoms play the role of edges.
Also, notice from
\cref{ex:querystructureiso1}
that a simpler mapping between the set of tables used (instead of the repeated table references)
in two queries alone would not work.

Our definition---by design---does not distinguish query patterns based on
operator execution orders.
Also by design, our definition does not include any notion of views or intermediate tables. 
This is achieved by excluding Intensional Database Predicates (as in $\Datalog$) from the definition of table references.
We again illustrate the intuition for that design with examples.

\begin{example}[Join orders and views do not affect patterns]
\label{ex:querystructureiso2}
Assume that the edges of a directed graph are stored in a binary relation $E(A,B)$.
Consider a query  
returning nodes that form the starting point of a length-3 directed path.
In unnamed RA where indices replace attribute names \cite{DBLP:books/aw/AbiteboulHV95},
we can write the query in two different ways, one applying projections early
and the other late:
\begin{align*}
	q_1(E) &= \pi_{1}\sigma_{2=3 \wedge 4=5}(\h{E} \times \h{E} \times  \h{E}) \\
	q_2(E) &= \pi_{1}\sigma_{2=3}(\h{E} \times \pi_{1}\sigma_{2=3}(\h{E} \times  \pi_{1}\h{E})) 	
\end{align*}
We can also write these two queries in the more familiar named perspective of RA. 
These queries encode the same algebraic operations but are more verbose, since 
the named perspective of RA requires a rename operator $\rho$ to express the identical sequence of operations:
\begin{align*}
	q_1(E) &= \pi_{E.A}\sigma_{E.B=F.A \wedge F.B=G.A}(\h{E} \times \rho_{E \rightarrow F} 
				\h{E} \times  \rho_{E \rightarrow G} \h{E}) \\
	q_2(E) &= \pi_{E.A}\sigma_{E.B=F.A}(\h{E} \times \pi_{F.A}\sigma_{F.B=G.A}(
				\rho_{E \rightarrow F} \h{E} \times \rho_{\h{E} \rightarrow G} \pi_{A} \h{E})) 	
\end{align*}

All four $\RA$ query expressions use the same relational pattern according to our definition.
We believe it is essential to separate the notion of a relational
query pattern from concerns regarding operator execution order, because
the latter is not meaningful for queries written in
declarative relational languages.
To see that, 
consider 
the logically-equivalent query in $\Datalog$:
\begin{align*}
	Q_3(x) 	&\datarule \h{E}(x,y), \h{E}(y,z), \h{E}(z,w).
\end{align*}
It declaratively specifies what attributes the three tables need to be joined on,
but it does not specify any order or joins nor when projections happen.
Furthermore, notice that RA query $q_1$ does not even specify a concrete join order between the three table references
and the query can be seen as a shorthand for either of two sequences, more precisely written as $(E \times E) \times E$ or $E \times (E \times E)$
depending on the preferred definition of the shorthand.

For a similar reason, temporary tables such as Intensional Database Predicates in $\Datalog$ do not count as table references.
Thus, the following $\Datalog$ query uses
the same logical pattern (find three edges that join and keep the starting node),
even though it defines the intermediate intensional database predicate $I$:
\begin{align*}
	I(y)  	&\datarule \h{E}(y,z), \h{E}(z,w). \\
	Q_4(x) 	&\datarule \h{E}(x,y), I(y).
\end{align*}
\end{example}

Notice that it follows immediately that two pattern-isomorphic queries 
need to have the same number of table references.
We believe that such a pattern-preserving mapping between queries is important if we want to help readers understand the \emph{exact logic}
(the exact \emph{logical pattern}) behind a relational query, irrespective of the language it is written in.
In particular, if we want to help users understand the relational pattern of an existing relational query with diagrams (recall that we focus on set semantics and binary logic),
we must be able to create this 1-to-1 correspondence between the query and its diagrammatic representation.

\section{Drill-down for 
\texorpdfstring{\cref{SEC:STRUCTUREISOMORPHISM}}{Section \ref{SEC:STRUCTUREISOMORPHISM}}:
Limits of Datalog for representing patterns
}

We have shown earlier that $\DatalogND$ cannot represent all Query patterns from $\NDTRC$.
We next use another example to illustrate that
this limit of $\Datalog$ does not only appear with deeply nested queries;
it already appears for simply nested queries with inequality conditions
and is an immediate consequence of $\Datalog$'s safety conditions for built-in predicates.

\begin{figure}[tbp]
\small
\raggedright
\begin{subfigure}[b]{.18\linewidth}
\begin{lstlisting}
SELECT DISTINCT R.A
FROM R 
WHERE exists
 (SELECT *
 FROM S
 WHERE S.B = R.B)
\end{lstlisting}
\vspace{-7mm}
\caption{}
\label{Fig_DatalogLimits_a}
\end{subfigure}
\hspace{7mm}	
\begin{subfigure}[b]{.18\linewidth}
\begin{lstlisting}
SELECT DISTINCT R.A
FROM R 
WHERE not exists
 (SELECT *
 FROM S
 WHERE S.B = R.B)
\end{lstlisting}
\vspace{-7mm}
\caption{}
\label{Fig_DatalogLimits_f}
\end{subfigure}
\hspace{7mm}	
\begin{subfigure}[b]{.18\linewidth}
\begin{lstlisting}
SELECT DISTINCT R.A
FROM R 
WHERE not exists
 (SELECT *
 FROM S
 WHERE S.B < R.B)
\end{lstlisting}
\vspace{-7mm}
\caption{}
\label{Fig_DatalogLimits_k}
\end{subfigure}
\hspace{7mm}
\begin{subfigure}[b]{.18\linewidth}
\begin{lstlisting}
SELECT DISTINCT R.A
FROM R 
WHERE not exists
 (SELECT *
 FROM S, R as R2
 WHERE S.A < R2.A
 AND R2.A = R.A)
\end{lstlisting}
\vspace{-7mm}
\caption{}
\label{Fig_DatalogLimits_p}
\end{subfigure}
\begin{subfigure}[b]{.18\linewidth}		
\begin{align*}
	&	\{ q(A) \!\mid\! \exists r \in R, \exists s \in S [\\[-1mm]
	&	q.A \!=\! r.A \wedge r.B \!=\! s.B]\}	
\end{align*}
\vspace{-4mm}
\caption{}
\label{Fig_DatalogLimits_b}
\end{subfigure}
\hspace{4.3mm}
\begin{subfigure}[b]{.18\linewidth}		
\begin{align*}
	&	\{ q(A) \!\mid\! \exists r \in R[q.A \!=\! r.A  \\[-1mm]
	&	\wedge \neg (\exists s \in S [s.B \!=\! r.B])]\}	
\end{align*}
\vspace{-4mm}
\caption{}
\label{Fig_DatalogLimits_g}
\end{subfigure}
\hspace{1mm}
\begin{subfigure}[b]{.18\linewidth}		
\begin{align*}
	&	\{ q(A) \!\mid\! \exists r \in R[q.A \!=\! r.A  \\[-1mm]
	&	\wedge \neg (\exists s \in S [s.A \!<\! r.A])]\}	
\end{align*}
\vspace{-4mm}
\caption{}
\label{Fig_DatalogLimits_l}
\end{subfigure}
\hspace{0.3mm}
\begin{subfigure}[b]{.18\linewidth}		
\begin{align*}
	& \{ q(A) \!\mid\! \exists r \in R[q.A \!=\! r.A  \\[-1mm]
	& \wedge \neg (\exists s \in S, r_2 \in R [\\[-1mm]
	& s.A \!<\! r_2.A \wedge r_2.A \!=\! r.A])]\}	
\end{align*}
\vspace{-4mm}
\caption{}
\label{Fig_DatalogLimits_q}
\end{subfigure}
\begin{subfigure}[b]{.18\linewidth}
\begin{align*}
	R \Join S
\end{align*}
\vspace{-4mm}
\caption{}
\label{Fig_DatalogLimits_c}
\end{subfigure}
\hspace{7mm}
\begin{subfigure}[b]{.18\linewidth}
\begin{align*}
	R - S
\end{align*}
\vspace{-4mm}
\caption{}
\label{Fig_DatalogLimits_h}
\end{subfigure}
\hspace{7mm}
\begin{subfigure}[b]{.18\linewidth}
\vspace{2mm}
\centering
\h{\Huge\frownie{}}
\vspace{1mm}
\caption{}
\label{Fig_DatalogLimits_m}
\end{subfigure}
\hspace{7mm}
\begin{subfigure}[b]{.18\linewidth}
\begin{align*}
	&R - \big( \pi_{R.A} \sigma_{S.A<R.A} (R \times S) \big)
\end{align*}
\vspace{-4mm}
\caption{}
\label{Fig_DatalogLimits_r}
\end{subfigure}
\begin{subfigure}[b]{.18\linewidth}
\begin{align*}
	Q(x) 	& \datarule R(x), S(x).
\end{align*}
\vspace{-4mm}
\caption{}
\label{Fig_DatalogLimits_d}
\end{subfigure}
\hspace{7mm}
\begin{subfigure}[b]{.18\linewidth}
\begin{align*}
	Q(x) 	& \datarule R(x), \neg S(x).
\end{align*}
\vspace{-4mm}
\caption{}
\label{Fig_DatalogLimits_i}
\end{subfigure}
\hspace{7mm}
\begin{subfigure}[b]{.18\linewidth}
\centering
\h{\Huge\frownie{}}
\vspace{1mm}
\caption{}
\label{Fig_DatalogLimits_n}
\end{subfigure}
\hspace{7mm}
\begin{subfigure}[b]{.18\linewidth}
\begin{align*}
	I(x)	& \datarule R(x),  S(y), x\!>\!y. 		\\[-1mm]
	Q(x) 	& \datarule R(x), \neg I(x).
\end{align*}
\vspace{-4mm}
\caption{}
\label{Fig_DatalogLimits_s}
\end{subfigure}
\begin{subfigure}[b]{.18\linewidth}
    \includegraphics[scale=0.4]{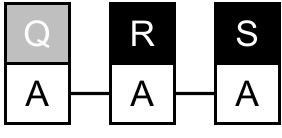}
	\vspace{-1mm}
    \caption{}
    \label{Fig_DatalogLimits_1}
    \label{Fig_DatalogLimits_e}
\end{subfigure}	
\hspace{7mm}
\begin{subfigure}[b]{.18\linewidth}
\includegraphics[scale=0.4]{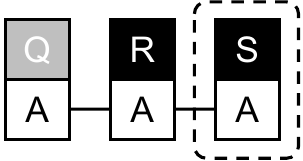}
\vspace{-1mm}
\caption{}
\label{Fig_DatalogLimits_2}
\label{Fig_DatalogLimits_j}
\end{subfigure}
\hspace{7mm}
\begin{subfigure}[b]{.18\linewidth}
\includegraphics[scale=0.4]{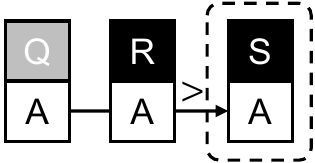}
\vspace{-1mm}
\caption{}
\label{Fig_DatalogLimits_o}
\label{Fig_DatalogLimits_3}
\end{subfigure}
\hspace{7mm}
\begin{subfigure}[b]{.18\linewidth}
\vspace{2mm}	
\includegraphics[scale=0.4]{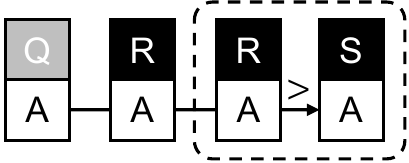}
\vspace{-5mm}
\caption{}
\label{Fig_DatalogLimits_t}
\label{Fig_DatalogLimits_4}
\end{subfigure}	
\caption{\Cref{ex:DatalogLimits}: 
The 5 rows show
$\NDSQL$, $\NDTRC$, $\NDRA$, $\DatalogND$, and \diagrams, respectively.
The first two column show $Q_1$ and $Q_2$, respectively. The last two columns show $Q_3$.
Notice that the $\NDSQL$, $\NDTRC$, and \diagram\
queries from the 3rd column have no pattern-isomorphic query
in $\NDRA$ or $\DatalogND$.
Instead, $\NDRA$ and $\DatalogND$ require an additional cross-join with $R.A$,
which is shown in the 4th column.
}
\label{fig:DatalogLimits}
\end{figure}

\begin{example}[Limits of Datalog]
Consider two unary tables $R(A)$ and $S(A)$ and three questions:
\label{ex:DatalogLimits}
\begin{align*}	
	&Q_1: \textrm{Find values from R that also appear in S.} \\
	&Q_2: \textrm{Find values from R that do not appear in S} \\
	&Q_3: \textrm{Find values from R for which no smaller value appears in S.} 
\end{align*}
These queries are shown in $\NDSQL$, $\NDTRC$, $\NDRA$, $\DatalogND$, and as \diagram\ in \cref{fig:DatalogLimits}, with 
$Q_1$ in the 1st column (\subref{Fig_DatalogLimits_a}, \subref{Fig_DatalogLimits_b}, \subref{Fig_DatalogLimits_c}, \subref{Fig_DatalogLimits_d}, \subref{Fig_DatalogLimits_e}), 
$Q_2$ in the 2nd (\subref{Fig_DatalogLimits_f}, \subref{Fig_DatalogLimits_g}, \subref{Fig_DatalogLimits_h}, \subref{Fig_DatalogLimits_i}, \subref{Fig_DatalogLimits_j}), and
$Q_3$ in the 3rd (\subref{Fig_DatalogLimits_k}, \subref{Fig_DatalogLimits_l}, \subref{Fig_DatalogLimits_m}, \subref{Fig_DatalogLimits_n}, \subref{Fig_DatalogLimits_o}).

Notice in the third column ($Q_3$) that the $\NDSQL$ (\subref{Fig_DatalogLimits_k}), $\NDTRC$ (\subref{Fig_DatalogLimits_l}), and \diagram\ (\subref{Fig_DatalogLimits_o})
queries have no pattern-isomorphic query
in $\NDRA$ (\subref{Fig_DatalogLimits_m}) or $\DatalogND$ (\subref{Fig_DatalogLimits_n}).
The safety condition of $\Datalog$ requires each variable to appear in a non-negated atom. 
This criterion requires a cross-join with the domain of $R.A$ in a separate rule before the negation can be applied on an equality predicate.
For the same reason, $\NDRA$ cannot apply the set difference directly and also requires an additional cross-join with $R.A$.
Even extending the available operators with an antijoin does not change this
since antijoins are only defined for equality conditions~\cite{DBLP:books/mg/SKS20}.

The 4th column (\subref{Fig_DatalogLimits_p}, \subref{Fig_DatalogLimits_q}, \subref{Fig_DatalogLimits_r}, \subref{Fig_DatalogLimits_s}, \subref{Fig_DatalogLimits_t}) shows the resulting resulting $\NDRA$ and $\DatalogND$ 
queries together with their pattern-isomorphic queries in 
 $\NDSQL$, $\NDTRC$, and as \diagrams.
\end{example}

\section{
\texorpdfstring{\cref{SEC:COMPLETENESS}}{Section \ref{SEC:COMPLETENESS}}:
Why disjunctions are harder to represent}
\label{subsec:disjunchard}

\begin{figure}[t]
\begin{subfigure}[b]{1\linewidth}
	\centering
	{\includegraphics[scale=0.25]{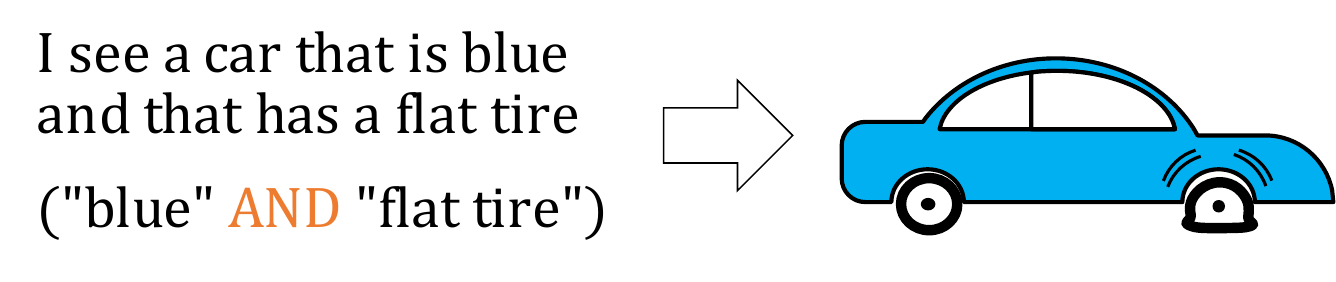}\hspace{22mm}\phantom{x}}
    \caption{}
	\label{Fig_blue_car_a}	
\end{subfigure}
\begin{subfigure}[b]{1\linewidth}
	\centering	
    \includegraphics[scale=0.25]{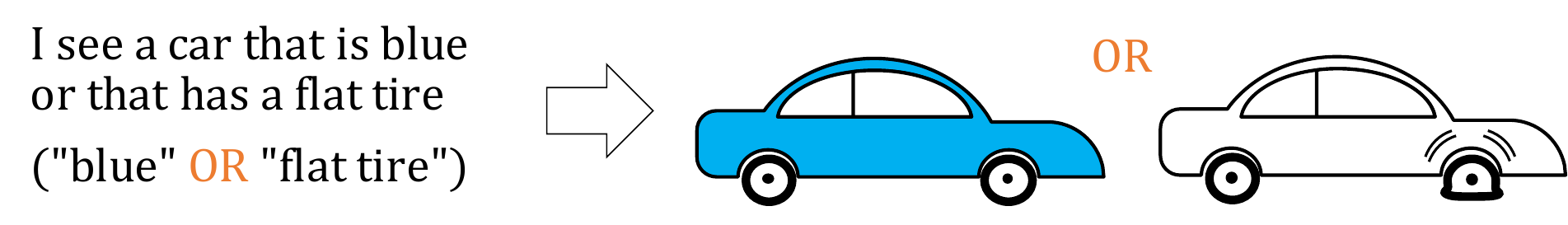}
    \caption{}
	\label{Fig_blue_car_b}	
\end{subfigure}
\caption{%
    A single situation (concrete arrangement of items) can only show conjunctive information~\cite{Shin:2002}, i.e., the car is blue.
    To show that the car is blue OR it has a flat tire, you need two situations and a visual construct to show disjunction (here, just the text `\textcolor{carfigORcolor}{OR}').
}
\label{Fig_blue_car}  
\end{figure}

It is useful to understand why disjunctions are more difficult to represent in a diagrammatic representation.
As a motivating example, assume Alice calls Bob and tells him ``I see a blue car that has a flat tire.''
What is the mental image that Bob has from this information? 
It is a car with two conditions:
it is blue, and it has a flat tire as in \cref{Fig_blue_car_a}.
Next, assume that Alice instead tells Bob
``I see a car that is either blue or that has a flat tire.''
What is the mental image that Bob has from that information? 

There is no single mental image that could capture that situation. 
Bob needs to imagine two \emph{different} images.
If Bob sees one image with two different cars (one blue, the other with a flat tire), 
then he actually sees \emph{two separate cars}.
Bob needs to add some additional visual symbol representing the logic that those are two \emph{different overlapping possible worlds}.
In the words of Shin~\cite{Shin:2002}, ``any situation'' (think of a concrete arrangement of items) ``can only display conjunctive information.''
This diagrammatic representation problem is not as apparent in text:
``\sql{Car.color = `blue' OR Car.tire = `flat'}''.\footnote{To provide some additional intuition,
recall that \emph{conjunctions} of selections can be simply modeled as a concatenation of simple selections,
e.g.\ $\sigma_{C_1 \wedge C_2}(R)$ is the same as $\sigma_{C_1}(\sigma_{C_2}(R))$. Thus conjunctions are an inherently more natural logical connective 
than disjunctions; disjunctions cannot be represented without additional visual symbols.}
Shin goes further and claims that ``Any diagrammatic system that seeks to represent disjunctive information needs to bring in an artificial syntactic device with its own convention'' \cite{shin_1995}.

\section{Proofs For \texorpdfstring{\cref{sec:completeness}}{Section \ref{sec:completeness}}}
\label{app:proofrelationalcompleteness}

\begin{proof}[Proof \cref{th:completeness}]
Given a safe $\TRC$ expression,
we pull any existential quantifier as early as possible: either to be at the start of the query or directly following a negation operator. 

First, consider a nested query with disjunctions in the $\sql{WHERE}$ conditions,
possibly nested with conjunctions.
Rewrite the conditions as DNF, i.e.\ as
\begin{align*}
	&\neg (\exists \vec r \in \vec R[c_1 \vee c_2  \cdots \vee c_k]) \\
\intertext{Here,
$\vec r$ is a set of table variables, $\vec R$ a set of tables, 
and each $c_i$ is a conjunction of guarded predicates.
Next, rewrite it as:}
	&\neg (\exists \vec r_1 \in \vec R[c_1]) \wedge 
		\neg (\exists \vec r_2 \in \vec R[c_2]) \cdots \wedge
		\neg (\exists \vec r_k \in \vec R[c_k])	
\end{align*}
This fragment is in $\NDTRC$ and can be visualized by $\diagramsND$.

Second, for remaining disjunctions in the top query $q_0$, rewrite the query as a union over queries without disjunction:
\begin{align*}
\phantom{=} & \{q(\vec A) \mid 
		\exists \vec r_1 \in \vec R_1[c_1] \vee 
		\exists \vec r_2 \in \vec R_2[c_2] \cdots \vee
		\exists \vec r_k \in \vec R_k[c_k] \}	\\
= &\{ q(\vec A) \mid 
	\exists \vec r \in \vec R[c_1] \}  \cup
  \{ q(\vec A) \mid 
	\exists \vec r \in \vec R[c_1] \} \cdots \cup
  \{ q(\vec A) \mid 
	\exists \vec r \in \vec R[c_k] \}			
\end{align*}
Since each individual query is in $\NDTRC$, 
the whole query can be represented by \diagrams\
by representing each individual $\NDTRC$ in a union cell. 
\end{proof}

\section{More examples for 
\texorpdfstring{\cref{sec:completeness}}{Section \ref{sec:completeness}}}

\begin{figure}[t]
\centering	
\begin{subfigure}[b]{.24\linewidth}		
\begin{lstlisting}
SELECT DISTINCT S.sname
FROM Sailor S, Reserves R
WHERE S.sid = R.sid
AND not 
 (not exists
  (SELECT *
  FROM Boat B
  WHERE color='red'
  AND R.bid = B.bid)
 AND not exists
  (SELECT *
  FROM Boat B
  WHERE color='blue'
  AND R.bid = B.bid)  
\end{lstlisting}
\vspace{-6mm}
\caption{}
\label{fig:Sailor_disjunction_a}
\end{subfigure}	
\hspace{1mm}
\begin{subfigure}[b]{.45\linewidth}
	\centering
    \includegraphics[scale=0.4]{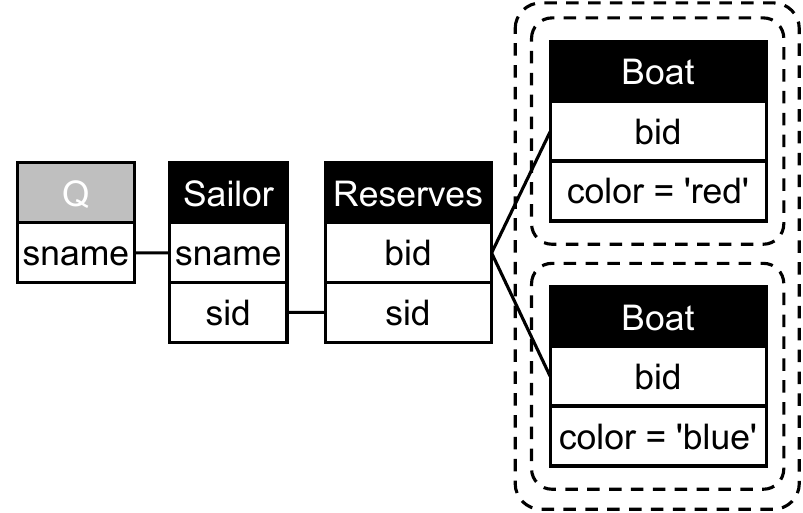}
    \caption{}
    \label{fig:Sailor_disjunction_b}
\end{subfigure}	
\caption{\Cref{ex:red_or_blue}:
The Query 
``Find sailors who reserve a red or blue boat'' 
can be represented with double negation 
in the non-disjunctive fragments of $\NDSQL$ (a) and \diagramNDnomath\ (b).}
\label{fig:Sailor_disjunction}
\end{figure}

\begin{example}[Red or blue]
Consider the following $\NDTRC$ query asking for sailors 
who have reserved a red or a blue boat:\label{ex:red_or_blue}		%
\begin{align*}
	\{ 
	& q(\sql{sname}) \mid \exists s \in \sql{Sailor}, r \in \sql{Reserves} [q.\sql{sname} = s.\sql{sname} \,\wedge  	\notag\\
	& s.\sql{sid} = r.\sql{sid} \wedge \h{\exists b \in \sql{Boat}[}										\\	%
	& \hspace{2mm}	b.\sql{bid} = r.\sql{bid} \wedge 
		(b.\sql{color} = \sql{`red'} \;\h{\vee}\; b2.\sql{color} = \sql{`blue'} ) \h{]} ] \}						\notag
\end{align*}
Using De Morgan's Law $(A \vee B) = \neg(\neg A \wedge \neg B)$ applied to quantifiers, 
we can transform the disjunction into double-negation with conjunction.
This transformation comes at the cost of repeated uses of extensional tables and is thus \emph{not pattern-preserving}:
\begin{align}
	\{ 
	& q(\sql{sname}) \mid \exists s \in \sql{Sailor}, r \in \sql{Reserves} [q.\sql{sname} = s.\sql{sname} \,\wedge  				\notag\\
	& s.\sql{sid} = r.\sql{sid} \wedge \h{\neg \big(}																\label{eq:red_or_blue_NDTRC}\\
	& \hspace{2mm}	\h{\neg(\exists b1 \in \sql{Boat}[}b1.\sql{bid} = r.\sql{bid} \wedge b1.\sql{color} = \sql{`red'} \h{])} \,\h{\wedge} 		\notag\\
	& \hspace{2mm}	\h{\neg(\exists b2 \in \sql{Boat}[}b2.\sql{bid} = r.\sql{bid} \wedge b2.\sql{color} = \sql{`blue'} \h{]) \big)} ] \}	 	\notag
\end{align}
\Cref{fig:Sailor_disjunction_a} shows 
\cref{eq:red_or_blue_NDTRC}
translated into canonical
$\NDSQL$
and 
\cref{fig:Sailor_disjunction_b}
its translation into a \diagram.
Notice how the non-disjunctive fragment repeats the boats table twice.
\end{example}

\begin{figure}[t]
\centering
\begin{subfigure}[b]{0.22\linewidth}		
\begin{lstlisting}
(SELECT DISTINCT R.A
FROM R, S, T
WHERE R.B > 5
AND R.A = S.A)
UNION
(SELECT DISTINCT R.A
FROM R, S, T
WHERE R.B > 5
AND R.A = T.A)
\end{lstlisting}
\vspace{-7mm}
    \caption{}
    \label{fig:disjunctionSQLor2}
\end{subfigure}	
\hspace{8mm}
\begin{subfigure}[b]{0.3\linewidth}
    \includegraphics[scale=0.39]{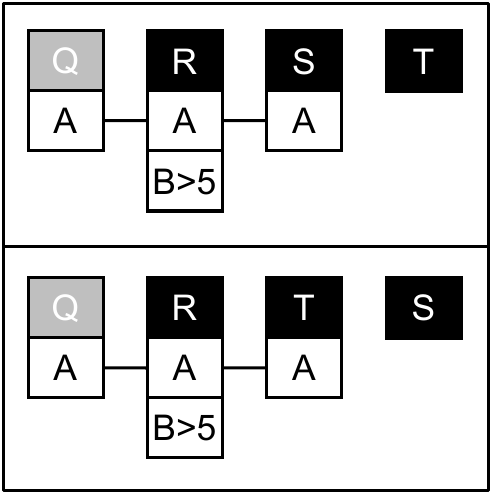}
	\vspace{-1mm}
    \caption{}
    \label{Fig_disjunction_or_QV}
\end{subfigure}	
\hspace{40mm}
\begin{subfigure}[b]{0.22\linewidth}		
\vspace{1mm}
\begin{lstlisting}
SELECT DISTINCT R.A
FROM R
WHERE R.B > 5
AND not (
  not exists
    (SELECT *
    FROM S, T
    WHERE S.A = R.A)
  AND not exists
    (SELECT *
    FROM S, T
    WHERE T.A = R.A))
\end{lstlisting}
\vspace{-7mm}
    \caption{}
    \label{fig:disjunctionSQLor2_pushed}
\end{subfigure}	
\hspace{8mm}
\begin{subfigure}[b]{0.3\linewidth}
    \includegraphics[scale=0.39]{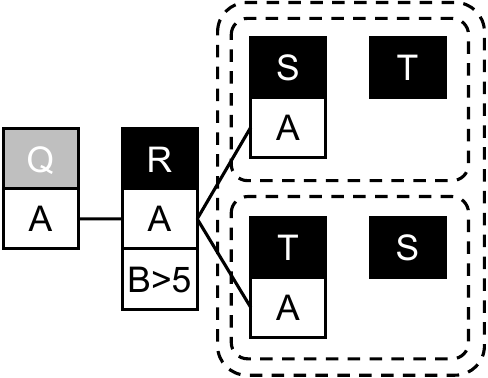}
	\vspace{-1mm}
    \caption{}
    \label{Fig_disjunction_or_QV_pushed}
\end{subfigure}	
\hspace{40mm}
\begin{subfigure}[b]{0.22\linewidth}		
\vspace{1mm}
\begin{lstlisting}
SELECT DISTINCT R.A
FROM R
WHERE R.B > 5
AND not exists 
  (SELECT *
  FROM R R2
  WHERE R2.A = R.A
  AND not exists
    (SELECT *
    FROM S, T
    WHERE S.A = R2.A)
  AND not exists
    (SELECT *
    FROM S, T
    WHERE T.A = R2.A))
\end{lstlisting}
\vspace{-7mm}
    \caption{}
    \label{fig:disjunctionSQLor2_pushed_forRA}
\end{subfigure}	
\hspace{8mm}
\begin{subfigure}[b]{0.3\linewidth}
    \includegraphics[scale=0.39]{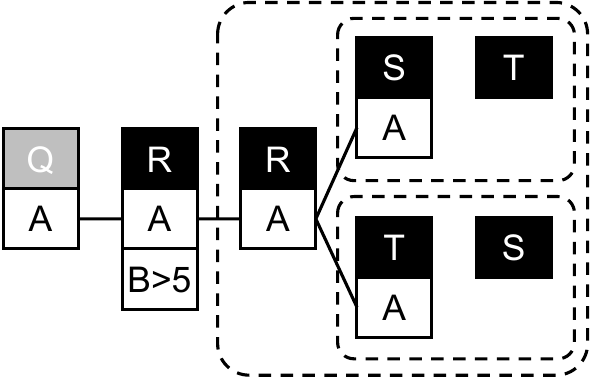}
	\vspace{-1mm}
    \caption{}
    \label{Fig_disjunction_or_QV_pushed_forRA}
\end{subfigure}	
\caption{Illustrations for \Cref{ex:disjunction_threeex} on replacing disjunctions:
(a, b) using union cells to replace disjunctions;
(c, d) using DeMorgan to replace disjunctions by double-negation and conjunction;
(e, f) further changing the table signatures to allow pattern-isomorphic representations
in $\DatalogND$ and $\NDRAA$.}
\label{fig:disjunction_threeex}
\end{figure}

\begin{example}[Union of queries]
\label{ex:disjunction_threeex}
Consider three tables $R(A,B)$, $S(A)$, and $T(A)$ and the running query from 
\cref{sec:nondisjunctivefragment}.
We can replace disjunction by pulling it to the root and replacing the query with a union of non-disjunctive 
$\NDTRC$ queries:
\begin{align*}
\phantom{=} &\h{\{} q(A) \mid
	\exists r \in R, \exists s \in S, \exists t \in T 
	[q.A \equal r.A \wedge r.B > 5 \;\wedge 
	\\ &\hspace{2mm} 	
	(r.A \equal s.A \;\h{\vee}\; 
	r.A \equal t.A)] \h{\}}	\\
= &\h{\{} q(A) \mid 
	\exists r \in R, \exists s \in S, \exists t \in T 
	[q.A \equal r.A \wedge r.B > 5 \wedge r.A \equal s.A] \h{\}}  \;\h{\cup}\; \\
  &\h{\{} q(A) \mid 
	\exists r \in R, \exists s \in S, \exists t \in T 
	[q.A \equal r.A \wedge r.B > 5 \wedge r.A \equal t.A] \h{\}}		
\end{align*}
\cref{fig:disjunctionSQLor2} shows the representation-equivalent $\SQL$ query.
\cref{Fig_disjunction_or_QV} shows this union of two separate $\diagramsND$ queries,
each in a separate union cell, and each with the same set attributes in the output table.
Notice that we cannot leave away the non-connected tables in the individual queries; if any of the tables are empty, 
then the query needs to return an empty result.

We can also rewrite this query directly as a non-disjunctive $\NDTRC$ query by using DeMorgan and repeating table references:
\begin{align*}
\phantom{=} &\{ q(A) \mid
	\exists r \in R, \h{\exists s \in S, \exists t \in T} 
	[ q.A \equal r.A \wedge r.B > 5 \;\wedge 
	\\ &\hspace{2mm} 
	\h{(} r.A \equal s.A \;\h{\vee}\; 
	r.A \equal t.A \h{)} ] \}	\\
= &\{ q(A) \mid
	\exists r \in R 
	[q.A \equal r.A \wedge r.B > 5 \;\wedge 
	\\ &\hspace{2mm} 
	\h{\exists s \in S, \exists t \in T [}r.A \equal s.A \;\h{\vee}\; 
	r.A \equal t.A \h{]} ] \}	\\	
= &\{ q(A) \mid
	\exists r \in R
	[q.A \equal r.A \wedge r.B > 5 \;\wedge 
	\\ &\hspace{2mm} 
	\h{(\exists s \in S, \exists t \in T [}r.A \equal s.A \h{]} \;\h{\vee}\; 
	\h{\exists s \in S, \exists t \in T [}r.A \equal t.A \h{])} ] \}	\\		
= &\{ q(A) \mid
	\exists r \in R
	[q.A \equal r.A \wedge r.B > 5 \;\wedge 
	\\ &\hspace{2mm} 	
	\h{\neg( \neg (\exists s \in S, \exists t \in T [}r.A \equal s.A \h{])} \;\h{\wedge}\; 
	\h{\neg (\exists s \in S, \exists t \in T [}r.A \equal t.A \h{]))} ] \}			
\end{align*}
\Cref{fig:disjunctionSQLor2} shows the pattern-isomorphic $\NDSQL$ query, and
\cref{Fig_disjunction_or_QV} shows the pattern-isomorphic $\diagramsND$ query.

Based on our results, this query cannot be represented 
with a pattern-isomorphic
$\NDRAA$ nor $\DatalogND$ queries.
But we can represent it by increasing the table signature:
following the steps from \hyperref[th:equivalence:trc-datalog]{part (4)}
in the proof of \cref{th:equivalence}:
\begin{align*}
\phantom{=} &\{ q(A) \mid
	\exists r \in R
	[q.A \equal r.A \wedge r.B > 5 \;\wedge 
	\\ &\hspace{2mm} 	
	\neg(
	\h{\exists r_2 \in R[r_2.A \equal r.A \;\wedge}  
	\\ &\hspace{2mm} 	
	\neg (\exists s \in S, \exists t \in T 
		[\h{r_2}.A \equal s.A ]) \wedge 
	\neg (\exists s \in S, \exists t \in T 
		[\h{r_2}.A \equal t.A  ]) \h{]} ) ] \}			
\end{align*}
This query now can be represented in pattern-isomorphic
$\NDRAA$ and $\DatalogND$ queries as follows:
\begin{align*}
\begin{aligned}
	I1(x) 	& \datarule S(x), T(\_). 		\\
	I2(x) 	& \datarule S(\_), T(x). 		\\	
	I3(x) 	& \datarule R(x,\_), \neg I1(x), \neg I2(x). 		\\
	Q(x) 	& \datarule R(x,y), \neg I3(x), y>5. 		
\end{aligned}
\end{align*}	
\begin{align*}
	\pi_A (\sigma_{B>5}R) - 
	((R \antijoin_{R.A=S.A}
	\pi_{S.A} (S \times T))
	\antijoin_{R.A=T.A}
	\pi_{T.A} (S \times T))
\end{align*}
\Cref{fig:disjunctionSQLor2_pushed_forRA,Fig_disjunction_or_QV_pushed_forRA}
show their pattern-isomorphic
representations in 
$\NDSQL$
and
$\diagramsND$.
\end{example}

\section{Textbook analysis (\texorpdfstring{\cref{SEC:TEXTBOOKANALYSIS}}{Section \ref{SEC:TEXTBOOKANALYSIS}})}
\label{appendix:textbook}

We give here more details on the 59 queries analyzed in 
\cref{sec:textbookanalysis}.
These queries and our analysis is available on OSF.\footnote{Queries: \url{https://osf.io/u7c4z/}}

\introparagraph{Approach}
We identified 5 popular database textbooks~
\cite{cowbook:2002, DBLP:books/mg/SKS20, Elmasri:dq, date2004introduction, ConnollyBegg:2015}
that include chapters on relational calculus. 
Textbooks were considered which the authors of this study had previously used examples from in their own database classes. 
A 6th popular textbook (``the complete book'' from Stanford~\cite{Garcia-MolinaUW2009:DBSystems})
is not included since it does not contain a section on relational calculus.

In each of the 5 books, we identified chapters that illustrate relational calculus queries with examples and extracted all example queries from these chapters. 
In case a query is given multiple times in pattern-isomorphic variants 
(e.g.\ first in Tuple Relational Calculus and then again in Domain Relational Calculus), we list the query only once in its earlier representation.
In addition, we also identified chapters on SQL queries and extracted those queries that have a logically-equivalent representation in relational calculus (thus no aggregates, arithmetic attributes, or outer joins). 

In total we identified 59 queries across those 5 textbooks. For each of the 5 languages 
\diagrams, $\RA$, $\Datalog$, 
$\queryviz$~\cite{DanaparamitaG2011:QueryViz}, and 
$\QBE$~\cite{DBLP:journals/ibmsj/Zloof77},
we analyzed which of those queries can be expressed \emph{in a pattern-isomorphic representation}.\footnote{Recall that since all considered query languages (except \queryviz) are relationally complete, all 59 queries have logically-equivalent representations in each  language. The question here is whether the languages have \emph{pattern-isomorphic} representations, not merely logically-equivalent representations.}
For the more interesting queries (e.g.\ those that have pattern-isomorphic representations in \diagrams\ but not in \RA) we included the \diagrams\ representation in the documentation.

\introparagraph{Detailed analysis}
(1) From the ``cow book'' by Ramakrishnan and Gehrke~\cite{cowbook:2002}, 
we extracted 25 queries from chapters 4.3.1 on \TRC, 4.3.2 on \DRC, and 5.2-5.4 on SQL. 
The number of queries that have pattern-isomorphic representations are:
24 for \diagrams,
22 for $\queryviz$ (it has no union operator and can only visualize a strict subset of the non-disjunctive fragment),
19 in \QBE, and
18 in $\RA$ or \Datalog.
The number for $\RA$ increases to 19 if the antijoin is added as an additional primitive relational operator (\cref{appendix:antijoins}).

(2) From the ``sailboat book" by Silberschatz, Korth, and Sudarshan~\cite{DBLP:books/mg/SKS20},
we extracted 8 queries from chapters 27.1 (on \TRC) and 27.2 (on \DRC).
The number of queries that have pattern-isomorphic representations are:
8 for \diagrams, and
7 for \queryviz, \QBE, $\RA$ and \Datalog.

(3) From the ``stone formation book'' by Elmasri and Navathe~\cite{Elmasri:dq}, we extracted
9 unique queries from chapter 8.6 (on TRC) and found only pattern-isomorphic queries in chapter 8.7 (on DRC).
The number of queries that have pattern-isomorphic representations are:
8 for \diagrams, \queryviz, \RA, \QBE, 
and 7 for \Datalog.

(4) From Date's book~\cite{date2004introduction}, we extracted 9 unique queries from chapter 8.3 (on TRC). 
The number of queries that have pattern-isomorphic representations are:
8 for \diagrams\ and \queryviz,
7 for $\RA$ and \QBE, 
and 6 for \Datalog.

(5) From the ``bookshelf book'' by Connolly and Begg~\cite{ConnollyBegg:2015}, we extracted
8 unique queries from section 5.2 (on TRC). 
All queries have pattern-isomorphic representations in all languages.

\section{Controlled User Experiment (\texorpdfstring{\Cref{SEC:USERSTUDY}}{Section~\ref{SEC:USERSTUDY}})}
\label{appendix:controlled-experiment}

This section gives additional details on the user study (\cref{appendix:controlled-experiment:details}) and provides links to 
the supplemental materials including preregistration, tutorial, actual stimuli, stimuli-creation code, collected data, and analysis code and results (\cref{appendix:controlled-experiment:links}).

\subsection{Additional details on experimental design}
\label{appendix:controlled-experiment:details}

\begin{figure}[t]
    \centering	
	\fbox{
    \includegraphics[scale=0.39]{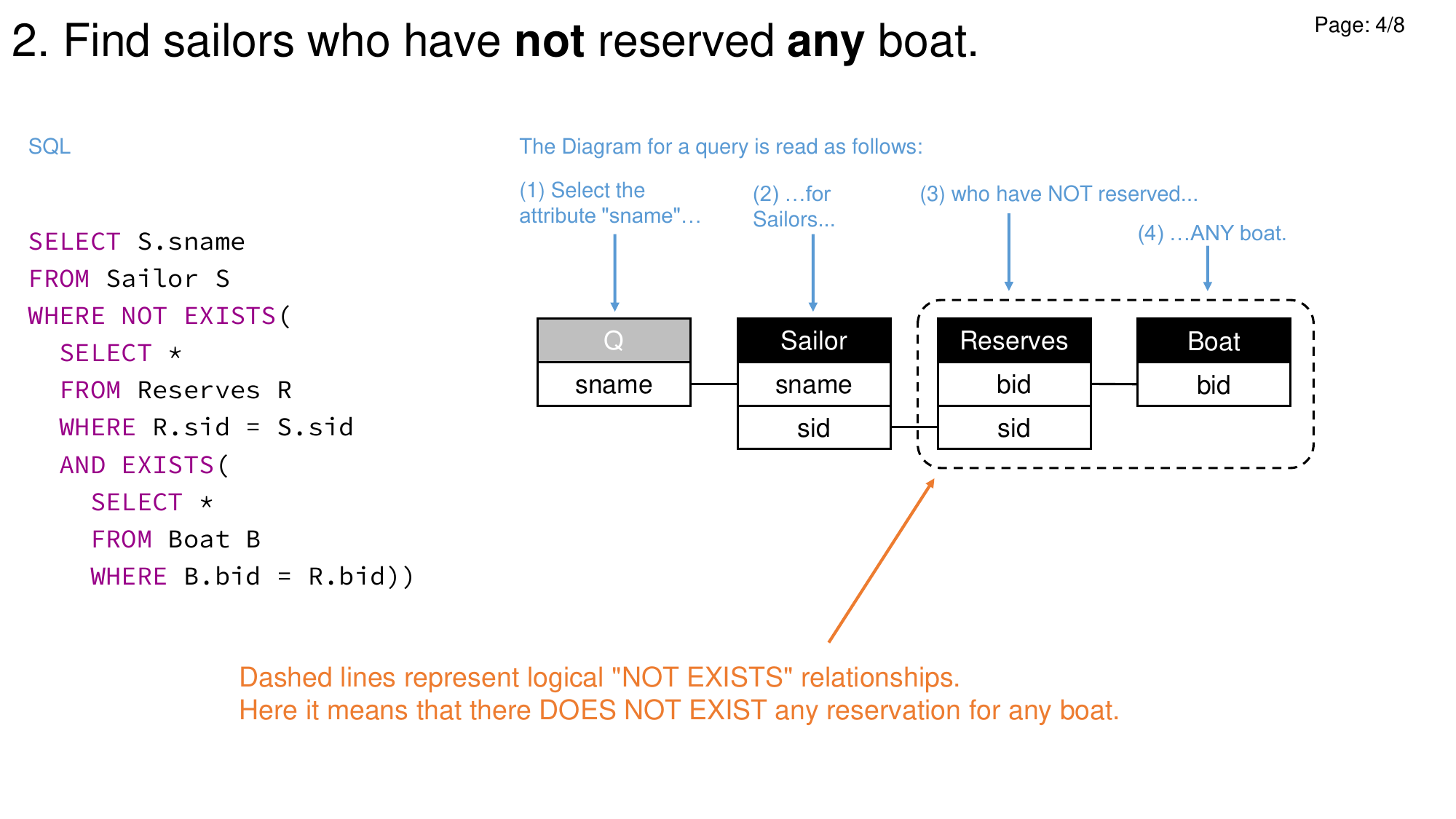}
	}
    \caption{One page of the 8-page tutorial for user study participants.}
    \label{Fig_tutorial_page_4}
\end{figure}

\introparagraph{Procedure}
After a short tutorial (see \cref{Fig_tutorial_page_4}), participants are asked to answer 32 questions, each with a different schema found or adapted from common textbooks.
Each question asks the participant to select which of 4 plain-English patterns is similar to the ones listed in \cref{sec:userstudy} 
correctly matches the shown query (but each question applied to a different schema, such as \emph{Students taking classes}, \emph{Actors playing in movies}, \emph{Suppliers supplying parts}, etc.).
The correct answer is provided after each question so participants \emph{can learn from mistakes}.
We analyze the time a participant spends on the question itself (not learning from the answer) 
as well as whether their chosen answer is correct.

\introparagraph{Randomization and Counterbalancing}
To reduce ordering effects, we start half the participants using SQL (group 1)
and the other using \diagrams\ (group 2), 
after which participants alternate conditions with each question 
(see \cref{SEC:USERSTUDY} and \Cref{Fig_study_design}).

The total number of possible sequences for each condition and for each half
is the multiset permutation $\binom{8}{2, 2, 2, 2} = 2520$.\footnote{\url{https://en.wikipedia.org/wiki/Permutation\#Permutations\_of\_multisets}: Each condition is shown 8 times per half. Each of the 4 patterns is shown 2 times.}
The total number of treatments is thus $2520^4$ (each of 2 conditions and each of 2 halves chosen independently from the $2520$ possible sequences, irrespective of which conditions are used).
We \emph{sample a new treatment (group and sequence) for each participant}.
I.e.\ recording the correct answers to one's treatment, and then sharing that information with another worker will not be very helpful to that worker.
The chance that, among $n$ participants, any two share the same questions (irrespective of condition) is approximately 
$1-e^{\frac{n^2/2}{2520^4}}$, which is around $1/5000$
for $n=50$.\footnote{\url{https://en.wikipedia.org/wiki/Birthday\_problem\#Approximations}}
In our study, none of the 171 workers who viewed the consent form were randomly assigned to the same treatment.

\begin{figure}[t]
    \centering
    \begin{subfigure}[t]{0.65\linewidth}
        \fbox{\includegraphics[width=\textwidth]{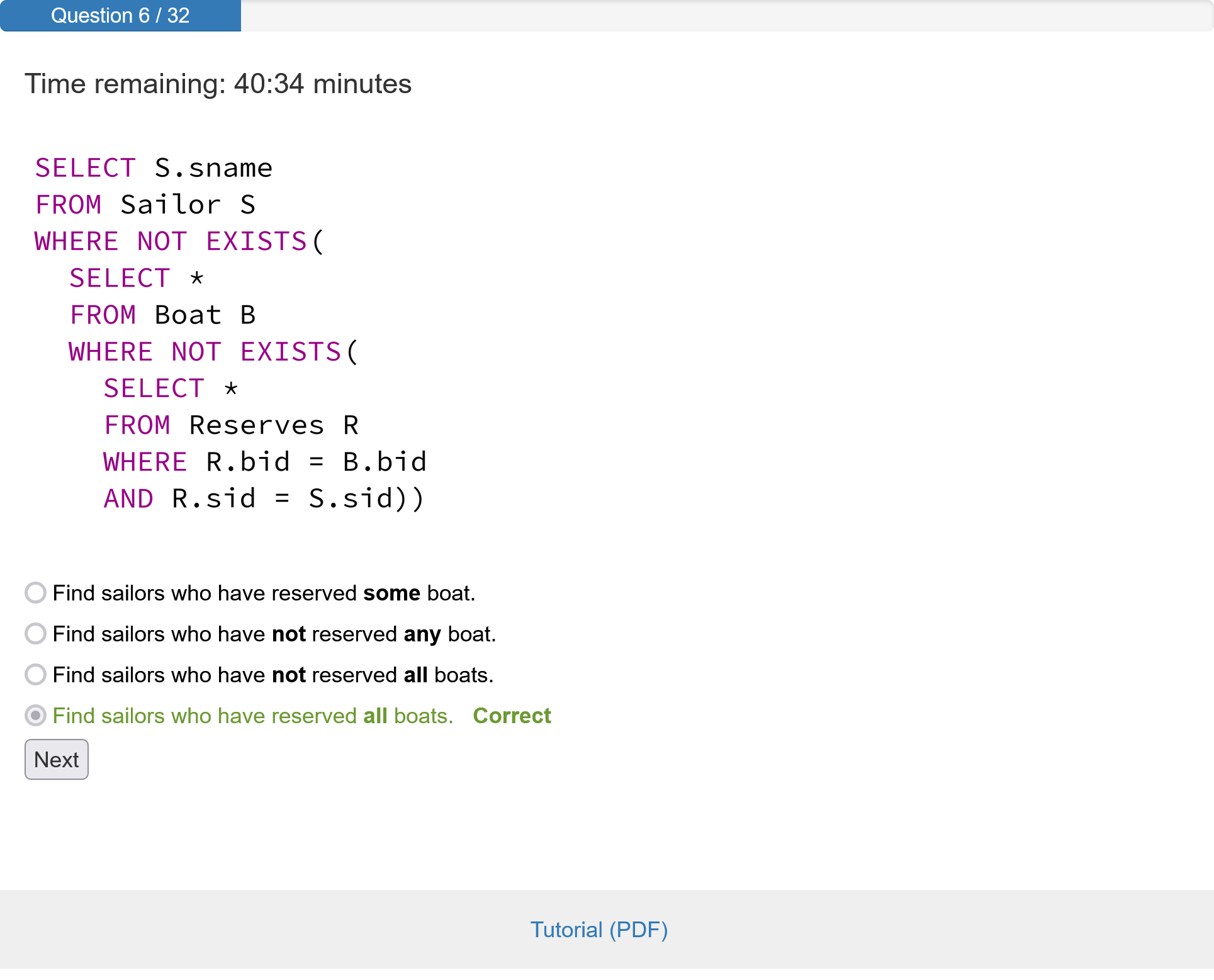}}
        \caption{$\SQL$}
        \label{fig:sailors_p4_sql_ui}
    \end{subfigure}
    \begin{subfigure}[t]{0.65\linewidth}
		\vspace{2mm}
        \centering
        \fbox{\includegraphics[width=\textwidth]{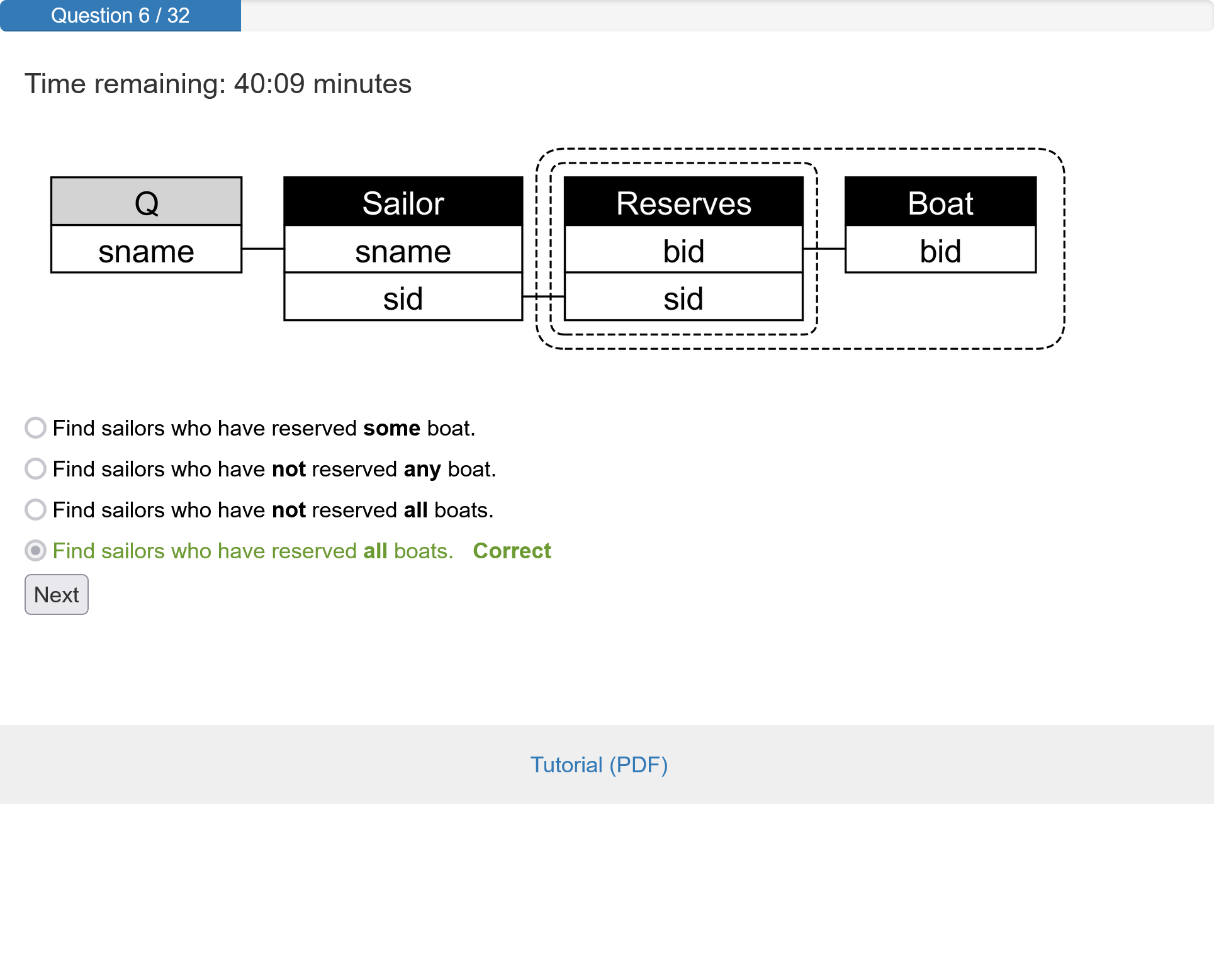}}
        \caption{\diagrams}
        \label{fig:sailors_p4_rd_ui}
    \end{subfigure}
    \caption{The user interface a participant would see for answering questions.
    The shown stimuli are pattern (4) on the sailors-reserve-boats schema for both conditions.
    This particular schema was used for the tutorial, so we excluded it from the 32 questions on the actual study.}
    \label{fig:sailors_p4_sql_and_rd_ui}
\end{figure}

\introparagraph{MTurk inclusion/rejection criteria}
We recruited participants through Amazon Mechanical Turk (AMT), 
limiting participation to adult workers who live in the United States, 
have at least 500 completed tasks approved by requesters,
have  at least 97\% of their completed tasks approved overall,
and who self-determine themselves as experienced SQL users based on the following prompt in the HIT description: ``Workers should be familiar with SQL at the level of an advanced undergraduate database class, in particular with nested SQL queries.''
The HIT is accepted if participants correctly answer at least 16 of the questions and submit within 50 minutes of starting.
We purposefully do not use a qualification exam on SQL as time-consuming qualification exams discourage many workers from participating and are prone to cheating 
(standardized questions and answers tend to be shared on AMT-specific forums).
Instead, we rely on self-determination of SQL experience.
Workers know that failing to submit 16/32 correct questions will get their HIT rejected, meaning they do not get paid, and their overall approval rate goes down---possibly precluding them from doing other HITs.
However, this approach can lead to many workers who skim or skip directions being rejected, lowering the \textit{requester's} HIT approval rate.
The requester approval rate is visible in the MTurk user interface, and workers may be wary of accepting HITs from a requester with a low rate of approving submissions.

\introparagraph{Payment for time and accuracy}
To make sure participants have a stake in the variables we measure (accuracy and speed),
we motivate correct and fast answers by paying participants according to their answering correctly and fast
(largely inspired by our earlier user study on \queryviz~\cite{Leventidis2020QueryVis}).
The base pay for an accepted submission is \$6.00 USD.
(1)~\emph{Accuracy bonus}: For every correctly answered question after the 16$^\textrm{th}$, the participant receives a bonus payment of \$0.20 USD for a maximum pay of \$9.20 USD.
(2)~\emph{Speed bonus}: Based on total test completion time, the participant will receive a percentage bonus on total pay (including the accuracy bonus).
Completion within 11 minutes awards a 5\% bonus for a total maximum pay of \$9.66.
Each minute faster gets the participant an additional 5\% bonus up to 40\% for completing within 4 minutes, with a maximum pay of \$12.88.
These values were determined using pilot data to target workers receiving a U.S.\ average living wage of \$25.02 USD/hr\footnote{\url{https://livingwage.mit.edu/articles/103-new-data-posted-2023-living-wage-calculator}} and average \$8.76 in total for our study.

\introparagraph{Choice of 50 participants}
Our goal was to recruit 50 successful participants.
We made the task (HIT) available in two batches of 60.
171 MTurk workers viewed the consent form, 162 accepted the task and began the study, 146 answered at least one question, 133 completed all questions, and 120 submitted the task.
42 workers accepted the HIT and began the study but returned it without completing it ($162-120$).
Only 58 of the 120 submissions had above the 50\% accuracy threshold.
Among those 58 participants, we used the data of the first 25 who started the task for each condition to get our planned 50 participants, balanced across two groups (\cref{Fig_study_design}).

\introparagraph{Stimuli and user interface}
\Cref{fig:sailors_p4_sql_and_rd_ui} demonstrates the interface a participant would see for answering questions during the study.
(\subref{fig:sailors_p4_sql_ui}) shows an example of a question with pattern (4) and the $\SQL$ condition, (\subref{fig:sailors_p4_rd_ui}) shows the same question posed with $\diagrams$.
Note that we showed participants each schema only once---this figure is simply for illustrative purposes.
A progress bar and countdown timer at the top inform participants of their progress.
Upon selecting an answer and pressing `Next', the chosen and correct answers are highlighted.
A link at the bottom provides quick access to a PDF of the tutorial.

\introparagraph{Median vs.\ mean}
We carefully considered using the mean or the median as summary statistics. 
Each has pros and cons, and decided to use the median for time and the mean for accuracy.
We considered 4 properties:
(1)~The median is \emph{more robust against outliers}. 
A good example of this is how median income is more representative than mean income when the focus is on the experience of individual citizens instead of the overall performance of a country.
The completion times for each question in our study are very imbalanced and have heavier tails than a Gaussian distribution.
(2)~On the other hand, the median \emph{has a slower convergence behavior}.
For example, when sampling from a Gaussian distribution, the 95\% confidence interval (CI) for the median will be larger than the one for the mean.
Thus to get the same statistical power, more data points are needed.
(3)~The mean is more appropriate for summarizing \emph{discrete choices}.
As an extreme example, consider a repeated Boolean event (how often does a coin come up head first).
The median only summarizes which of the two events occurred more times, 
whereas the mean gives a better estimate of the probability
(and even a lossless one together with the number of tosses).
(4)~When summarizing ratios (e.g., the time improvements in our case), the \emph{mean gives a biased estimate}, 
even in the limit of infinitely many data points. 
For example, consider two conditions A and B, and three users, U1, U2, and U3, with recorded times for both conditions being:

\begin{center}
    \begin{tabular}{cccc}
         User & A & B & Ratio B/A\\
         \hline
         U1 & 200 & 100 & 0.5\\
         U2 & 150 & 150 & 1.0\\
         U3 & 100 & 200 & 2.0
    \end{tabular}
\end{center}

\noindent    
Thus one user is faster using B, one user is slower, and one is the same in both conditions.
However, the mean ratio is 1.17, whereas the median ratio is 1.0.

For these reasons, we are using the median for summarizing individual per-participant times
and ratios of time improvement.
However, we use the mean to summarize relative error rates and their differences.

\subsection{Supplemental Material}
\label{appendix:controlled-experiment:links}

All supplemental material required to reproduce our results from the data we collected or replicate our study with additional participants is available on the Open Science Framework (OSF).
Here we provide links to key materials:
\begin{itemize}[leftmargin=*]
    \item The main OSF folder: \url{https://osf.io/q9g6u/}
    \item Study tutorial: \url{https://osf.io/mruzw}
    \item Stimuli-generating code: \url{https://osf.io/kgx4y}
    \item The stimuli: \url{https://osf.io/d5qaj}
    \item Stimuli/schema index CSV: \url{https://osf.io/u8bf9}
    \item Stimuli/schema index JSON: \url{https://osf.io/sn83j}
    \item Server code for hosting the study: \url{https://osf.io/suj4a}.
    \item Collected data: \url{https://osf.io/8vm42}
    \item The final analysis code: \url{https://osf.io/f2xe3}
    \item The analysis can be compared with our planned and preregistered analysis code: \url{https://osf.io/4zpsk/}
\end{itemize}

\subsection{Additional results (preregistered)}
Here we describe an additional result included as part of the preregistration but not listed under its hypotheses.
This result was demoted to the appendix for lack of space.

\resultbox{\begin{resultW}(\textbf{Patterns and Speed}) 
\emph{For all 4 different patterns},
we have strong evidence that
participants were meaningfully faster at identifying each pattern using $\diagrams$ than $\SQL$.
\end{resultW}}

\noindent
We also calculated the median per-participant time for identifying each of the 4 patterns in each of the two conditions.
\Cref{q3_figure} shows the interesting observation that pattern 4 with double negation 
(``Find \textit{sailors} who have \textbf{not} \textit{reserved} \textbf{all} \textit{boats}'')
was not the one that took participants the longest to recognize.
The median improvement of $\diagrams$ over $\SQL$ per pattern and participant is still statistically significant
\emph{for each pattern individually}
(i.e.\ the 95\% CIs were fully below the ratio 1.00).
Please see \cref{tab:q3_table} for details.

\begin{table}
    \centering
    \begin{tabular}{cccc}
 & \textbf{Pattern}&\textbf{ Median time or ratio}& 95\% CI\\
         RD&  P1&  9.25&  [6.67, 10.47]\\
         SQL&  P1&  11.62&  [10.26, 13.35]\\
         ratio RD/SQL&  P1&  .64&  [.49, .78]\\
         RD&  P2&  11.62&  [10.26, 13.35]\\
         SQL&  P2&  13.32&  [11.95, 17.53]\\
         ratio RD/SQL&  P2&  .83&  [.70, .97]\\
         RD&  P3&  10.89&  [9.51, 13.22]\\
         SQL&  P3&  14.13&  [11.19, 21.01]\\
         ratio RD/SQL&  P3&  .66&  [.53, .77]\\
         RD&  P4&  8.14& [7.02, 11.58] \\
         SQL& P4& 11.95&[10.60, 14.92]\\
         ratio RD/SQL& P4& .71&[.60, .86]\\
    \end{tabular}
    \caption{Median times and ratios along with 95\% BCa confidence intervals for \cref{q3_figure}.}
    \label{tab:q3_table}
\end{table}

\begin{figure}[t]
\centering	
\includegraphics[scale=0.39]{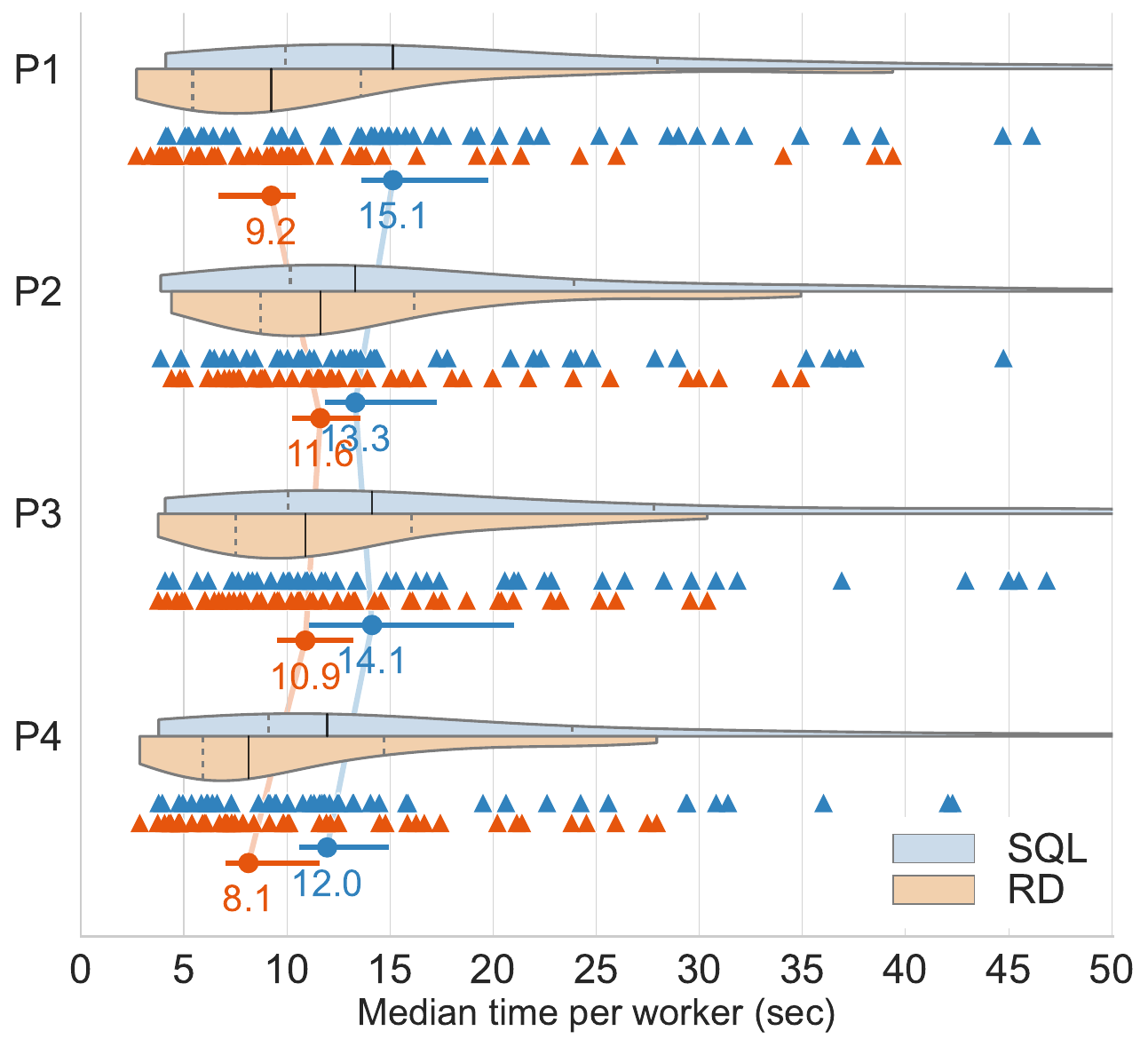}
\caption{Result (4):
    P1--P4 stand for patterns 1--4, respectively.
    The median time per worker for each pattern and condition is shown with bootstrapped 95\% BCa confidence intervals and two views of the data distribution.
	It becomes clear that different patterns take different times to recognize,
	and $\diagrams$ allow participants to do this faster than $\SQL$.
	Although confidence intervals per pattern overlap here, the analysis of the 95\% confidence intervals of the ratios per participant and pattern show statistical significance.
	For a discussion of why these CIs shouldn't be directly compared, see \cite{knol2011MisuseCI,Payton2003OverlappingCIs,Austin2002NoteOverlapCIs}.
    }
\label{q3_figure}
\end{figure}

\subsection{Exploratory analyses (not preregistered)}

This section details an exploratory analysis of a subset of the accepted participants,
which we conducted after collecting the data.
We did not preregister these analyses.
The results should be taken as preliminary findings that need to be backed up by summative studies.

\begin{figure}[t]
    \centering
    \includegraphics[scale=0.39]{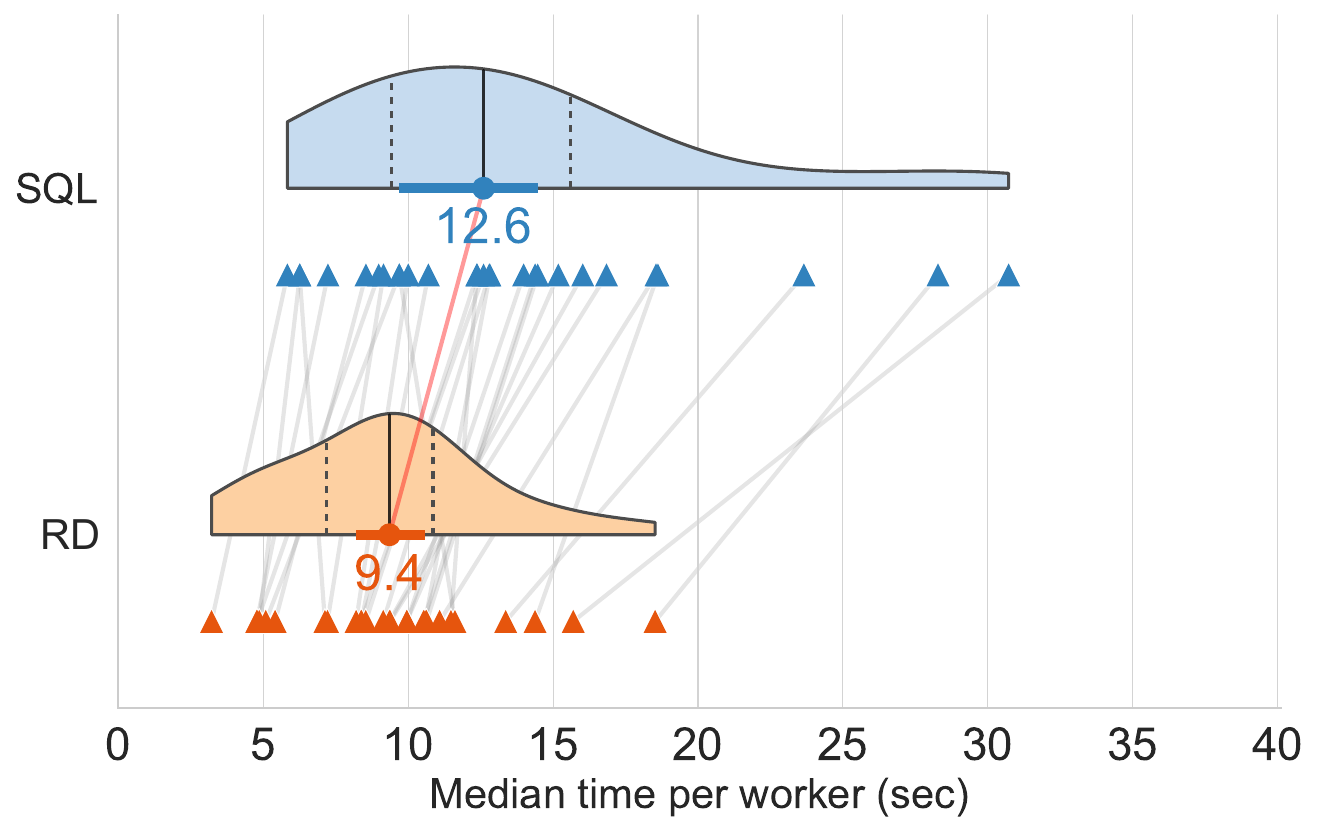}
    \caption{%
	\emph{Exploratory analysis:} Question 1 for $n=27$ with accuracy $>90\%$.
    Compare with
	\hyperref[{q1_figure1_variant1}]{Figure~\ref*{q1_figure1_variant1} (top)}.
    }
    \label{q1_figure1_variant3}
\end{figure}

\begin{figure}[t]
    \centering	
    \includegraphics[scale=0.39]{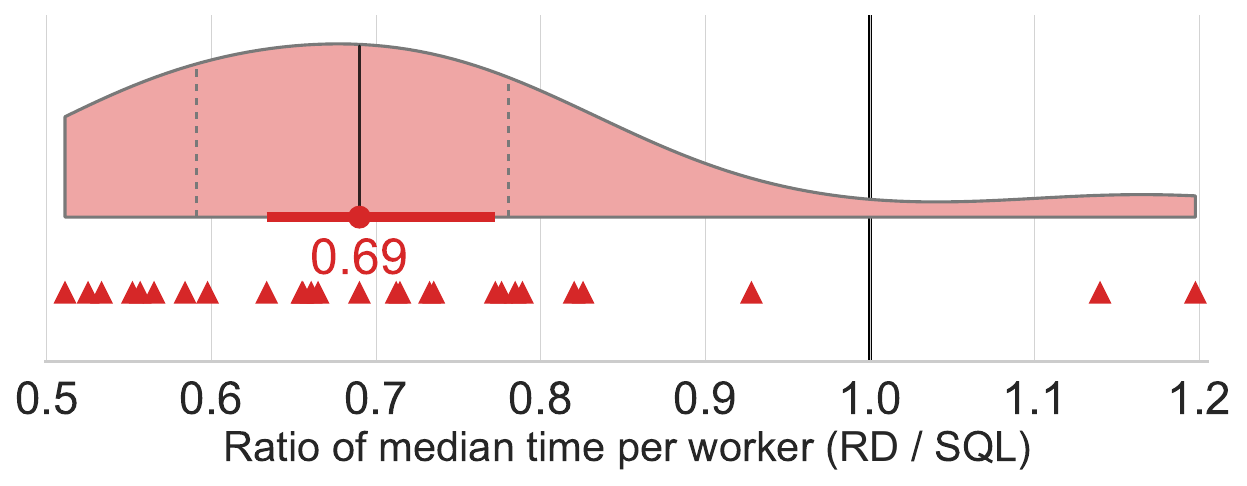}
    \caption{%
	\emph{Exploratory analysis:} Question 1 for $n=27$ with accuracy $>90\%$.
    Compare with 
	\hyperref[{q1_figure1_variant1}]{Figure~\ref*{q1_figure1_variant1} (bottom)}.
	}
    \label{q1_figure2_variant3}
\end{figure}

\begin{figure}[t]
    \centering
    \includegraphics[scale=0.39]{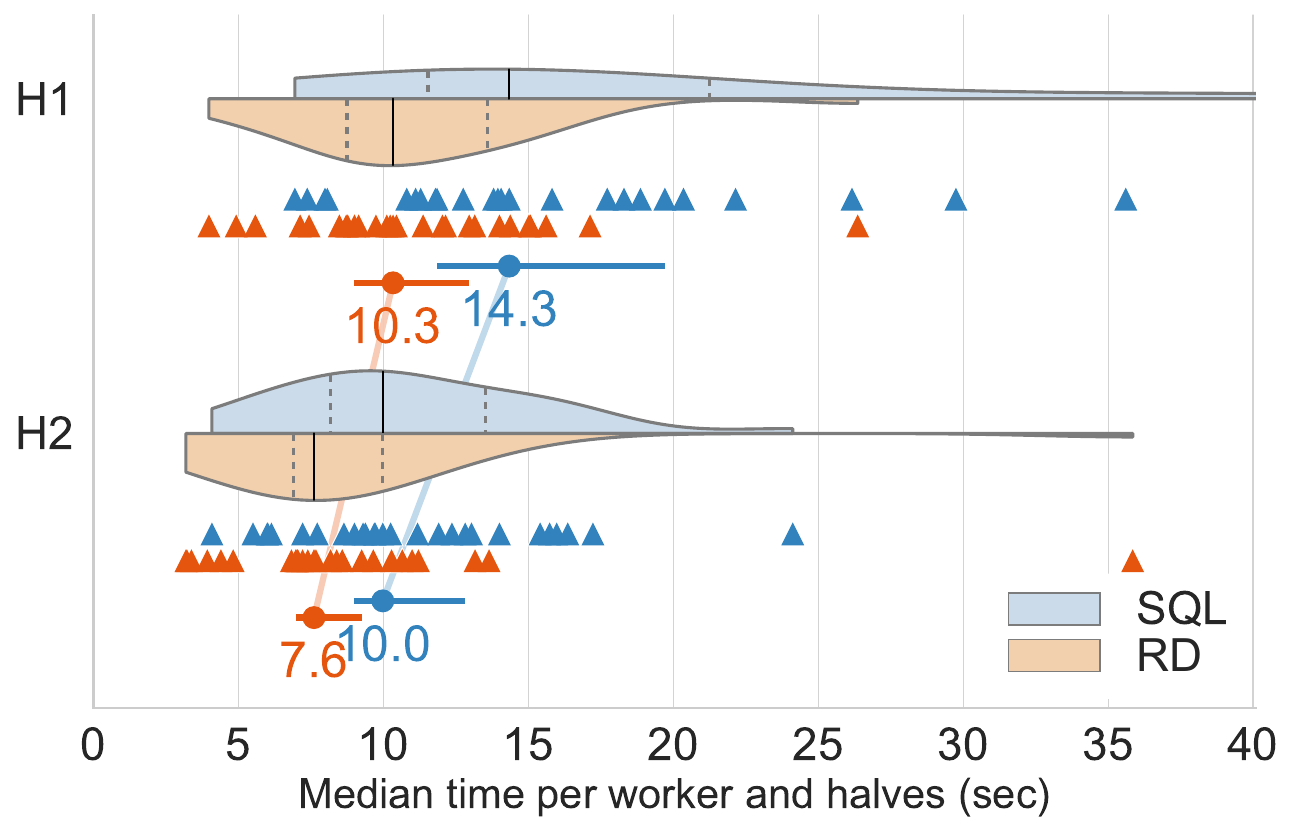}
    \caption{%
	\emph{Exploratory analysis:} Question 2 for $n=27$ with accuracy $>90\%$.
        Compare with \Cref{q2_figure_variant1}.    
	}
    \label{q2_figure_variant3}
\end{figure}

\begin{figure}[t]
    \centering	
    \includegraphics[scale=0.39]{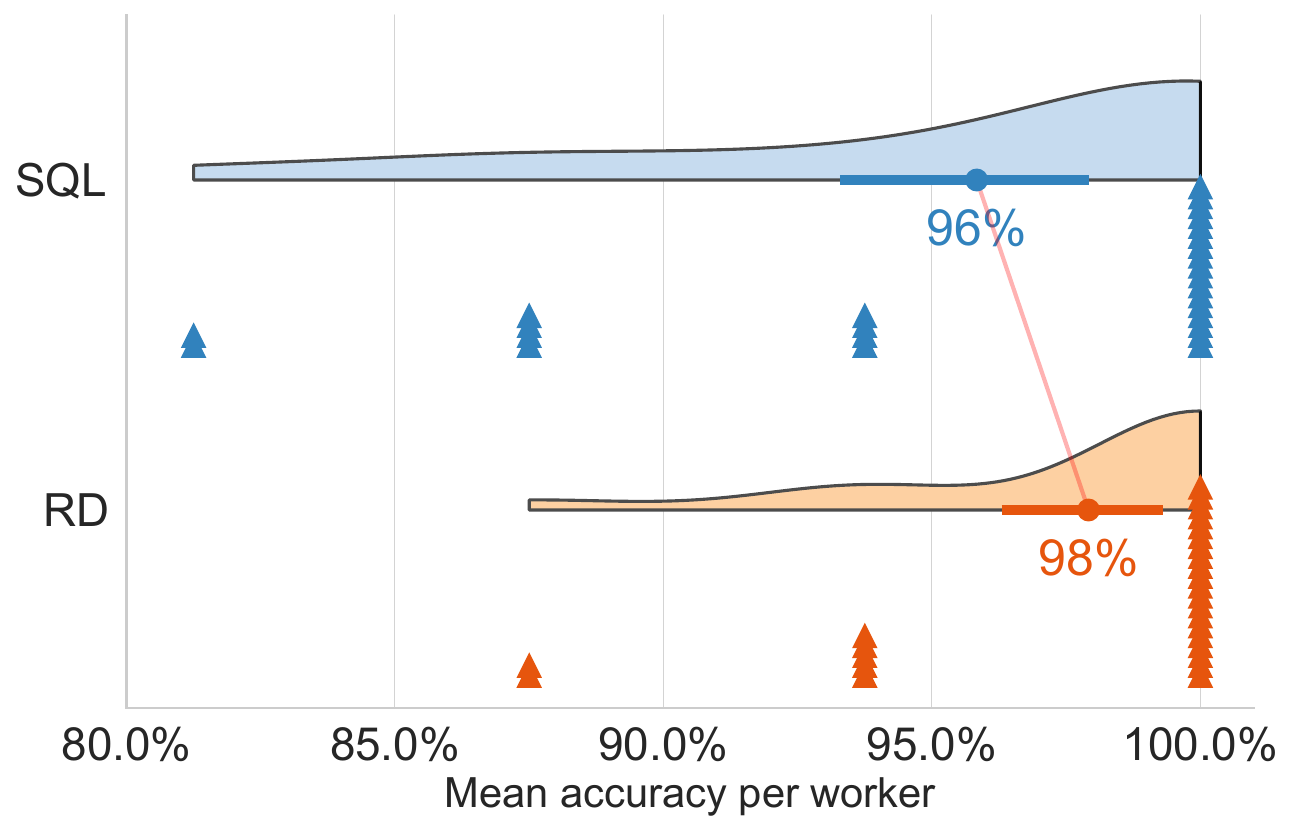}
    \caption{%
    \emph{Exploratory analysis:} Question 3 for $n=27$ with accuracy $>90\%$.
    Compare with 
	\hyperref[{q4_figure_variant1}]{Figure~\ref*{q4_figure_variant1} (top).}
    }
    \label{q4_figure_variant3}
\end{figure}

\begin{figure}[t]
    \centering	
    \includegraphics[scale=0.39]{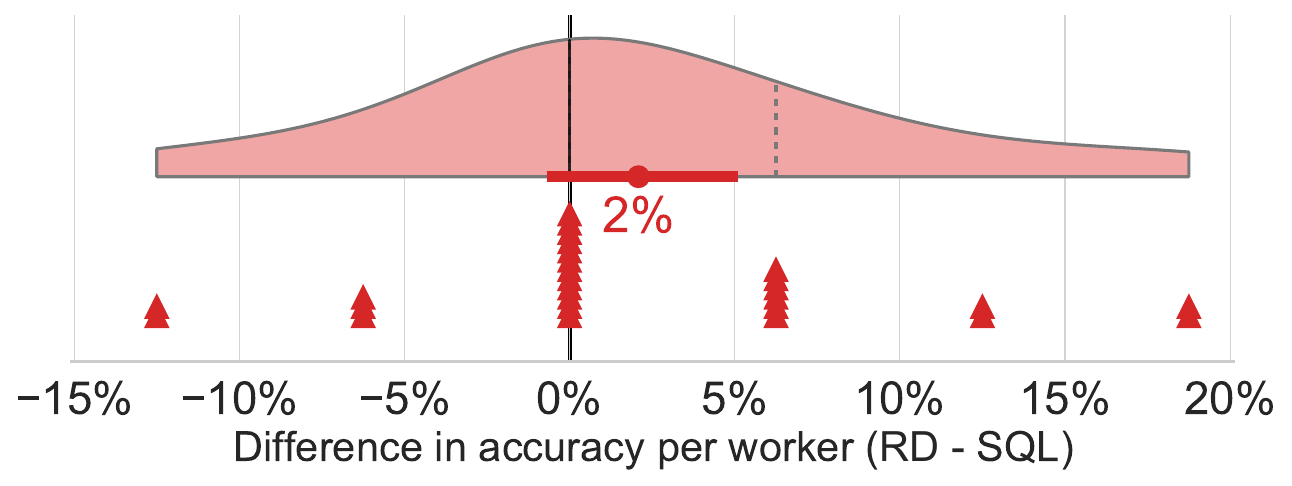}
    \caption{%
    \emph{Exploratory analysis:} Question 3 for $n=27$ with accuracy $>90\%$.
    Compare with 
	\hyperref[{q4_figure_variant1}]{Figure~\ref*{q4_figure_variant1} (bottom)}.
    }
    \label{q4_figure2_variant3}
\end{figure}

\begin{figure}[t]
    \centering	
    \includegraphics[scale=0.39]{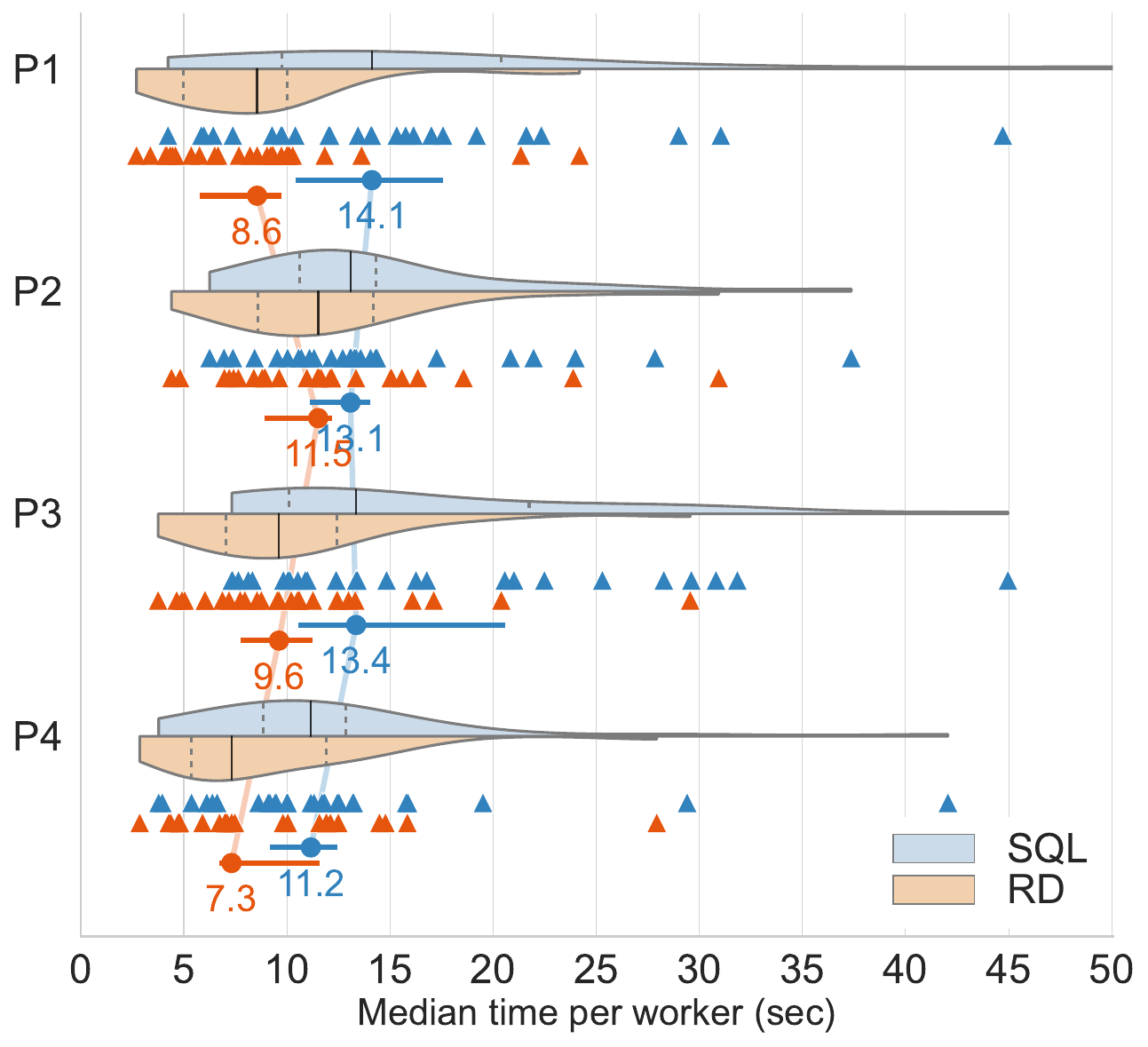}
    \caption{
	Question 4 for $n=27$ with accuracy $>90\%$.
    Compare with \Cref{q3_figure}.
    }
    \label{q3_figure_variant3}
\end{figure}

\introparagraph{(5) Results with stricter accuracy requirement}
Our threshold for acceptance was 50\% accuracy, thus a participant needs to correctly answer at least $k=16$ among the $n=32$ questions. 
\emph{What is the expected fraction of users who finish the HIT by completely randomly choosing from the 4 available choices?}
Modeling the number of successful answers $X$ as a random variable following a binomial distribution $B(n,p)$
with $p=1/4$, 
the probability of getting exactly $k$ right is $\P[X=k] = \binom{n}{k} p^k (1-p)^{n-k}$.
We are interested in 
$\P[X \geq k] =  
1-\P[X \leq k-1]$ 
which is 0.2\%
for $n=32$, $k=16$, $p=1/4$.\footnote{This can be calculated explicitly e.g.\ with 
\textsf{scipy.stats.binom} (\url{https://docs.scipy.org/doc/scipy/reference/generated/scipy.stats.binom.html}).}
Given that 162 MTurk users attempted the HIT, only 133 answered all the questions, of those only 58 got at least half correct, and some successful participants had low accuracy on the 16 questions shown in $\SQL$ (\cref{q4_figure_variant1}),
we asked how the results would change by only considering those among the first 50 users with an accuracy of 90\% (thus maximally 3 incorrect answers).

\Cref{q1_figure1_variant3,q1_figure2_variant3,q2_figure_variant3,q3_figure_variant3}
show our preliminary finding that filtering to users with at least 90\% correct answers has results similar to our planned analysis.
The median time, confidence intervals, and data distribution are surprisingly similar despite the smaller sample size.
We believe this is because the higher accuracy threshold reduces a lot of the randomness 
from users who barely passed (see a few participants with very low accuracy for $\SQL$ in \hyperref[{q4_figure_variant1}]{Figure~\ref*{q4_figure_variant1} (top)}).

Regarding accuracy,  \Cref{q4_figure_variant3,q4_figure2_variant3} show the expected difference caused by the increased threshold: there is now much less evidence for a difference in accuracy between conditions
(which can partially be explained by the fact that most participants in that group actually perfect accuracies).

Notice that filtering the 50 users down to those with an accuracy of 90\% or more led to a slight imbalance 
with 13 users in group 1 and 14 users in group 2.

\subsection{Results from pilot study with PhD students}

Before preregistering our study design, we conducted a pilot study with $n=13$ PhD students studying databases at a U.S.\ University.
This helped us refine our hypotheses; test our website, database, and MTurk implementation code; and test our analysis code.
Pilot testing also helped us to identify any confusing or incorrect stimuli.

\begin{figure*}[t]
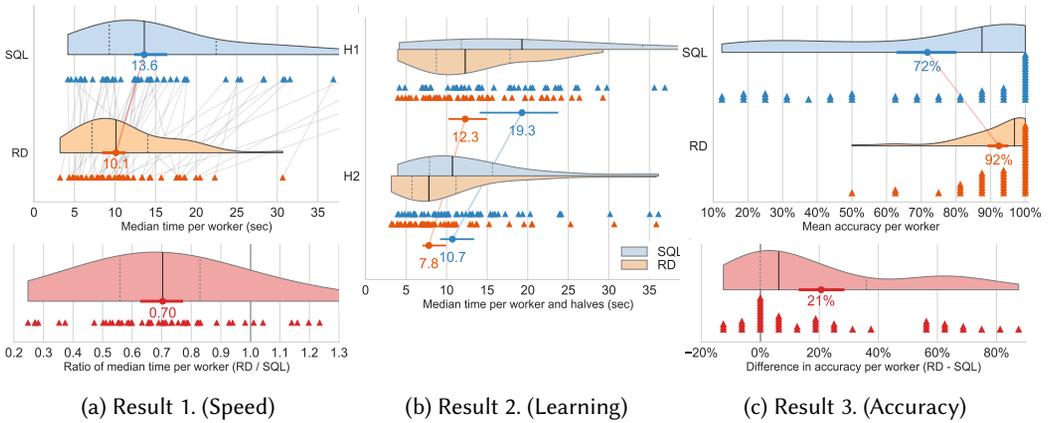

\centering	
\begin{subfigure}[b]{0.34\linewidth}	
    \includegraphics[scale=0.22]{figs/q1_figure1_variant1}
    \includegraphics[scale=0.22]{figs/q1_figure2_variant1}	
    \caption{Result 1. (Speed)}
    \label{q1_figure1_variant1-dupe}
\end{subfigure}
\hspace{-4mm}
\begin{subfigure}[b]{0.34\linewidth}	
    \centering	
    \includegraphics[scale=0.22]{figs/q2_figure_variant1}
	\vspace{4mm}
    \caption{Result 2. (Learning)}
    \label{q2_figure_variant1-dupe}
\end{subfigure}
\hspace{-4mm}
\begin{subfigure}[b]{0.34\linewidth}	
    \centering	
    \includegraphics[scale=0.22]{figs/q4_figure_variant1}
    \includegraphics[scale=0.22]{figs/q4_figure2_variant1}	
    \caption{Result 3. (Accuracy)}
    \label{q4_figure_variant1-dupe}
\end{subfigure}
\caption{%
This figure duplicates \cref{fig:studyResults} to support direct comparison with data from the pilot testing in \cref{fig:studyresults-pilot}.
See the caption of \cref{fig:studyResults} and surrounding text for a description of the encodings used.
}
\label{fig:studyresults-dupe}
\end{figure*}

\begin{figure*}[t]
\centering	
\begin{subfigure}[b]{0.34\linewidth}	
    \includegraphics[scale=0.22]{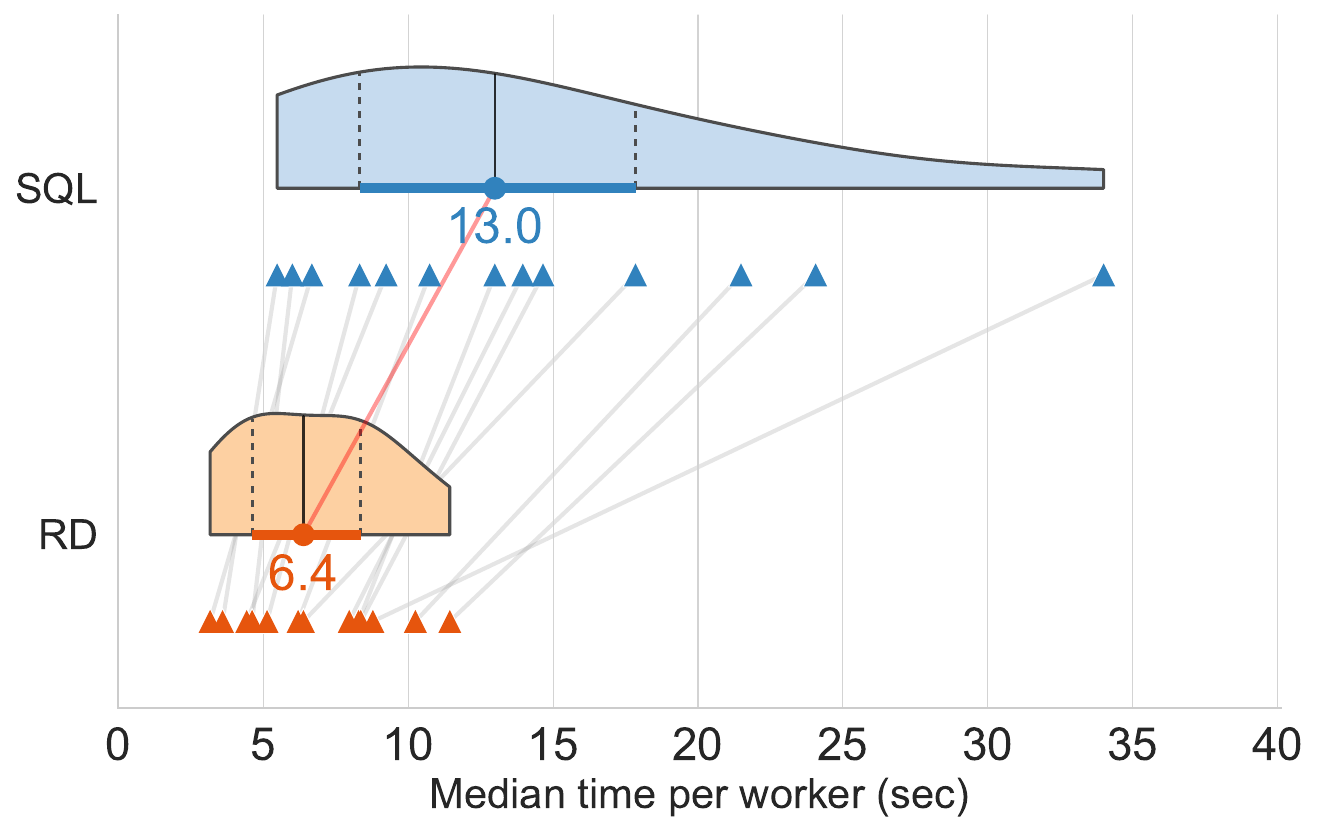}
    \includegraphics[scale=0.22]{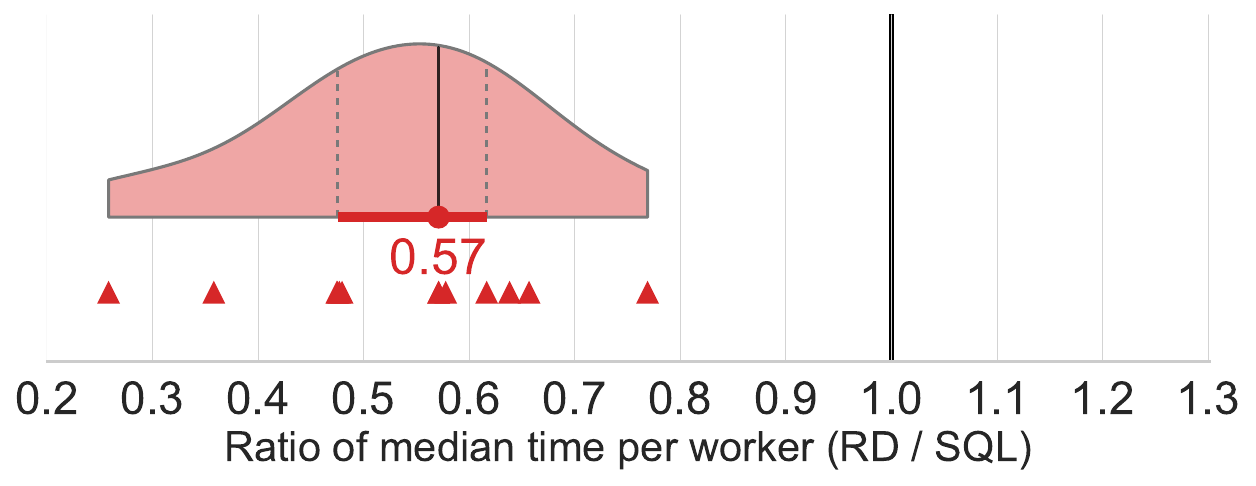}	
    \caption{\textit{Pilot} Result 1. (Speed)}
    \label{q1_figure1_variant1-pilot}
\end{subfigure}
\hspace{-4mm}
\begin{subfigure}[b]{0.34\linewidth}	
    \centering	
    \includegraphics[scale=0.22]{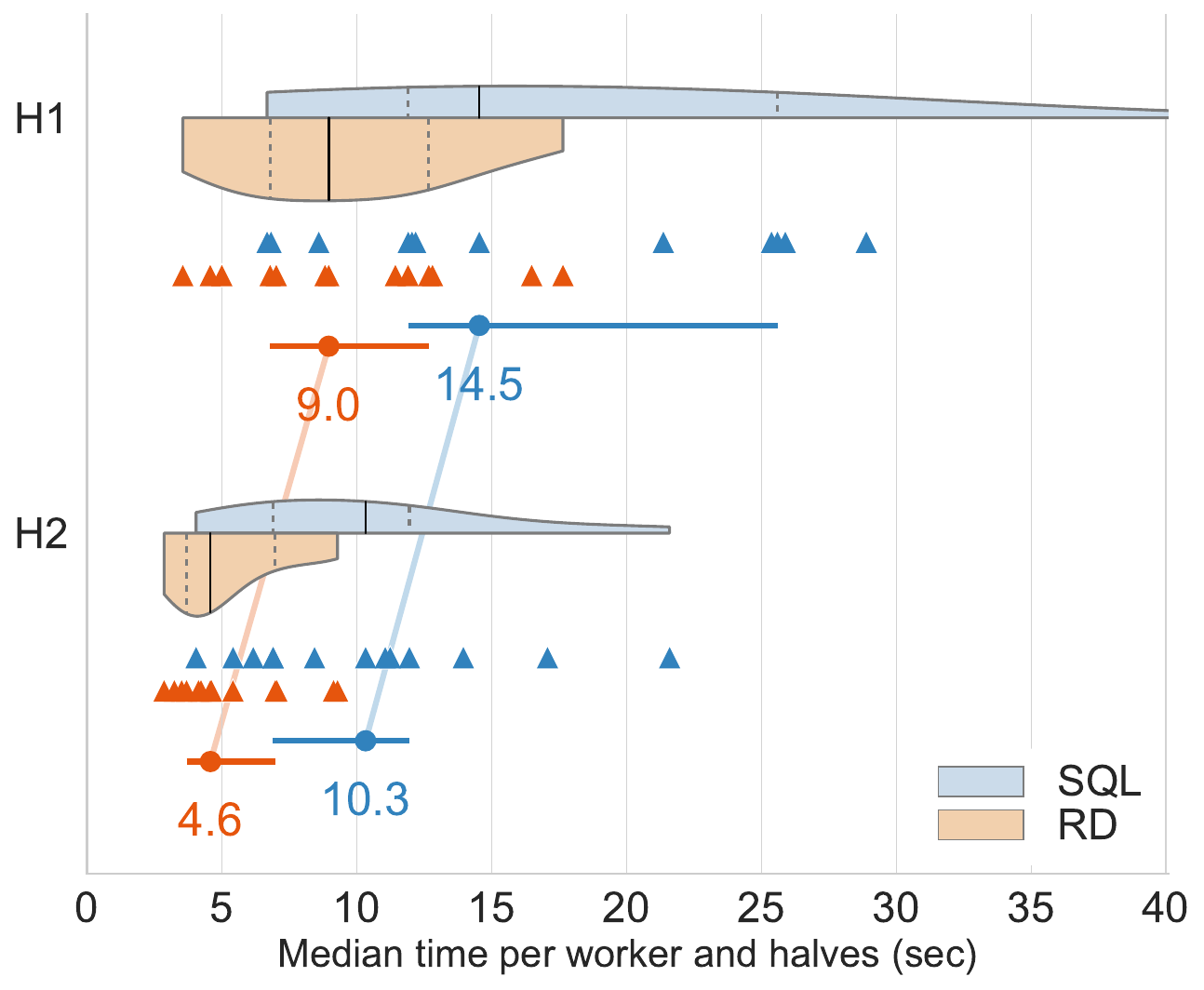}
	\vspace{4mm}
    \caption{\textit{Pilot} Result 2. (Learning)}
    \label{q2_figure_variant1-pilot}
\end{subfigure}
\hspace{-4mm}
\begin{subfigure}[b]{0.34\linewidth}	
    \centering	
    \includegraphics[scale=0.22]{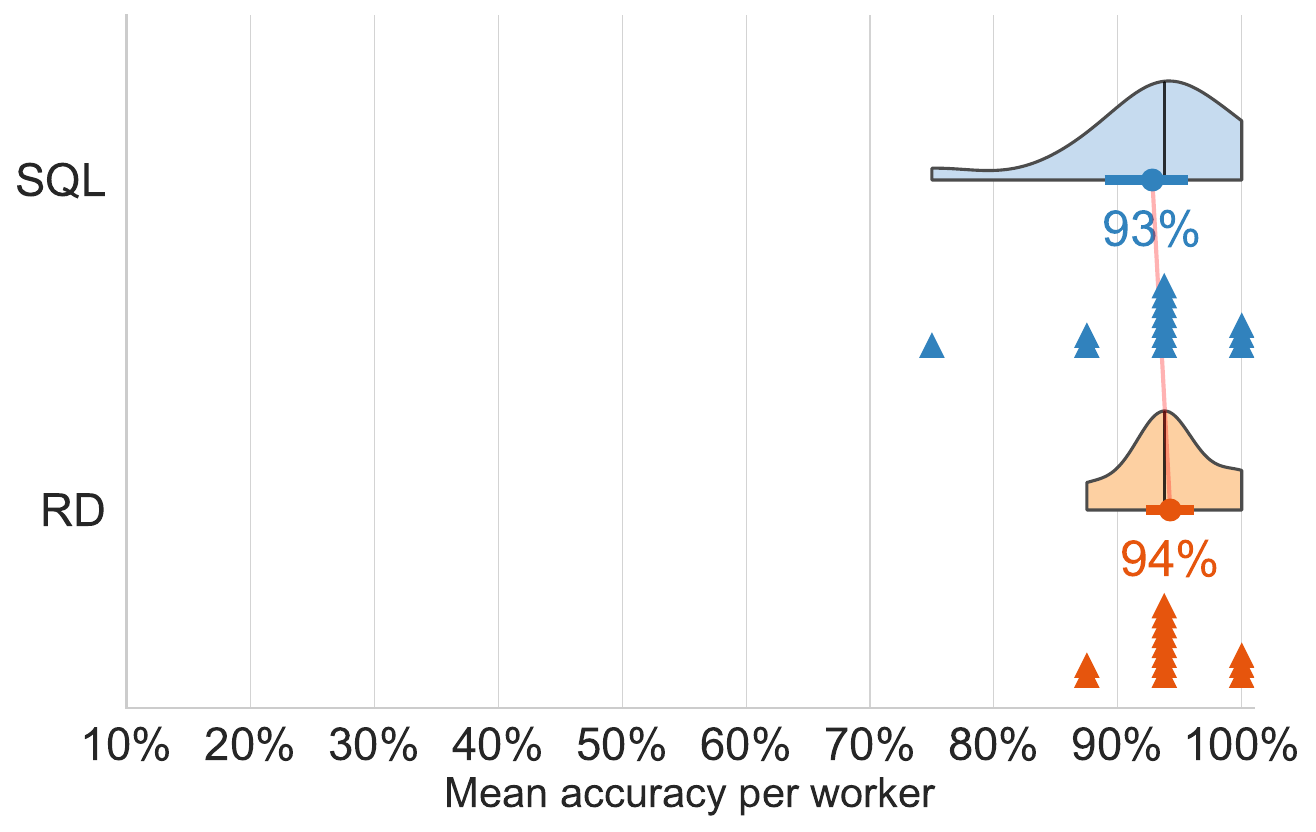}
    \includegraphics[scale=0.22]{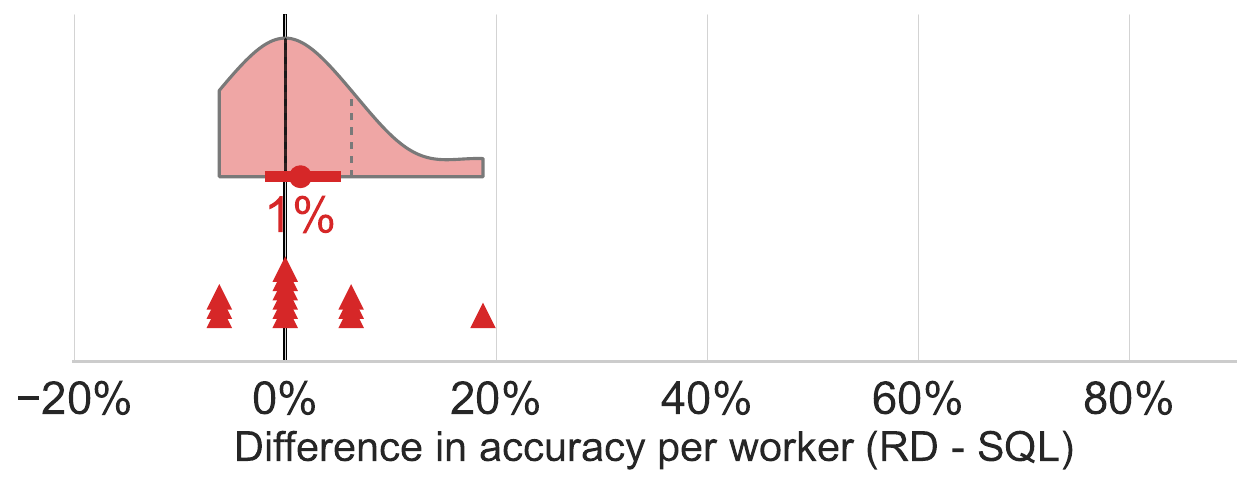}	
    \caption{\textit{Pilot} Result 3. (Accuracy)}
    \label{q4_figure_variant1-pilot}
\end{subfigure}
\caption{%
Results of analyzing data from the $n=13$ \emph{pilot} participants using the approach we preregistered for the full study.
These pilot participants were PhD students studying databases.
Compare with \cref{fig:studyresults-dupe} (\cref{fig:studyResults}) where we present the main study results for the Amazon Mechanical Turk participants.
Results 1 (Speed) and 2 (Learning) were consistent between the PhD students and MTurk workers.
For Result 3 (Accuracy), the PhD students had much higher accuracy using $\SQL$ than the MTurk workers, and their accuracy was comparable when using $\diagrams$.
}
\label{fig:studyresults-pilot}
\end{figure*}

To provide an exploratory comparison of our main results versus using student participants, we analyzed the pilot data using our preregistered analysis code and present the results in \cref{fig:studyresults-pilot}.
Compare these with the main results in \cref{fig:studyResults}, duplicated here in \cref{fig:studyresults-dupe} to support direct comparison.
Note that these results should be considered \textit{exploratory} rather than \textit{summative} as the study design was not preregistered at that point.
Moreover, the results are limited by our using an incorrect counterbalancing approach that we have since rectified.
We do not believe the introduced ordering effects would dramatically affect the results, but we cannot rule out a possible effect.

Our Pilot Results 1 (Speed) and 2 (Learning) were consistent with what we found in the eventual study.
However, Pilot Result 3 (Accuracy) stands in contrast.

\resultbox{\begin{resultE}(\textbf{Speed}) 
We have some exploratory evidence that PhD database students were meaningfully faster at identifying patterns using $\diagrams$ than $\SQL$:
median ratio $\diagrams/\SQL = $ 0.57, 95\% CI [0.48, 0.62].
\end{resultE}}

\noindent
\hyperref[{q1_figure1_variant1}]{Figure~\ref*{q1_figure1_variant1-pilot} (top)}
shows the median per-participant times per condition (and overall median across participants),
and 
\hyperref[{q1_figure1_variant1}]{Fig.~\ref*{q1_figure1_variant1-pilot} (bottom)}
the per-participant ratios between median times.

\resultbox{\begin{resultE}(\textbf{Learning}) 
We have some exploratory evidence that PhD students got meaningfully faster during the study in both conditions.
\end{resultE}}

\noindent
\Cref{q2_figure_variant1-pilot} shows the individual times for H1 ($1^\textrm{st}$ half) and H2 ($2^\textrm{nd}$ half),
together with medians and CIs.
Not shown are the median ratios H1$/$H2 we used for inference, which were
0.54, 95\% CI [0.72, 0.80]
for $\diagrams = $ and
0.59, 95\% CI [0.43, 0.70]
for $\SQL$.

\resultbox{\begin{resultE}(\textbf{Accuracy}) 
We have some exploratory evidence that PhD database students have comparable correctness with \diagrams\ and $\SQL$: 
mean difference in accuracy $\diagrams-\SQL = $ 1\%, 95\% CI [-2\%, 5\%].
\end{resultE}}

\hyperref[{q4_figure_variant1}]{Figure~\ref*{q4_figure_variant1-pilot} (top)}
shows the per-participant accuracy in each condition and the overall mean accuracies.
\hyperref[{q4_figure_variant1}]{Figure~\ref*{q4_figure_variant1-pilot} (bottom)}
shows the per-participant difference in accuracy and overall mean.

This last pilot finding is what led to our preregistered hypothesis that participants make a comparable number of correct responses using $\SQL$ and $\diagrams$.
However, it turns out that MTurk workers were considerably more often correct using $\diagrams$ than $\SQL$.
An additional summative study is necessary to say whether there is actually a difference between the proportion of correct responses provided by MTurk workers (who are presumably less skilled with SQL) and PhD students studying databases (who are presumably more skilled).
It appears that may be true though, in which case a subsequent study could better answer the question: ``Do $\diagrams$ help inexperienced $\SQL$ users identify common relational query patterns more accurately but only have comparable accuracy for experienced users?''
Future studies could also better answer the question: ``How comparable in $\SQL$ proficiency are self-described $\SQL$ experts on MTurk to students who get good grades in university database courses?''

\section{More Related Work for \texorpdfstring{\cref{SEC:RELATEDWORK}}
{Section \ref{SEC:RELATEDWORK}}}
\label{app:extra_related_work}

\subsection{Peirce's existential graphs (\texorpdfstring{\cref{sec:Peirce}}{Section \ref{sec:Peirce}})}
\label{app:PeirceDetails}
\phantom{x}\vspace{-4mm}

We discussed earlier that Lines of Identity (LIs) in beta graphs have multiple meanings (existential quantification and identity), 
and this ``function overload'' can make the graphs ambiguous.
We now discuss this important point in more detail.

\introparagraph{Problems from abusing lines in beta graphs}
While over 100 years old, Peirce's beta system has led 
over the years to multiple misinterpretations and 
an ongoing discussions about how to interpret a valid beta graph correctly.
The literature contains many attempts to provide formal ``interpretations'' and provide consistent readings of these graphs.
How can it be that something supposed to be formal still allows so much ambiguities?
In our opinion, beta graphs have one important design problem leading to those misunderstandings:
it is the \emph{overloading of the meaning of the Lines of Identity (LI)},
and thus one instances of an abuse of lines as symbols.
As mentioned before, LIs are used to denote two different concepts:
($i$) the \emph{existence} of objects (intuitively an existential quantification of a variable in $\DRC$ such as $\exists x$), and
($ii$) the \emph{identity} between objects (intuitively, R.A=S.A in $\TRC$).
This non-separation of concerns
and function overload
leads to unfixable ambiguities.
We illustrate with a few examples.

\begin{example}[A red boat]
\label{ex:aredboat}%
Consider
the sentence ``There is a red boat''
shown in
\cref{Fig_redboat_EG} as a beta graph.
As beta graphs cannot represent constants, the graph requires a special unary predicate ``red boat.''
The LI represents both ``there exists something'' and ``that something is equal to a red boat.''
Thus a line (which arguably suggests two items being connected or joined) is meant as a quantified variable,
and the beta graph can be interpreted in $\DRC$ as:
\begin{align*}
	\exists x[\sql{RedBoat}(x)]
\end{align*}

\Cref{Fig_redboat} shows the same sentence as a \diagram.
Notice that ``there exists something'' is represented by just placing this something (a predicate) on the canvas.
There is no need for an existential line.
Also notice how the modern UML diagram allows predicates (relational atoms) with several attributes, and one of those attributes 
can be set equal to a constant (here \selPredicateBox{\textup{color = `red'}}). 
It can be read rather naturally like a $\TRC$ statement:
\begin{align*}
	\exists b \in \sql{Boat}[\sql{B.color=`red'}]
\end{align*}
\end{example}

\begin{figure}[t]
\centering	
\hspace{10mm}
\begin{subfigure}[b]{.25\linewidth}
	\centering
    \includegraphics[scale=0.5]{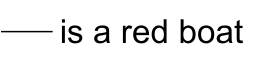}
\caption{Beta graph}
\label{Fig_redboat_EG}
\end{subfigure}	
\hspace{0mm}
\begin{subfigure}[b]{.25\linewidth}
	\centering
    \includegraphics[scale=0.5]{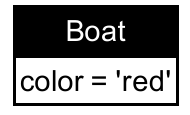}
\caption{\diagram}
\label{Fig_redboat}
\end{subfigure}	
\hspace{10mm}
\caption{\Cref{ex:aredboat} illustrates that lines in beta graphs (called ``Lines of Identities'' or LIs) suffer from ``function overload'': LIs are used as a symbol for both 
equivalence between objects, and
\emph{existential quantification}.
\diagrams\ avoid this function overload by 
using lines only to represent built-in (comparison) predicates,
and assuming proper predicates (relations with attributes) to be existentially quantified.}
\end{figure}

The interpretation of beta graphs where one LI represents one existentially-quantified variable
can at times be intuitive and simply correspond to a modern $\DRC$ interpretation 
(see also \cref{ex:similarities}).
However, such a  simple interpretation is not always possible.

\begin{example}[Exactly one red boat]
	\label{ex:exactlyoneredboat}
Consider the sentence ``There exists exactly one red boat.''
\Cref{Fig_exactly_one_red_boat_EGa,Fig_exactly_one_red_boat_EGb}
show two beta graphs with \emph{different cut nestings} that can both be read as
\begin{align*}
	\exists x[\sql{RedBoat}(x) \wedge \neg(\exists y[\sql{RedBoat}(y) \wedge x \neq y]
\end{align*}
Now a single LI needs to represent \emph{two} existentially-quantified variables,
and two \emph{different nestings} of the cuts can represent the same statement.
Contrast this with \cref{Fig_exactly_one_red_boat}, read in $\TRC$ as:
\begin{align*}
\exists b \in \sql{Boat}
	&[\sql{B.color=`red'} \wedge \neg(\exists b_2 \in \sql{Boat} \\
	&[b_2.\sql{color}=\sql{`red'} \wedge b.\sql{bid} \neq b_2.\sql{bid}])]
\end{align*}
Notice here that the inequality is simply represented by a \emph{label of a join} between two predicates. 
Two tuple variables are represented by two different atoms and the interpretation is unambiguous:
\emph{There exists a boat whose color is red, and there does not exist another boat whose color is red and whose $\sql{bid}$ is different.}
\end{example}

\begin{figure}[t]
\centering
\begin{minipage}{0.45\linewidth}
\begin{subfigure}[b]{\linewidth}
	\centering
    \includegraphics[scale=0.5]{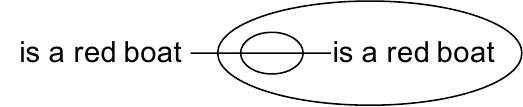}
\caption{}
\label{Fig_exactly_one_red_boat_EGa}
\end{subfigure}
\begin{subfigure}[b]{\linewidth}
	\centering
    \includegraphics[scale=0.5]{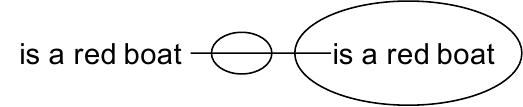}
\caption{}
\label{Fig_exactly_one_red_boat_EGb}
\end{subfigure}
\end{minipage}
\hspace{2mm}
\begin{minipage}{0.38\linewidth}
\vspace{7mm}
\begin{subfigure}[b]{\linewidth}
	\centering
    \includegraphics[scale=0.5]{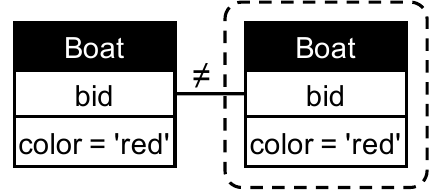}
\caption{}
\label{Fig_exactly_one_red_boat}
\end{subfigure}
\end{minipage}
\caption{\Cref{ex:exactlyoneredboat}: The combination of LI's and nesting symbols (called ``cuts'' in beta graphs) provide ambiguous ways to nest cuts.}
\vspace{-1mm}
\end{figure}

The fact that one LI can branch into multiple endings (also called \emph{ligatures}), may have \emph{loose endings}, and may represent multiple existentially-quantified variables, together with cuts being applied to such LI's can quickly lead to hard-to-interpret diagrams (see e.g., the increasingly-unreadable figures in \cite[pp. 42-49]{Shin:2002}).
This led to several attempts in the literature to provide ``reading algorithms'' of those graphs (e.g., \cite{Zeman:1964,Roberts:1973,Shin:2002}) and rather complicated proofs of the expressiveness of beta graphs~\cite{Zeman:1964}, assuming a correct reading.
For example, the paper by Dau \cite{Dau:2006} points out an error in Shin's reading algorithm \cite{Shin:2002}.
However, Dau's correction to Shin~\cite{Dau:2006} itself also has errors.
For example, the interpretation of the right-most diagram in \cite[Fig 2]{Dau:2006} 
(reproduced as \cref{fig:Fig_Dau}) is wrong and misses one equality.
The given interpretation is
\begin{align*}
&\exists x. \exists y. \exists z [S(x) \wedge P(y) \wedge T(z) \wedge \neg(x = y \wedge y = z)]
\intertext
{
whereas it should be
}
&\exists x. \exists y. \exists z [S(x) \wedge P(y) \wedge T(z) \wedge \neg(x = y \wedge y = z \wedge x = z)]
\end{align*}
This is just an intuitive example how difficult beta graphs are in practice to interpret, even by the experts, 
and even by experts pointing out errors from other experts.

\begin{figure}[t]
\centering
\begin{subfigure}[b]{0.7\linewidth}
	\centering
	\includegraphics[scale=0.35]{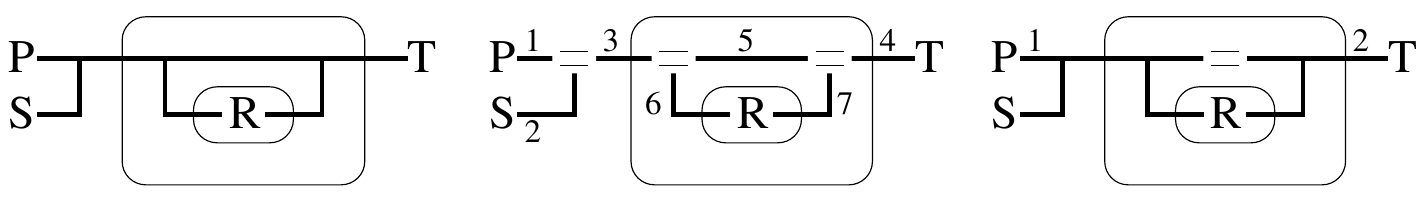}
	\caption{}
\end{subfigure}	
\hspace{5mm}
\begin{subfigure}[b]{0.2\linewidth}
	\centering
	\includegraphics[scale=0.19]{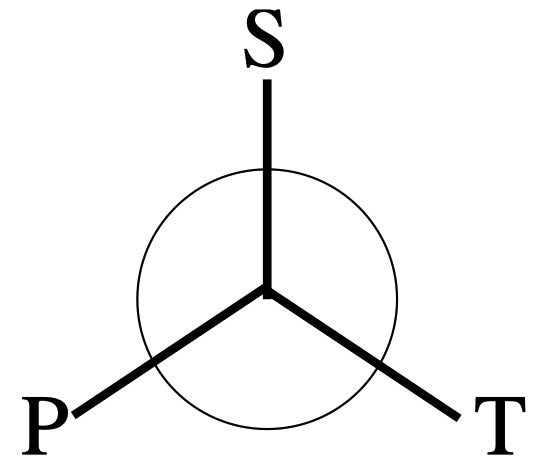}
	\caption{}
	\label{fig:Fig_Dau}%
\end{subfigure}
\vspace{0mm}
\caption{
(a): Figure copied from Dau \cite{Dau:2006} discussing an example beta graph whose interpretation provided by Shin \cite{Shin:2002} is incorrect, 
together with two alternative ways of splitting the LIs in order to interpret the graph correctly.
The details of the arguments are intricate and not important here. 
What matters is that a lot of disagreement exists as to how to interpret
LI's correctly.
\diagrams\ avoid this problem entirely by using lines only as comparison predicates.
(b): Right-most diagram of Figure 2 in Dau \cite{Dau:2006}: a difficult-to-interpret Peirce beta graph.
The reading algorithm presented is wrong and misses one equality.%
}
\end{figure}

\introparagraph{Why \diagrams\ avoid the problem}
\diagrams\ use the line only for connecting two attributes. 
The type of connection is unambiguously represented by a label.
Quantification is represented by predicates themselves.
Thus, on a more philosophical level, 
we think that our visual formalism solves those problems based on a more modern interpretation of first order logic:
$\TRC$ was created by Edgar Codd in the 1960s and 1970s
in order to provide a declarative database-query language for data manipulation in the relational data model~\cite{Codd:1970}.
In contrast, beta graphs were proposed even before first-order logic, which was only clearly articulated some years after Peirce's death in the 1928 first edition of David Hilbert and Wilhelm Ackermann's ``Grundzüge der theoretischen Logik''~\cite{HilbertAckerman:1928}.
Zeman, in his 1964 PhD thesis~\cite{Zeman:1964}, was the first to note that beta graphs are isomorphic to first-order logic with equality. However, the secondary literature, especially Roberts~\cite{Roberts:1973} and Shin~\cite{Shin:2002}, does not agree on just how this is so~\cite{wiki:existentialgraph}. 
We did not start from Peirce's beta graphs and attempt
to fix issues that have been occupying a whole community for years.
Rather, we started from 
the modern UML reading of relational schemas and an understanding of $\TRC$, 
and tried to achieve a minimal visual extension to provide relational completeness and pattern-isomorphism to $\TRC$.
This happens to provide a natural solution for the interpretation problems of beta graphs.
We believe that \diagrams\ provide a clean, unambiguous, and, in hindsight, simple abstraction of query patterns.

\begin{figure}[t]
\centering	
\begin{subfigure}[b]{.4\linewidth}
	\centering
    \includegraphics[scale=0.5]{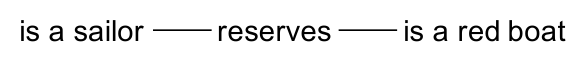}
	\vspace{2mm}
    \caption{}
    \label{Fig_Sailor_some_Sailors_some_red_boat_EG}
\end{subfigure}	
\hspace{3mm}
\begin{subfigure}[b]{.4\linewidth}
	\centering	
    \includegraphics[scale=0.5]{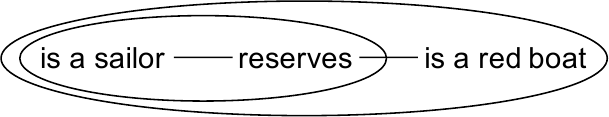}
    \caption{}
    \label{Fig_all_red_boats_some_Sailors_EG}
\end{subfigure}	
\hspace{5mm}
\begin{subfigure}[b]{.4\linewidth}
\vspace{1mm}	
	\centering	
    \includegraphics[scale=0.5]{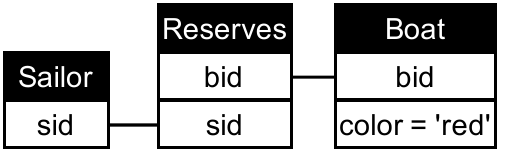}
\caption{}
\label{Fig_Sailor_some_Sailors_some_red_boat}
\end{subfigure}	
\hspace{5mm}
\begin{subfigure}[b]{.4\linewidth}
\vspace{1mm}
	\centering	
    \includegraphics[scale=0.5]{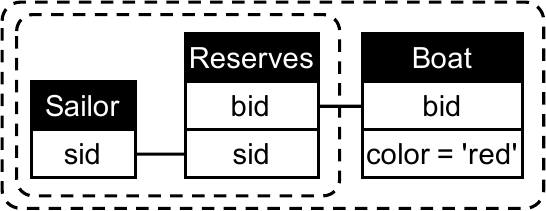}
\vspace{-1mm}
\caption{}
\label{Fig_all_red_boats_some_Sailors}
\end{subfigure}	
\caption{\Cref{ex:similarities}: Diagrams comparing the representations of negation in beta graphs (top) and \diagrams\ (bottom).%
}%
\label{fig:EG_comparison}%
\end{figure}

\begin{example}[Nested negation] 
\label{ex:similarities}%
\Cref{fig:EG_comparison}
shows 2 beta graphs: 
\begin{align*}
&\textrm{\subref{Fig_Sailor_some_Sailors_some_red_boat_EG}:    \emph{There exists a sailor who reserved a red boat.}}			\\
&\textrm{\subref{Fig_all_red_boats_some_Sailors_EG}: 			   \emph{All red boats were reserved by some sailor.}}
\end{align*}
Beta graphs cannot represent constants and thus need to replace a selection of boats that are red with a dedicated new predicate
``is a red boat.''
Their respective translations into $\DRC$ are:
\begin{align*}
	\textrm{\subref{Fig_Sailor_some_Sailors_some_red_boat_EG}: }
	&\exists x, y[\sql{Sailor}(x) \wedge \sql{RedBoat}(y) \wedge \sql{Reserves}(x,y)]	\\
	\textrm{\subref{Fig_all_red_boats_some_Sailors_EG}: }
	&\neg (\exists y[\sql{RedBoat}(y) \wedge \neg(\exists x[\sql{Sailor}(x) \wedge \sql{Reserves}(x,y)])])
\end{align*}
Contrast the beta graphs with their respective \diagrams\
in \cref{fig:EG_comparison} and $\TRC$:
\begin{align}
&\begin{aligned}
\textrm{\subref{Fig_Sailor_some_Sailors_some_red_boat}: }
\exists s \in \sql{Sailor},
	&b \in \sql{Boat}, r \in \sql{Reserves} [r.\sql{bid} = b.\sql{bid} \wedge \\
	&r.\sql{sid} = s.\sql{sid} \wedge b.\sql{color} = \sql{`red'}]
\end{aligned}
\\
&
\textrm{\subref{Fig_all_red_boats_some_Sailors}: }
\neg(\exists b \in \sql{Boat}
	[B.\sql{color}=\sql{`red'} \wedge \neg (\exists s \in \sql{Sailor}, r \in \sql{Reserves}	\notag\\
	&\hspace{24mm}[r.\sql{bid} = b.\sql{bid} \wedge r.\sql{sid} = s.\sql{sid}])])
\end{align}
\end{example}

\begin{figure*}[t]
\centering
\begin{subfigure}[b]{.25\linewidth}
    \includegraphics[scale=0.3]{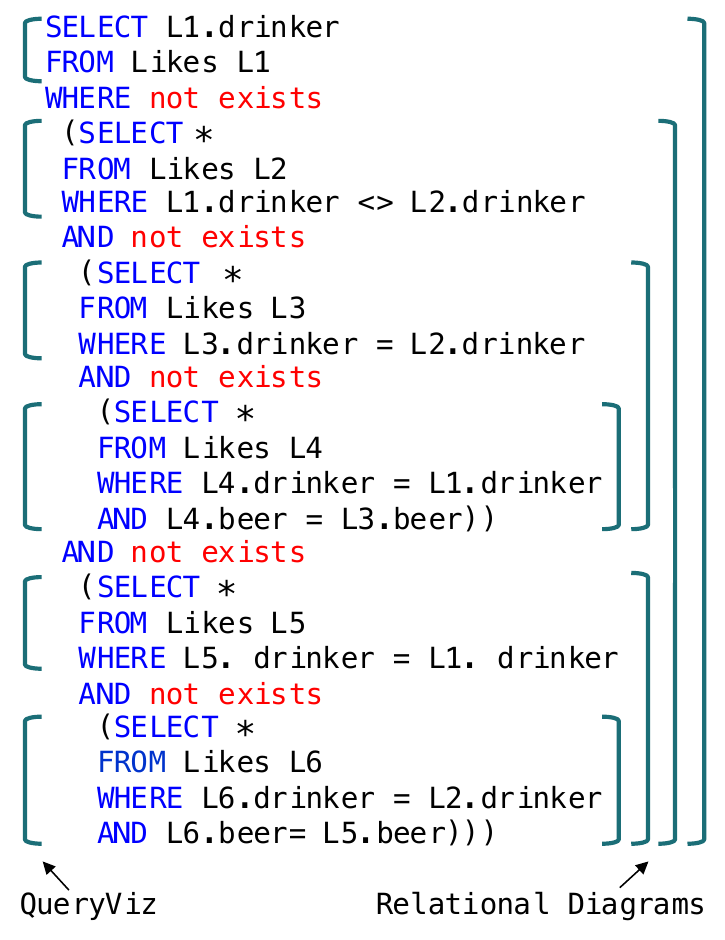}
	\vspace{-5mm}
    \caption{}
    \label{Fig_long_nested_query}
\end{subfigure}	
\hspace{3mm}
\begin{subfigure}[b]{.355\linewidth}
    \includegraphics[scale=0.3]{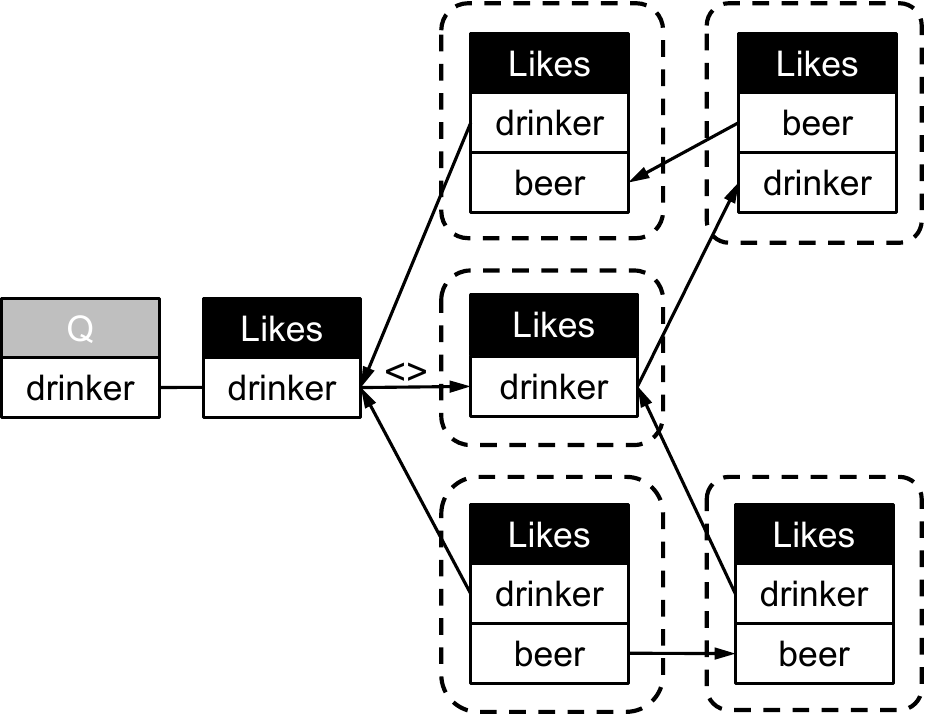}
	\vspace{5mm}
    \caption{}
    \label{Fig_QV_unique_beer_tastes2}
\end{subfigure}	
\begin{subfigure}[b]{.355\linewidth}
    \includegraphics[scale=0.3]{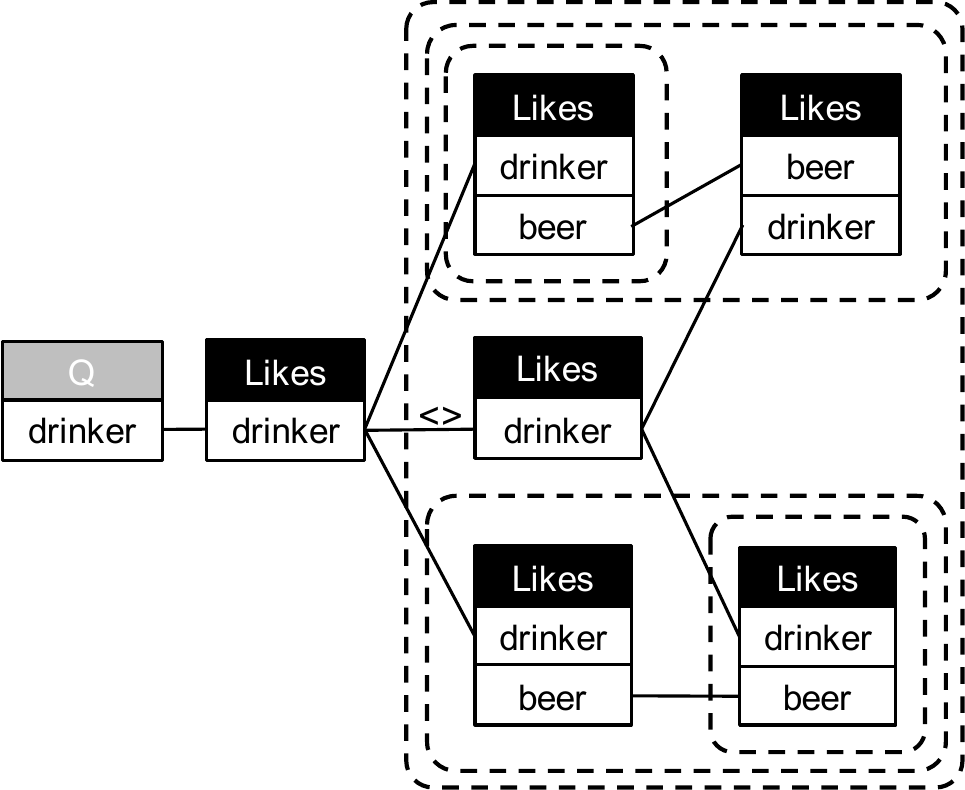}
	\vspace{0mm}	
    \caption{}
    \label{Fig_QV_unique_beer_tastes1}
\end{subfigure}	
\caption{Illustrations for \Cref{ex:uniquebeers}:
(a) Unique-set-query ``\emph{Find drinkers with a unique beer taste}'' used by \cite{Leventidis2020QueryVis}, 
(b) $\queryviz$ diagram with reading order encoded by arrows (redrawn according to \cite{Leventidis2020QueryVis}),
(c) \diagrams\ with a nested scoping and no need for arrows.
}
\label{fig:beerquery}
\end{figure*}

\subsection{\queryviz~(\texorpdfstring{\cref{sec:queryvis1_vs_2}}{Section \ref{sec:queryvis1_vs_2}})}
\label{appendix:queryvis}

We use the ``unique beer taste'' query that was proposed in the $\queryviz$ paper \cite{Leventidis2020QueryVis} to show the difference in design decisions.

\begin{example}[Unique-set-query]
\label{ex:uniquebeers}
Consider the $\SQL$ query from \cref{Fig_long_nested_query}
asking to find ``\emph{drinkers who like a unique set of beers},'' i.e.\ no other drinker
likes the exact same set of beers.
The scoping brackets to the left of the query in  \cref{Fig_long_nested_query} 
show the content of boxes used by $\queryviz$, which include all tables \emph{from each individual query block}.
Without the additional visual symbol of arrows, 
this diagram becomes ambiguous to interpret.
To mitigate this problem,  
the design of \queryviz~\cite{Leventidis2020QueryVis,DanaparamitaG2011:QueryViz}
uses
directed arrows with an implied \emph{reading order}~(\cref{Fig_QV_unique_beer_tastes2}).

The scoping brackets to the right in \cref{Fig_long_nested_query} 
show the nesting of the variables scopes in queries, which are also reflected in
the dashed bounding boxes in \diagrams\ 
(\cref{Fig_QV_unique_beer_tastes1}).
\end{example}

The design decision
in $\queryviz$ \cite{Leventidis2020QueryVis} 
are justified in terms of usability 
(for ``most'' queries the diagrams are not ambiguous and the reduction in nesting simplifies their interpretation), 
yet requires overloading of the meaning of arrows.
Two conceptual problems with these diagrams are:
(1) $\queryviz$ requires each partition of the canvas to contain a relation from the relational schema.
Our earlier examples from \cref{fig:Sailor_disjunction,fig:disjunction_oneex} show examples that can thus not be handled.
(2) $\queryviz$ does not guarantee unambiguous visualizations for nested queries with nesting depth $\geq 4$. 
This was alluded to already in \cite{Leventidis2020QueryVis}, and we next give an example to illustrate:

\begin{example}[Ambiguous $\queryviz$]
	\label{ex:ambiguousQV}
We next give a minimum example for when $\queryviz$ becomes ambiguous.
Consider the two different SQL queries~\cref{Fig_ambiguous_QV_1,Fig_ambiguous_QV_2}. 
Following the algorithm given in \cite{Leventidis2020QueryVis}, 
both lead to the same visual representation \cref{Fig_ambiguous_QV}.
In other words, it is not possible to uniquely interpret the diagram in \cref{Fig_ambiguous_QV}.
\end{example}

\begin{figure}[t]
\centering	
\begin{subfigure}[b]{.21\linewidth}		
\begin{lstlisting}
SELECT DISTINCT *
FROM R
WHERE not exists
 (SELECT *
 FROM S
 WHERE not exists
  (SELECT *
  FROM T
  WHERE T.A = R.A
  AND T.B = S.B
  AND not exists
   (SELECT *
   FROM U
   WHERE U.C = S.C
   AND not exists
    (SELECT *
    FROM V
    WHERE V.D = T.D
    AND V.E = U.E))))
\end{lstlisting}
\vspace{-6mm}
\caption{}
\label{Fig_ambiguous_QV_1}
\end{subfigure}	
\hspace{5mm}
\begin{subfigure}[b]{.21\linewidth}		
\begin{lstlisting}
SELECT DISTINCT *
FROM R
WHERE not exists
 (SELECT *
 FROM V
 WHERE not exists
  (SELECT *
  FROM T
  WHERE T.A = R.A
  AND T.D = V.D
  AND not exists
   (SELECT *
   FROM U
   WHERE U.E = V.E
   AND not exists
    (SELECT *
    FROM S
    WHERE S.B = T.B
    AND S.C = U.C))))
\end{lstlisting}
\vspace{-6mm}
\caption{}
\label{Fig_ambiguous_QV_2}
\end{subfigure}	
\hspace{3mm}
\begin{subfigure}[b]{0.4\linewidth}
    \includegraphics[scale=0.35]{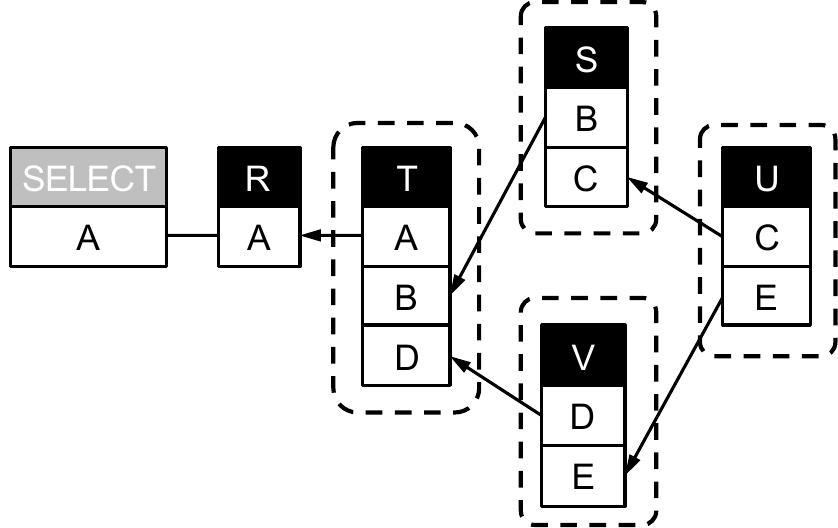}
	\vspace{4mm}
    \caption{}
    \label{Fig_ambiguous_QV}
\end{subfigure}	
\hspace{-1mm}
\caption{
\Cref{ex:ambiguousQV}:
Minimal example showing that $\queryviz$ is not sound for nested queries with 4 levels: 
Two \emph{different queries} (a), (b) that are translated into the same $\queryviz$ diagram (c).}
\label{fig:QV_ambiguousness}
\end{figure}

\introparagraph{Problems from abusing lines in $\queryviz$}
Similar to beta graphs, we think, 
that the design of $\queryviz$
abuses the line symbol by using it for two purposes: 
for ($i$) joining atoms and 
($ii$) for \emph{representing the negation hierarchy}.
In contrast, \diagrams\ use the line only for connecting two attributes
and represent the negation hierarchy explicitly by nesting negation boxes.
\diagrams\ fix the completeness and soundness issues, 
and, in addition, can show logical sentences and queries or sentences lacking tables in one or more of the negation scopes of nested queries.

\subsection{Query-By-Example (QBE)}
\label{sec:QBE}

The development of QBE~\cite{DBLP:journals/ibmsj/Zloof77} was strongly influenced by $\DRC$.
However, QBE can express relational division only by using COUNT
or by 
breaking the query into two logical steps and using a temporary 
relation~\cite[online appendix]{cowbook:2002}.
But in doing so, QBE
uses the query pattern from $\RA$ and $\DatalogN$ of implementing relational division (or universal quantification)
in a dataflow-type, sequential manner, 
requiring two occurrences of the Sailor table.

\begin{figure}[t]
\centering
    \begin{subfigure}[b]{1\linewidth}
		\hspace{30mm}
        \includegraphics[scale=0.35]{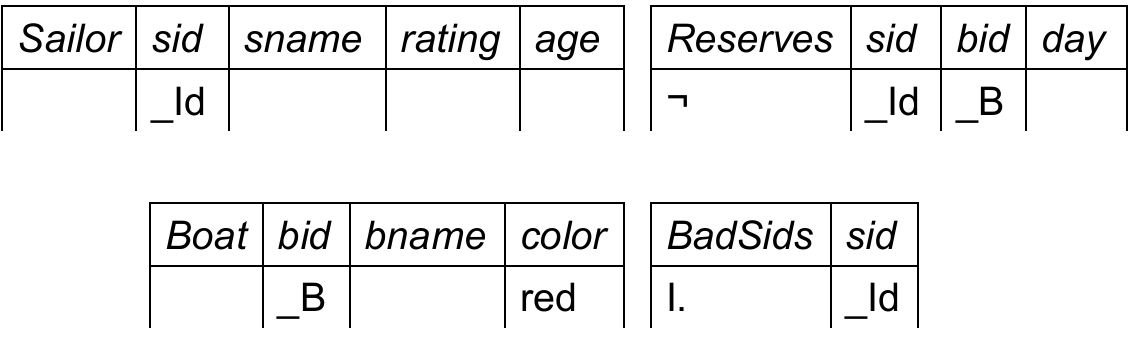}
        \caption{Temporary relation \sql{BadSids(sid)} (``I.'' stands for insert)}
        \label{Fig_QBE_red_boats_a}
    \end{subfigure}	
    \begin{subfigure}[b]{1\linewidth}
		\hspace{30mm}		
        \includegraphics[scale=0.35]{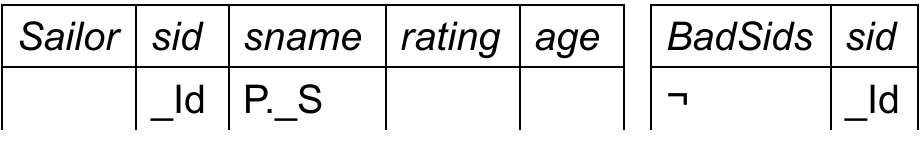}
        \caption{Actual answer \sql{Q(sid)} (``P.'' stands for print)}
        \label{Fig_QBE_red_boats_b}
    \end{subfigure}	
\caption{\Cref{QBE_redboats}: QBE needs to create a temporary relation \sql{BadSids} in order to express relational division.
It does follow the query pattern of relational algebra and not relational calculus.
}
\label{Fig_QBE_red_boats}
\vspace{-3mm}
\end{figure}

\begin{example}[Sailors reserving all red boats in QBE]
\label{QBE_redboats}
	Consider the query 
	``\emph{Find sailors who have reserved all red boats}'' 
	QBE needs to first create an intermediate relation 
	``\sql{BadSids}''
	that stores all Sailors for whom there is a red boat that is not reserved by the sailor
	(\cref{Fig_QBE_red_boats_a}),
	and then finds all the other Sailors
	(\cref{Fig_QBE_red_boats_b}).
	The pattern of this query in QBE thus matches exactly the one of $\DatalogN$ 
	(it requires two occurrences of the relational \sql{Sailor} instead of one as in calculus),
	which is arguably a more dataflow (one relation accessed after the other) than logical query language pattern:
	\begin{align}
	\begin{aligned}
		I_1(x,y)	& \datarule \sql{Reserves}(x,y,\_). 		 \\
		I_2(x)	& \datarule \sql{Sailor}(x,\_,\_,\_), \sql{Boat}(y,\_,\sql{'red'}), \neg I_1(x,y). 		 \\
		Q(y)	& \datarule \sql{Sailor}(x,y,\_,\_), \neg I_2(x). 
	\end{aligned}		
	\label{eq:QBE_equivalent_Datalog}
	\end{align}	
	\end{example}	

More formally, the QBE query from \cref{Fig_QBE_red_boats} is pattern-isomorphic to the $\DatalogN$ query in
\cref{eq:QBE_equivalent_Datalog}.
Furthermore, the following logically-equivalent $\TRC$ query
\cref{q:trc:sailorsallredboats}
has no pattern-isomorphic representation in QBE:
(i.e.\ with one single occurrence of the \sql{Sailor} relation).
\begin{align}
\begin{aligned}
\!\!\!\!\!\!\{q(\sql{sname}) \mid 
	& \exists s \in \sql{Sailor} [q.\sql{sname}=s.\sql{sname} \ \wedge \\
	&\neg (\exists b \in \sql{Boat} [b.\sql{color} = \sql{`red'} \wedge \\
	&\neg (\exists r \in \sql{Reserves} [r.\sql{bid} = b.\sql{bid} \wedge r.\sql{sid} = s.\sql{sid}])]) \}
\end{aligned}
\label{q:trc:sailorsallredboats}
\end{align}

\subsection{DFQL (\texorpdfstring{\cref{sec:DFQL}}{Section \ref{sec:DFQL}})}
\label{appendix:DFQL}

DFQL (Dataflow Query Language) is an example visual representation that is relationally complete \cite{DBLP:journals/iam/ClarkW94,DBLP:journals/vlc/CatarciCLB97,Girsang:DFQL}
by mapping its visual symbols to the operators of relational algebra.
Aside from providing a basic set of operators derived from the requirements for being as expressive as first-order predicate calculus, 
DFQL also provides a diagrammatic representation of grouping operators in both comparison functions and aggregations.
In contrast to several other similarly-motivated visual query languages that represent operators of relational algebra with different visual symbols,
a detailed Master's thesis \cite{Girsang:DFQL} lists the language and its constructs in enough detail to allow us to 
create visualizations of new queries in DFQL.
Following the same procedurality as $\RA$, DFQL expresses the dataflow in a top-down tree-like structure. 
However, since DFQL focuses on the 1-to-1 correspondence to relational algebra, 
it also can not generate a pattern-isomorphic diagram for query
\cref{fig:SQL_equivalence1}
which has no pattern-isomorphic representation in $\RA$. 
See the following example for details:

\begin{example}[Sailors reserving all red boats in DFQL]
\label{DFQL_redboats}
Representing the query ``Find sailors who have reserved all red boats'' 
(recall \cref{ex:sailorsallredboats,Fig_Sailor_Sailors_all_red_boats}) in the formalism of DFQL, 
the entire query can be visualized in one single connected tree-like diagram (unlike QBE which needs to visualize a temporary table to hold the intermediary values).
However, since the language is based on the operators of RA, 
there is no pattern-isomorphic expression of the query (\cref{Fig_Sailor_Sailors_all_red_boats}) in relational algebra.
Instead, the logically-equivalent representation in $\RA$ is as follows:
	\begin{align}
	\begin{aligned}	
		\mathit{Q} =
		\pi_{\sql{sname}} \big(
			&\sql{Sailor}\Join \big(
				\pi_{\sql{sid}} \sql{Sailor} - \pi_{\sql{sid}}\big(
				(\pi_{\sql{sid}} \sql{Sailor} \\
			&\times \pi_{bid}\sigma_{color='red'} \mathit{Boat}) - \pi_{\sql{sid,bid}} \sql{Reserves}\big) \big) \big)
	\end{aligned}
	\end{align}
	The join between \sql{Sailor S} and \sql{Sailor S2} is necessary to project column \sql{sname} from the table.
	This later query can be visualized by DFQL in a pattern-preserving way as \cref{Fig_DFQL_red_boats}. 
	One can easily find a 1-to-1 mapping between DFQL operators and this $\RA$ expression.
	
	Notice that for the same arguments, there is also no pattern-isomorphic expression of the query shown in 
	\cref{Fig_RA_vs_Datalog_g} and DFQL needs two extensional tables for input table $R$ to represent that query.
\end{example}

\begin{figure}[t]
\centering
\hspace*{-10mm}
\includegraphics[scale=0.37]{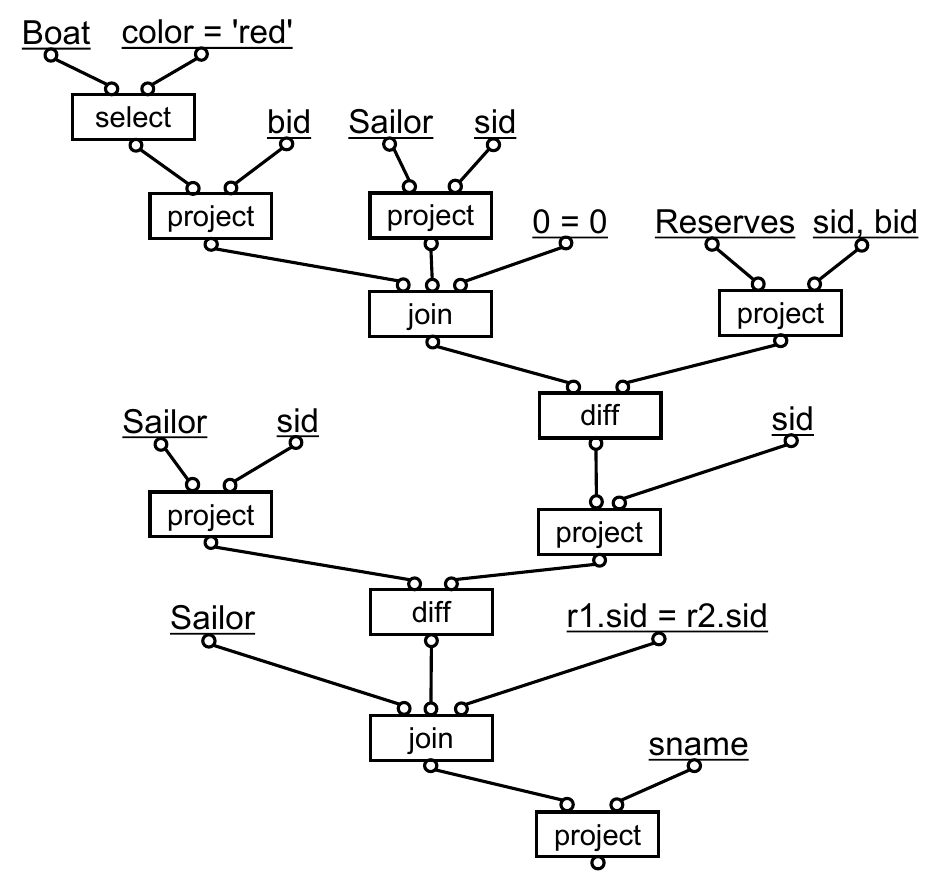}
\caption{\Cref{DFQL_redboats}: DFQL visualization of the query ``Find sailors who have reserved all red boats'' 
that is structurally equivalent to the pattern used by relational algebra. 
The \emph{diff} operator is equivalent to binary \emph{$-$} (minus) in RA and the tautology ``0 = 0'' in \emph{join} operator is required to create a Cartesian Join in DFQL \cite{DBLP:journals/iam/ClarkW94}.
Compare the difficulty in perceiving a logical pattern in this visualization against the one from \diagrams\ in \cref{Fig_Sailor_Sailors_all_red_boats}.
}
\label{Fig_DFQL_red_boats}
\end{figure}

\subsection{Tools for query visualizations}

The four projects that we know of that focus on the problem of visualizing existing relational queries are
QueryVis~\cite{DanaparamitaG2011:QueryViz,Leventidis2020QueryVis,gatterbauer2011databases}
(which we showed is not relationally complete, yet which inspired a lot of our work),
\mbox{GraphSQL}~\cite{DBLP:conf/dexaw/CerulloP07},
Visual SQL~\cite{DBLP:conf/er/JaakkolaT03},
(both of which maintain the 1-to-1 correspondence to SQL,
and syntactic variants of the same query like \cref{fig:SQLvariety} lead to different representations),
and Snowflake join \cite{snowflake} (which is a pure query visualization approach that focuses on join queries with optional grouping, but does not support any nested queries with negation).
Compared to all these visual representations, ours is the only one that is relationally complete and that can preserve and represent all logical patterns in the non-disjunctive fragment of relational query languages.

\subsection{Tools for query plan visualizations}

\begin{figure}[t]
\centering	
\begin{subfigure}[b]{.23\linewidth}		
\begin{lstlisting}
SELECT DISTINCT F.person
FROM Frequents F, 
     Likes L, 
     Serves S
WHERE F.person = L.person
AND F.bar = S.bar
AND L.drink = S.drink
\end{lstlisting}
	\vspace{-3mm}
	\caption{$\SQL$}
	\label{sql:Fig_ExampleExistsCyclic}
\end{subfigure}	
\hspace{8mm}
\begin{subfigure}[b]{.36\linewidth}
    \includegraphics[scale=0.4]{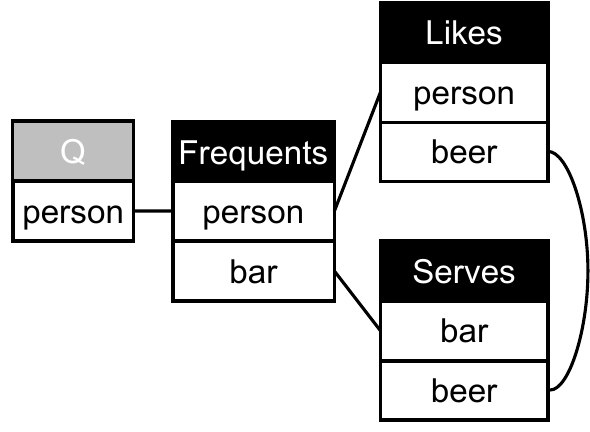}
	\caption{\diagram}
	\label{Fig_ExampleExistsCyclic}
\end{subfigure}	
\begin{subfigure}[b]{.99\linewidth}
\centering
	\vspace{3mm}
	\includegraphics[scale=0.5]{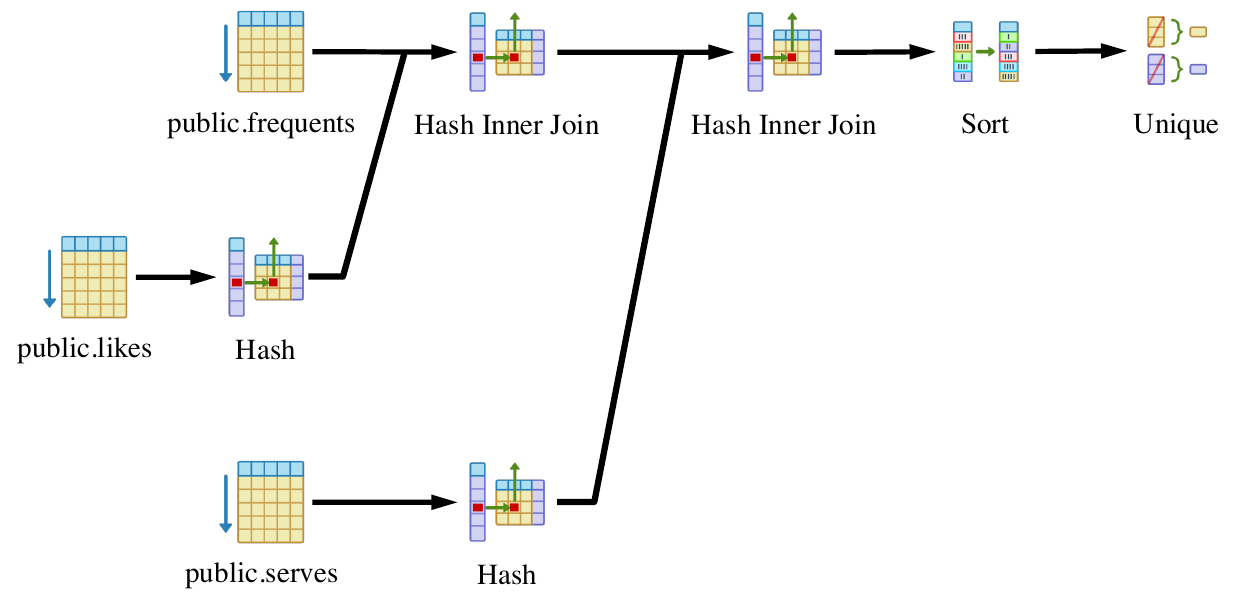}
	\caption{Postgres EXPLAIN}
	\label{Fig_explain_PersonBarDrink_PG}
\end{subfigure}	
\caption{\Cref{ex:cyclicquery}:
Find persons who frequent some bar that serves some drink they like.}
\label{Fig_explain_PersonBarDrink}
\end{figure}

\diagrams\ are a query visualization and thus conceptually different from a \emph{query plan visualization} 
that many $\SQL$ users will be familiar from the EXPLAIN command in PostgreSQL~\cite{postgres}.
We explain with an example used with kind permission from the authors of~\cite{Gatterbauer2022PrinciplesQueryVisualization}.

\begin{example}[Cyclic query]
\label{ex:cyclicquery}
Consider the following cyclic query 
over the beer drinkers database introduced by Ullman~\cite{Ullman1988PrinceplesOfDatabase}:
\begin{align*}
\begin{aligned}
	&\{q.\mathit{person} \mid 
			\exists f \in \mathit{Frequents}, 
			\exists l \in \mathit{Likes},
			\exists s \in \mathit{Serves}	\\
	&	[q.\mathit{person}=f.\mathit{person} \wedge
		f.\mathit{person}=l.\mathit{person} \, \wedge \\
	&	l.\mathit{drink} = s.\mathit{drink,}  \wedge
		s.\mathit{bar} = f.\mathit{bar} ] \}
\end{aligned}
\end{align*}
The query asks for drinkers who frequent bars that serve some beer they like.
\Cref{sql:Fig_ExampleExistsCyclic} shows the query in $\SQL$
and \cref{Fig_ExampleExistsCyclic} as $\diagram$.

\Cref{Fig_explain_PersonBarDrink_PG} shows a query plan chosen by PostgreSQL~\cite{postgres}.
Notice that the produced query plans does not captures the cyclic nature of the join of the query and instead shows the query as a tree.
\end{example}

A \emph{query plan visualization} targets the physical query execution and
represents HOW a query is executed and often helps reason the user about ways to make the query run faster.
In contrast, a query visualization, 
such as $\diagrams$,
attempts to represent WHAT a query does (i.e.\ its intent) and possibly the relational pattern it uses.
Contrast the query visualization in 
\cref{Fig_ExampleExistsCyclic},
which shows the join pattern and that this query is cyclic.

Similarly, query visualizations are also different from \emph{query dashboards},
i.e.\ tools that help monitor key characteristics of queries (such as Vertica analyzer~\cite{DBLP:conf/sigmod/SimitsisWBW14})
or visualize and compare the cost or speed of execution plans (such as Picasso~\cite{DBLP:journals/pvldb/Haritsa10}).

\subsection{Applications for query interpretation}

\emph{Query Interpretation} is the problem of reading and understanding an existing query. 
It is often as hard as {query composition}, i.e., creating a new query~\cite{DBLP:journals/csur/Reisner81}. 
In the past, several projects have focused on building Query Management Systems 
that help users issue queries by leveraging an existing log of queries. Known systems to date include CQMS~\cite{DBLP:conf/cidr/KhoussainovaBGKS09,KhoussainovaKBS:2011}, SQL QuerIE~\cite{DBLP:conf/ssdbm/ChatzopoulouEP09,QueRIERecommendations:2010}, DBease~\cite{LiFWWF2011:DBease}, and SQLShare~\cite{HoweC2010:SQLshare}.
All of those are motivated by making SQL composition easier and thus databases more usable~\cite{DBLP:conf/sigmod/JagadishCEJLNY07}, especially for non-sophisticated database users.
An essential ingredient of such systems is a \emph{query browse} facility, i.e., 
a way
that allows the user to browse and quickly choose between several queries proposed by the system. 
This, in turn, requires a user to \emph{quickly understand} existing queries.

Whereas visual systems for \emph{specifying} queries have been studied
extensively
(a 1997 survey by Catarci et al.~\cite{DBLP:journals/vlc/CatarciCLB97} 
cites over 150 references), 
the explicit reverse problem
of visualizing and thereby helping \emph{interpret a relational query that has already been written}
has not drawn 
much attention, despite 
very early \cite{Reisner1975:HumanFactors,DBLP:journals/csur/Reisner81}
and very recent work
\cite{Leventidis2020QueryVis}
repeatedly showing
that visualizations of relational queries can help users understand them faster than SQL text.

\end{document}